\documentclass[11pt]{article}

\usepackage{amsmath, amssymb, amsthm, amsfonts}
\usepackage{tcolorbox}
\usepackage[utf8]{inputenc}
\usepackage{accents}
\usepackage{setspace}
\usepackage{hyperref}

\usepackage{appendix}
\usepackage{graphicx}
\usepackage{multirow}
\usepackage{ragged2e}
\usepackage{algorithm}
\usepackage{algpseudocode}
\usepackage{paralist}

\usepackage{framed}
\usepackage{adjustbox}

\usepackage{array}
\renewcommand{\arraystretch}{1.5}
\newcolumntype{P}[1]{>{\centering\arraybackslash}p{#1}} 
\newcolumntype{M}{>{\centering\arraybackslash}m} 

\usepackage[framemethod=tikz]{mdframed}
\usepackage[font=small]{caption}
\usepackage{xspace}
\usepackage[letterpaper,hmargin=1.0in,vmargin=1.0in]{geometry}
\usepackage[numbers]{natbib}
\newtheorem{theorem}{Theorem}[section]

\newtheorem{lemma}[theorem]{Lemma}
\newtheorem{definition}[theorem]{Definition}
\newtheorem{observation}[theorem]{Observation}
\newtheorem{claim}[theorem]{Claim}
\newtheorem{fact}[theorem]{Fact}

\definecolor{darkgreen}{rgb}{0,0.5,0}
\algnewcommand\algorithmicswitch{\textbf{switch}}
\algnewcommand\algorithmiccase{\textbf{case}}

\algdef{SE}[SWITCH]{Switch}{EndSwitch}[1]{\algorithmicswitch\ #1\ \algorithmicdo}{\algorithmicend\ \algorithmicswitch}%
\algdef{SE}[CASE]{Case}{EndCase}[1]{\algorithmiccase\ #1}{\algorithmicend\ \algorithmiccase}%
\algtext*{EndSwitch}%
\algtext*{EndCase}%

\newcommand{\eps}{\varepsilon}

\DeclareMathOperator*{\poly}{poly}

\newcommand{\set}[1]{\left\{#1\right\}}

\def\Vor{\mbox{\sf Vor}}

\def\YES{\mbox{\tt YES}}
\def\NO{\mbox{\tt NO}}
\def\ID{\mbox{\sf ID}}
\def\Pr{\mbox{\textup{Pr}}}
\newcommand{\Vn}[1]{{\sf #1}}

\newcommand{\SuperDeg}{\Delta_\textup{super}}

\newcommand{\LowDeg}{\Delta_\textup{low}}

\newcommand{\MedDeg}{\Delta_\textup{med}}

\newcommand{\CrowdedVer}{V_\textup{crwd}}
\newcommand{\DesertedVer}{V_\textup{dsrt}}

\newcommand{\SuperDegEdges}{E_\textup{super}}
\newcommand{\HighDegEdges}{E_\textup{high}}
\newcommand{\MedDegEdges}{E_\textup{med}}
\newcommand{\LowDegEdges}{E_\textup{low}}

\newcommand{\RepEdges}{E_\textup{rep}}
\newcommand{\CellEdges}{E_\textup{bckt}}

\newcommand{\rep}{\textup{rep}}
\newcommand{\cell}{\textup{bckt}}
\newcommand{\low}{\textup{low}}

\newcommand{\high}{\textup{high}}
\newcommand{\super}{\textup{super}}

\newcommand{\SuperDegSpanner}{H_\textup{super}}
\newcommand{\HighDegSpanner}{H_\textup{high}}

\newcommand{\LowDegSpanner}{H_\textup{low}}

\newcommand{\RepSpanner}{H_\textup{rep}}
\newcommand{\CellSpanner}{H_\textup{bckt}}

\newcommand{\Gdense}{G_\textup{dense}}
\newcommand{\Hdense}{H_\textup{dense}}
\newcommand{\Gsparse}{G_\textup{sparse}}
\newcommand{\Hsparse}{H_\textup{sparse}}
\newcommand{\Edense}{E_\textup{dense}}
\newcommand{\Esparse}{E_\textup{sparse}}
\newcommand{\Vdense}{V_\textup{dense}}
\newcommand{\Vsparse}{V_\textup{sparse}}
\newcommand{\HdenseI}{H_\textup{dense}^\textrm{(I)}}
\newcommand{\HdenseB}{H_\textup{dense}^\textrm{(B)}}

\newcommand{\func}[1]{\textsc{#1}}
\newcommand{\neighborP}{\func{Neighbor}}
\newcommand{\degreeP}{\func{Degree}}
\newcommand{\adjacencyP}{\func{Adjacency}}

\newcommand{\CentersOf}{S}

\newcommand{\ClusterOf}{\mbox{\sf C}}
\newcommand{\CellOf}{\mbox{\sf Bucket}}

\newcommand{\RepsOf}{\mbox{\sf Reps}}
\newcommand{\CentersOfRepsOf}{RS}

\newcommand{\dist}{\mbox{\sf dist}}

\renewcommand{\paragraph}[1]{\vspace{0.15cm}\noindent {\bf #1.}}

\newcommand{\Hcal}{\mathcal{H}}

\date{}

\title{Local Computation Algorithms for Spanners}

\author{
	Merav Parter\thanks{Weizmann IS. {\tt merav.parter@weizmann.ac.il}}
	\and
	Ronitt Rubinfeld \thanks{CSAIL, MIT and TAU. {\tt ronitt@csail.mit.edu}}
	\and
	Ali Vakilian \thanks{CSAIL, MIT. {\tt \{vakilian,anak\}@mit.edu}}
	\and
	Anak Yodpinyanee \footnotemark[3]
}

\begin{document}

\maketitle

\begin{abstract}
A graph spanner is a fundamental graph structure that faithfully 
preserves the pairwise distances in the input graph up to a small multiplicative stretch.
The common objective in the computation of spanners is to achieve the 
best-known existential size-stretch trade-off {\em efficiently}.

Classical models and algorithmic analysis of graph spanners 
essentially assume that the algorithm can read the input graph, 
construct the desired spanner, and write the answer to the output tape. 
However, when considering massive graphs containing 
millions or even billions of nodes not only
the input graph, but also the output spanner might be too 
large for a single processor to store.

To tackle this challenge, we initiate the study of {\em local computation algorithms (LCAs)} for graph spanners in general graphs, 
where the algorithm should locally decide whether a given edge $(u,v) \in E$ belongs to the output (sparse) spanner or not. 
Such LCAs give the user the ``illusion'' that
a \emph{specific} sparse spanner for the graph is maintained, 
without ever fully computing it.
We present several results for this setting, including:
\begin{itemize}
\item[$\bullet$] For general $n$-vertex graphs and for parameter $r \in \{2,3\}$, there exists an LCA
	for $(2r-1)$-spanners with $\widetilde{O}(n^{1+1/r})$ edges
	and sublinear probe complexity of $\widetilde{O}(n^{1-1/2r})$. These size/stretch trade-offs are best possible (up to polylogarithmic factors).
\item[$\bullet$] For every $k \geq 1$ and $n$-vertex graph with maximum degree $\Delta$, 
	there exists an LCA for  $O(k^2)$ spanners with $\widetilde{O}(n^{1+1/k})$ edges, probe complexity of $\widetilde{O}(\Delta^4 n^{2/3})$, and random seed of size $\mathrm{polylog}(n)$. This improves upon, and extends the work of [Lenzen-Levi, ICALP'18].
\end{itemize}

We also complement these constructions by providing a polynomial lower bound on the probe complexity of LCAs for graph spanners that holds even for the simpler task of computing a sparse connected subgraph with $o(m)$ edges. 

To the best of our knowledge, our results on 3 and 5-spanners are the first LCAs with sublinear (in $\Delta$) probe-complexity for $\Delta = n^{\Omega(1)}$.
\end{abstract}

\section{Introduction}
One of the fundamental structural problems in graph theory is to find a {\em sparse} structure which preserves the pairwise distances of vertices. In many applications, it is crucial for the sparse structure to be a {\em subgraph} of the input graph; this problem is called the {\em spanner} problem.
For an input graph $G=(V,E)$, a {\em $k$-spanner} $H \subseteq G$  (for $k \geq 1$)
satisfies that for any $v,u\in V$, the distance from $v$ to $u$ in $H$ is at most $k$
times the distance from $v$ to $u$ in $G$, where $k$ is referred to as the
{\em stretch} of the spanner. Furthermore, to reduce the {\em cost} of the solution, it is desired
to output a minimum size/weight such subgraph $H$. 
The notion of spanners was introduced by~\citet{peleg1989graph} and has 
been used widely in different applications such as routing schemes~\cite{awerbuch1992routing,Peleg:2000}, synchronizers~\cite{peleg1989optimal,awerbuch1990network}, SDD's~\cite{spielman2011spectral} and spectral sparsifiers~\cite{kapralov2012spectral}.

It is folklore that for every $n$-vertex graph $G$, 
there exists a $(2k-1)$-spanner $H \subseteq G$ with $O(n^{1+1/k})$ edges.
In particular, if the {\em girth conjecture} of \citet{erdos1965some} is true, then this size-stretch trade-off is optimal.
Spanners have been considered in many different models such as distributed algorithms~\cite{derbel2006fast,derbel2007deterministic,baswana07,DerbelGPV08, derbel2009local,pettie2010distributed,elkin2017efficient} 
and dynamic algorithms~\cite{elkin2011streaming,baswana2012fully,BodwinK16,bernstein2019deamortization}. 

\paragraph{Local computation of small stretch spanners}
When the graph is so large that it does not fit into the main memory, the existing
algorithms are not sufficient for computing a spanner. Instead, we aim at designing an algorithm that answers \emph{queries} of the form ``is the edge $(u,v)$ in the spanner?''~without computing the whole solution upfront.
One way to get around this issue is to consider the {\em Local Computation Algorithms (LCAs)} model (also known as the \emph{Centralized Local model}), introduced by \citet{RubinfeldTVX11} and \citet{alon2012space}. There can be many different plausible $k$-spanners; however, the goal of LCAs for the $k$-spanner problem is to design an algorithm that, given access to {\em primitive} probes (i.e. \neighborP, \degreeP~and \adjacencyP~probes) on the input graph $G$, for each {\em query} on an edge $e \in E(G)$ {\em consistently} with respect to a {\em unique} $k$-spanner $H \subseteq G$ (picked by the LCA arbitrarily), outputs whether $e \in H$. The performance of the LCA is measured based on the {\em quality of solution} (i.e. number of edges in $H$) and the {\em probe complexity} (the maximum number of probes per each query) of the algorithm\footnote{We may also measure the \emph{time complexity} of an LCA. In our LCAs, the time complexities are clearly only a factor of $\poly(\log n)$ higher than the corresponding probe complexities, so we focus our analysis on probe complexities.}. In other words, an LCA gives us the ``illusion'' as if we have query access to
a precomputed {\em $k$-spanner} of $G$.

The study of LCAs with {\em sublinear} probe complexity for {\em nearly linear} size spanning subgraphs (or {\em sparsifiers}) is initiated by \citet{LRR14,levi2016local} for some restricted families of graphs such as minor-closed families.
However, their focus is mainly on designing LCAs that preserve the {\em connectivity} while allowing the stretch factor to be as large as $n$. 
Moreover, in their work, the input graph is sparse (has $O(n)$ edges), while the classical $k$-spanner problem becomes relevant only when the input graph is dense (with superlinear number of edges). Recently, \citet{LeviLenzen17}  designed the first sparsifier LCA in {\em general graphs} with $(1+\eps)n$ edges, stretch $O(\log^2 n\cdot \poly(\Delta/\eps))$ and probe complexity of $O(\poly(\Delta/\eps) \cdot n^{2/3})$, where $\Delta$ is the maximum degree of the input graph.

In this work, we show that sublinear time LCAs for spanners are indeed
possible in several cases. We give: (I) $3$ and $5$-spanners for general graphs with optimal trade-offs between the number of edges and the stretch parameter (up to polylogarithmic factors), and (II) general $k$-spanners, either in the dense regime (when the minimum degree is at least $n^{1/2-1/(2k)}$) or in the sparse regime (when the maximum degree is $n^{1/12-\eps}$).

\paragraph{Broader scope and agenda: local computation algorithms for dense graphs} 
LCAs have been established by now for a large collection of problems, including Maximal Independent Set, Maximum Matching, and Vertex Cover \cite{RubinfeldTVX11,alon2012space,mansour2013local,even2014deterministic,mansour2015constant,reingold2016new, LeviRY17}. These algorithms typically suffer from a probe complexity that is exponential in $\Delta$ and thus are efficient only in the sparse regime when $\Delta=O(1)$. 

To this end, obtaining LCAs even with a polynomial dependency in $\Delta$ is a major open problem for many classical local graph problems, as noted in \cite{mansour2012converting,LeviRY17,ghaffari2018sparsifying}. For instance, recently \citet{ghaffari2018sparsifying} obtained an LCA for the MIS problem with probe complexity of $\Delta^{O(\log\log \Delta)}\cdot \log n$ improving upon a long line of results. Their result also illustrates the connection between LCAs with good dependency on $\Delta$, and algorithms for the massively parallel computation model with sublinear space per machine. Recently, \cite{LeviRY17} and \cite{LeviLenzen17} provided LCAs with probe complexities \emph{polynomial} in $\Delta$ for the problems of $(1-\eps)$-maximum matching and sparse connected subgraphs, respectively. Note that in the context of spanners, such algorithms are still inefficient when the maximum degree is polynomial in $n$, which is precisely the setting where graph sparsification is applied. 

\subsection{Additional related work: spanners in many other related settings}\label{sec:morerelated}
\paragraph{Local distributed algorithms} 
The construction of spanners in the distributed local model, where messages are unbounded, has been studied extensively
in both the randomized and the deterministic settings \cite{baswana07,elkin2017efficient,derbel2006fast,derbel2007deterministic,DerbelGPV08, derbel2009local,pettie2010distributed}: the state of the art of both randomized and deterministic constructions is $O(k)$ rounds.

\paragraph{Dynamic algorithms for graph spanners} In the dynamic setting, the goal is to maintain a spanner
in a setting where edges keep on being inserted or deleted. The main complexity measure is the \emph{update time} which is the computation time needed to fix the current spanner upon a single edge insertion or deletion. 
Most of the dynamic algorithms for spanners maintain an auxiliary clustering structure that aids this modification of current spanner.  \cite{BodwinK16} provided the first dynamic algorithms with sublinear worst-case update time for 
$3$-spanners and $5$-spanners. Recently, \cite{bernstein2019deamortization} showed a general deamortization technique that provides worst-case update time of $\widetilde{O}(1)$ with high probability, for any fixed stretch value of $k$. It would be interesting to see if those recent tools can be useful in the local centralized setting as well. Note that in the LCA setting there is a polynomial lower bound even without the stretch constraint, thus our setting is provably harder.

\paragraph{Streaming algorithms} 
In the setting of dynamic streaming, the input graph is presented online as a long stream of insertions and deletions to its edges. For spanners, the goal is to maintain a sparse spanner for the graph using small space and few passes over the stream. \citet{ahn2012graph} showed the first a sketch-based algorithm for spanners in this setting, yielding $(k^{\log_2 5}-1)$-spanner with $\widetilde{O}(n^{1+1/k})$ edges and $O(\log k)$ passes. \citet{kapralov2014spanners} showed an alternative tradeoff yielding $O(2^k)$-spanner with $\widetilde{O}(n^{1+1/k})$ edges using only \emph{two} passes. In {\em dynamic} streaming, one can keep the entire solution and the challenge is to update the solution though the pass over the stream. In contrast, in the LCA model, one cannot afford keeping the entire solution (i.e., already the number of vertices is too large) but the input graph remains as is.

\subsection{Our results and techniques}
\newcounter{ideacounter}

In this paper we initiate the study of LCAs for graph spanners in \emph{general graphs} which concerns with the following task:
\emph{How can we decide quickly (e.g., sublinear in $n$ time) if a given edge $e$ belongs to a sparse spanner (with fixed stretch) of the input graph, without preprocessing and storing any auxiliary information?} 
In the design of LCAs for graph problems, 
the set of defined {\em probes} to the input graph plays an important role. 
Here we consider the following common probes: 
{\neighborP~probes} (``what is the $i^\textrm{th}$ neighbor of $u$''?), {\degreeP~probes} (``what is $\deg(u)$?'') and {\adjacencyP~probes} (``are $u$ and $v$ neighbors''?) \cite{goldreich2011brief,goldreich1998property}. We emphasize that the answer to an \adjacencyP~probe on an ordered pair $\langle u, v \rangle$ is the index of $v$ in $\Gamma(u)$ if\footnote{$\Gamma(u)$ denotes the neighbor set of $u$, whereas $\Gamma^+(u) = \Gamma(u) \cup \{u\}$.} the edge exists and $\bot$ otherwise.  
Note that if the maximum degree in the input graph is $O(1)$, each \adjacencyP~probe can be implemented by $O(1)$ number of \neighborP~probes.

The problem of designing LCAs for spanners is closely related to designing LCAs for {\em sparse connected subgraphs} with $(1+\eps)n$ edges which was first introduced by \cite{LRR14}. With the exception of\cite{LeviLenzen17}, a long line of results for this problem usually concerns special \emph{sparse} graph families, rather than general graphs.
A summary of these results with a comparison to our results is provided in 
Table~\ref{table:results}. 
\begin{table*}[!h]
\centering
\renewcommand{\arraystretch}{1.5} 
\resizebox{\textwidth}{!}{%
\begin{tabular}{|c||c|c|c|c|c|}
\hline
\multicolumn{2}{|c|}{\bf Reference} & \shortstack{\bf Graph Family} & \shortstack{\bf \# Edges} & \shortstack{\bf Stretch Factor} & \shortstack{\bf Probe Complexity} \\ \hline\hline
\parbox[t]{2mm}{\multirow{7}{*}{\rotatebox[origin=c]{90}{\bf Prior Works}}} & \multirow{3}{*}{\cite{LRR14}} & Bounded Degree Graphs & $(1+\eps)n$ & $-$ & $\Omega(\sqrt{n})$\\
& {} & Expanders & $(1+\eps)n$ & $-$ & $O(\sqrt{n})$\\
& {} & Subexponential growth & $(1+\eps)n$ & $-$ & $O(\sqrt{n})$\\
\cline{2-6}
& \cite{levi2015quasi} & Minor-free & $(1+\eps)n$ & $\poly(\Delta,1/\eps)$ & $\poly(\Delta,1/\eps)$\\ \cline{2-6}
& \cite{levi2016local} & Minor-free & $(1+\eps)n$ & $O((\log \Delta)/\eps)$ & $\poly(\Delta,1/\eps)$\\ \cline{2-6}
& \cite{levi2017constructing} & Expansion $({1/ \log n})^{1+o(1)}$ & $(1+\eps)n$ & super-exponential in ${1/\eps}$ & super-exponential in ${1/\eps}$ \\ \cline{2-6}

& {\cite{LeviLenzen17}} & General & $(1+\eps)n$ & $O(\log^2 n \cdot \poly(\Delta/\eps))$ & $O(n^{2/3}\cdot\poly(\Delta/\eps))$\\\hline

\parbox[t]{2mm}{\multirow{4}{*}{\rotatebox[origin=c]{90}{\bf This Work}}} 
& Theorem~\ref{thm:35-main}& General  & $\widetilde{O}(n^{1+1/r})$ & $2r-1~(r\in\{2,3\})$ & $\widetilde{O}(n^{1-1/(2r)})$\\
& Theorem~\ref{thm:fivespannerk} & Min degree $O(n^{1/2-1/(2k)})$ & $\widetilde{O}(n^{1+1/k})$ & $5$ & $\widetilde{O}(n^{1-1/(2k)})$\\ 
& Theorem~\ref{thm:lowdegree} & Max degree $O(n^{1/12-\eps})$ & $\widetilde{O}(n^{1+1/k})$ & $O(k^2)$ & $\widetilde{O}(n^{1-4\epsilon})$\\
& Theorem~\ref{them:lowrbound} & General & $o(m)$ & any $k\leq n$ & $\Omega(\min\{\sqrt{n}, n^2/m\})$ \\ \hline
\end{tabular}
}
\caption[Results for graph spanners in the LCA model]{Table of results on LCAs for the spanner problem. The symbol $'-'$ indicates that the stretch is not analyzed. The input graph is a simple graph with $n$ vertices, $m$ edges, maximum degree $\Delta$, and belongs to the indicated graph family. $\widetilde{O}$ hides a factor of $\poly(\log n, k)$.}
\label{table:results}
\end{table*}

\subsubsection{LCAs for $3$ and $5$-Spanners for General Graphs}
Our first contribution is the local construction of $3$ and $5$-spanners for general graphs, while achieving the optimal trade-offs between the number of edges and the stretch factors (up to polylogarithmic factors)\footnote{Indeed, the girth conjecture of Erd\H{o}s is resolved for these stretch factors; see e.g.,\cite{wenger1991extremal}.}. In particular, our LCAs have $n^{1-\eps}$ probe complexity even when the input graph is dense with $\Delta=\Omega(n)$; note that in such a case, given a query edge $(u,v)$, the LCA should return yes or no without being able to inspect the neighbor lists $\Gamma(u)$ and $\Gamma(v)$.
In what follows we show how to manipulate the common distributed construction by \citet{baswana07} to yield LCAs for $3$-spanners and $5$-spanners with sublinear probe complexity.

\paragraph{The common distributed approach} \label{par:common3}
Most distributed spanner constructions are based on thinning the graph via clustering:
construct a random set $S$ of {\em centers} by adding each vertex to $S$ independently with some fixed probability. For each vertex $v$ sufficiently close to a center in $S$, include the edges of the shortest path connecting $v$ to its closest member $s \in S$: this induces a \emph{cluster} around each center $s\in S$, where every pair of vertices in the same cluster are connected by a short path. Then, add edges connecting pairs of neighboring clusters to ensure the desired stretch factor.

The following algorithm constructs a $3$-spanner $H \subseteq G$ with $\widetilde{O}(n^{3/2})$ edges. 
First, add to $H$ all edges incident to vertices of degree at most $\sqrt{n}$. 
Second, pick a collection $S$ of centers by sampling each vertex independently with probability $\Theta(\log n/\sqrt{n})$.
Each vertex $v$ of degree at least $\sqrt{n}$ picks a \emph{single} neighboring center $s \in S\cap \Gamma(v)$ (which exists w.h.p.) as its \emph{center}, then adds $(v,s)$ to $H$, forming a collection of $|S|= O(\sqrt{n})$ clusters (stars) around these centers. Lastly, every vertex $u$ adds only one edge to each of its neighboring clusters -- note that this last step may add edges whose endpoints are both non-centers. This results in a 3-spanner: For omitted edge $(u,v)$ in $G$, if $u$ and $v$ are in the same cluster, then they have a path of length $2$ through their shared center $s$. If $u$ and $v$ are in different clusters, an edge from $u$ to some other vertex $w$ in $v$'s cluster would have been chosen, providing the path $\langle u, w, s, v\rangle$ of desired stretch $3$ connecting $u$ and $v$, where $s$ is $v$'s center.

\paragraph{The challenge and key ideas} 
Recall that our goal is to design an LCA for $3$-spanners $H \subseteq G$ of size $\widetilde{O}(n^{3/2})$ and probe complexity of $\widetilde{O}(n^{3/4})$: the LCA is given an edge $(u,v)$ and must answer whether $(u,v) \in E(H)$. 
First, if $\deg(u)$ or $\deg(v)$ is at most $\sqrt{n}$, then the algorithm can immediately say $\YES$. This requires only two \degreeP~probes for the endpoints $u,v$. Hence, the interesting case is where both $u$ and $v$ have degrees at least $\sqrt{n}$. 

We start by sampling each vertex into the center set $S$ with probability of $p=\Theta(\log n/\sqrt{n})$, thus w.h.p.~guaranteeing that each high-degree vertex has at least one sampled neighbor. For clarity of explanation, assume that given the ID of a vertex $v$, the LCA algorithm can decide (with no further probes) whether $v$ is sampled. Upon selecting the set of centers $S$, the above mentioned distributed algorithm has two degrees of freedom (which our LCA algorithm will enjoy).
First, for a high-degree vertex $v$, there could be potentially many sampled neighbors in $S$: the distributed algorithm lets $v$ join the cluster of an arbitrarily sampled neighbor. The second degree of freedom is in connecting a high-degree vertex to neighboring clusters. In the distributed algorithm, a vertex connects to an arbitrarily chosen neighbor in each of its neighboring clusters. Since the answers of the LCA algorithm should be consistent, it is important to carefully fix these decisions to allow small probe complexity. 

\paragraph{The na\"{i}ve approach for 3-spanners and its shortcoming}
The most na\"{i}ve approach is as follows: for each $v$, traverse the list $\Gamma(v)$ in a fixed order and pick the \emph{first} neighbor that satisfies the required conditions. That is, 
a vertex joins the cluster of its \emph{first} sampled neighbor (center) and connects to its \emph{first} representative neighbor in each of its neighboring clusters. To analyze the probe complexity of such a construction, consider a query edge $(u,v)$ where $\deg(u),\deg(v)\geq \sqrt{n}$. 
By probing for the first $\sqrt{n}$ neighbors of $u$ and $v$, one can compute the cluster centers $c_u$ and $c_v$ of $u$ and $v$ with high probability. The interesting case is where $u$ and $v$ belong to different clusters. In such a case, the LCA algorithm should say $\YES$ only if $v$ is the first neighbor of $u$ that belongs to the cluster of $c_v$. 
To check if this condition holds, the algorithm should probe for each of the neighbors $w$ of $u$ that appears before $v$ in $\Gamma(u)$, and say $\NO$ if there exists such earlier neighbor $w$ that belongs to the cluster of $c_v$. 
Here, it remains to show how this cluster-membership testing procedure is implemented.

A \emph{cluster-membership test}, for a pair $\langle s,w \rangle$ with $s \in S$, must return $\YES$ iff $w$ belongs to the cluster of the center $s$. The above mentioned algorithm thus makes $O(\deg(v))$ 
cluster-membership tests for each $w$ preceding $u$ in $\Gamma(v)$ and $s = c_v$.  
Since each center is sampled with probability $p=\log n/\sqrt{n}$, the probe complexity of a single cluster-membership test is $O(\sqrt{n})$ w.h.p., leading to a total probe complexity of $O(\deg(v)\cdot \sqrt{n})$. 

\stepcounter{ideacounter}\paragraph{Idea (\Roman{ideacounter}) -- Multiple centers for efficient cluster-membership test}
The key idea in our solution is to pick the cluster centers in a way that allows answering each cluster-membership test for a pair $\langle s,w \rangle$ using a single $\neighborP$ probe! Towards this goal, we let each high-degree vertex join \emph{multiple} clusters, instead of just one. In particular, for a vertex $w$, we look at the subset $\Gamma_1(w)$ consisting of its first $\sqrt{n}$ neighbors in $\Gamma(w)$. We then let $w$ join the clusters of all sampled neighbors in $\Gamma_1(w)\cap S$. Since each vertex is a center with probability $p$, this implies that, w.h.p., $w$ joins $\Theta(p\cdot |\Gamma_1(w)|)=\Theta(\log n)$ many clusters. 
Though this approach adds a multiplicative $O(\log n)$ factor to the size of our spanner, it will pay off dramatically in terms of the probe complexity of our LCA.
In particular, this modification enables the algorithm to test cluster-membership with a single {\adjacencyP} probe: the vertex $w$ belongs to the cluster of $s$, if the index of $s$ in $w$'s neighbor-list is at most $\sqrt{n}$ (the index is returned by the \adjacencyP~probe on $u$ and $s$). This idea alone decreases the probe complexity of our LCA to $\widetilde{O}(\deg(w))$.

\stepcounter{ideacounter}\paragraph{Idea (\Roman{ideacounter}) -- Neighborhood partitioning} 
The multiple center technique above allows our LCA to handle edges adjacent to a vertex $u$ of degree at most $n^{3/4}$. For $\deg(u)>n^{3/4}$, our LCA cannot afford to look at all neighbors of $u$. 
To this end, we \emph{partition} the neighbors of $u$ into \emph{blocks} of size $n^{3/4}$ each. 
Rather than adding only one edge between $u$ to each neighboring cluster, 
we make the decision on which edges to keep for each block \emph{independently},
by scanning only the block containing $v$ and keeping $(u,v)$ if $v$ belongs to the cluster
that was not previously seen in this block.
Though this leads to an increase in the number of edges  by a factor of $\deg(u)/n^{3/4} \leq n^{1/4}$, we can now
keep the probe complexity down to $\widetilde{O}(n^{3/4})$ as we only need to scan the block containing $v$ given the query $(u,v)$ instead of $u$'s entire neighbor-list. To keep the size of the spanner small, e.g., $\widetilde{O}(n^{3/2})$, we use the fact that 
$O(n^{1/4} \log n)$ sampled vertices are enough to hit the neighborhoods of all vertices with degree more than $n^{3/4}$ with high probability.
Since for each block of size $n^{3/4}$ in the neighborhood of $u$ the algorithm adds $O(|S|)$ edges, the total number of edges added per vertex is  $O(|S|\cdot \deg(u)/n^{3/4})=\widetilde{O}(n^{3/2})$, as desired. 

\paragraph{Overview of the LCA for $5$-spanners} 
For $5$-spanners, the desired number of edges is $\widetilde{O}(n^{4/3})$. This allows us to immediately add to the spanner all edges incident to low-degree vertices $u$ with $\deg(u)=\widetilde{O}(n^{1/3})$. The common distributed construction for $5$-spanners computes $O(n^{2/3})$ clusters by sampling each center independently with probability $\Theta(\log n/n^{1/3})$. By letting each high-degree vertex (i.e., with $\deg(u)=\Omega(n^{1/3})$) join the cluster of one of its sampled neighbors,  the spanner contains a collection of $O(n^{2/3})$ (vertex-disjoint) clusters that, w.h.p., cover all high-degree vertices. Finally, each pair of neighboring clusters $C_1, C_2$ are connected by adding an edge 
$(u,v) \in (C_1 \times  C_2)\cap E$ to the spanner $H$. 
It is straightforward to verify that $H$ is a $5$-spanner of size $\widetilde{O}(n^{4/3})$. 

Designing LCAs for the $5$-spanner problem turns out to be significantly more challenging than the $3$-spanner case. The reason is that deciding whether an edge $(u,v)$ is in the $5$-spanner requires information from the {\em second neighborhoods} of $v$ and $u$, which is quite cumbersome when one cannot even read the entire neighborhood of a vertex. Our solution extends the $3$-spanner construction in two ways: some of the edges added to our $5$-spanner are between \emph{cluster} pairs, instead of edges between a vertex and a cluster as in the $3$-spanner solution. Another set of edges added to the $5$-spanner is between pairs of vertex and cluster, but unlike the $3$-spanner case, these clusters have now \emph{radius two}. 

\stepcounter{ideacounter}\paragraph{Idea (\Roman{ideacounter}) -- Cluster partitioning (bucketing)}
The standard clustering-based construction of $5$-spanners adds an edge between every pair of neighboring clusters (stars).  
This clustering-based construction cannot be readily implemented with the desired probe complexity. To see why, consider clusters centered at $s$ and $t$, containing $u$ and $v$ respectively.
A na\"{i}ve attempt spends $\deg(s) \cdot \deg(t)$ probes for vertices between these clusters, as to consistently pick a unique edge between the two clusters. 

One of our tools extends the idea of neighborhood partitioning from $3$-spanner into cluster partitioning.
Each of the $O(n^{2/3})$ clusters is partitioned into balanced {\em buckets} of size $\Theta(n^{1/3})$.\footnote{Note that each cluster may have at most one bucket of size $o(n^{1/3})$.} The algorithm then picks only one edge between any pair of neighboring buckets.  Since the number of buckets can be shown to be $\widetilde{O}(n^{2/3})$, the spanner size still remains $\widetilde{O}(n^{4/3})$.
Unlike partitioning \emph{neighbor-lists}, partitioning a \emph{cluster} requires the full knowledge of its members -- which are no longer nicely indexed in a list.
To be able to efficiently partition a clusters, the algorithm allows only vertices with degree at most $n^{5/6}$ to be chosen as cluster centers. The benefit of this restriction is that one can inspect the entire neighborhood of a center in $O(n^{5/6})$. The drawback of this approach is that it only clusters vertices that have sufficiently many neighbors (i.e., at least $n^{1/3}$) with degree less than $n^{5/6}$. The remaining vertices are handled via their high-degree neighbors (i.e., of degree at least $n^{5/6}$) as described next. 

\stepcounter{ideacounter}\paragraph{Idea (\Roman{ideacounter}) -- Representatives}
Using the neighborhood-partitioning idea from $3$-spanner, all vertices with degree at least $n^{5/6}$ can be clustered by sampling $\widetilde{O}(n^{1/6})$ cluster centers. 
By partitioning the neighborhood of each high-degree vertex into disjoint blocks each of size $\widetilde{O}(n^{5/6})$, one can construct a $3$-spanner for all edges incident to these high-degree vertices with probe complexity of $\widetilde{O}(n^{5/6})$ while using $\widetilde{O}(n^{4/3})$ edges. 
To take care of vertices of degrees less than $n^{5/6}$ that have many high-degree neighbors, we let them join the cluster of their high-degree neighbors, hence creating {\em clusters of depth $2$}.

To choose which cluster to join (in the second level), our vertex, which has many high-degree neighbors, simply chooses and connects itself to one or more high-degree neighbors, called its \emph{representatives}. To determine the representatives of a vertex $u$, we simply pick $\Theta(\log n)$ random neighbors of $u$, and w.h.p.~one of them will have high-degree, and hence is chosen as $u$'s representative.

We implement our LCA by first picking $|S| = \widetilde{O}(n^{1/6})$ centers.
Consider the query edge $(u,v)$ where $\deg(u),\deg(v) \geq n^{1/3}$ and $u$ has many high-degree neighbors. Here, $u$ has $\Theta(\log n)$ representatives, each of which has $\Theta(\log n)$ centers in $S$ w.h.p., so $u$ belongs to $O(\log^2 n)$ clusters.
As in the $3$-spanner case, we keep $(u,v)$ if $v$ is the first neighbor of $u$ in the cluster that $v$ belongs to.  We find the representatives of each neighbor of $u$ by making $O(\log n)$ probes, and for all these $\deg(u) \cdot O(\log n) = \widetilde{O}(n^{5/6})$ representatives, check if they belong to any of $v$'s $O(\log^2 n)$ clusters with $\widetilde{O}(n^{5/6})$ total probes.

\begin{theorem}[$3$ and $5$-spanners]\label{thm:35-main}
For every $n$-vertex simple undirected graph $G$, there exists an LCA for $(2r-1)$-spanners with $\widetilde{O}(n^{1+1/r})$ edges and probe complexity $\widetilde{O}(n^{1- 1/(2r)})$ for $r\in\{2,3\}$. Moreover, the algorithm only uses a seed of $O(\log^2 n)$ random bits.
\end{theorem}
In fact, if $G$ has \emph{minimum} degree $\omega(n^{1/3})$, we may apply the $5$-spanner construction (with modified parameters) to obtain $5$-spanners with even smaller number of edges as indicated in Table~\ref{table:results} (Theorem~\ref{thm:fivespannerk}): this minimum degree assumption indeed allows even sparser spanners, bypassing the girth conjecture that holds for \emph{general} graphs.

\subsubsection{LCA for $O(k^2)$-spanners}  
Our second contribution is the local construction of $O(k^2)$-spanners with $O(n^{1+1/k})$ edges for any $k \geq 1$, which has sub-linear probe complexity for graphs of maximum degree $\Delta = O(n^{1/12-\eps})$. 
Our approach improves upon and extends the recent work of \citet{LeviLenzen17}.
The  work of \cite{LeviLenzen17} aims at locally constructing a spanning subgraph with $O(n)$ edges, but the stretch parameter of their subgraph might be as large as $O(\poly(\Delta)\log^2 n)$. In addition, this construction requires a random seed of polynomial size.  
In our construction, we reduce the stretch parameter of the constructed subgraph to $O(k^2)$, independent of both $n$ and $\Delta$, while using only $\widetilde{O}(n^{1+1/k})$ edges. 
In addition, we implement our randomized constructions using $\poly(\log n)$ independent random bits, whereas \cite{LeviLenzen17} uses $\poly(n)$ bits. We remark that for the LCAs with large stretch parameter considered in \cite{LeviLenzen17}, our techniques can still be applied to exponentially reduce the required amount of random bits, and save a factor of $\Delta$ in the probe complexity. 

\begin{theorem}[$O(k^2)$-spanners]\label{thm:lowdegree}
For every integer $k\geq 1$ and every $n$-vertex simple undirected graph $G$ with maximum degree $\Delta$, there exists a (randomized) LCA for $O(k^2)$-spanner with $\widetilde{O}(n^{1+1/k})$ edges and probe complexity $\widetilde{O}(\Delta^4 n^{2/3})$. Moreover, the algorithm only uses $O(\log^2 n)$ random bits.
\end{theorem}
The high level structure is as in~\cite{LeviLenzen17}: for a given stretch parameter $k$, partition the edges in $G$ into the \emph{sparse} set $\Esparse$ and the \emph{dense} set $\Edense$. Roughly speaking, the sparse set $\Esparse$ only consists of edges $(u,v)$ for which the $k$-neighborhood in $G$ of either $u$ or $v$ contains at most $O(n^{2/3})$ vertices.
For this sparse region in the graph, we can simulate a standard distributed algorithm for spanners~\cite{baswana07,Censor-HillelPS16} (using only a poly-logarithmic number of random bits), with small probe complexity. This yields an LCA handling the sparse edges with $O(\Delta^2 n^{2/3})$ probe complexity. 

To take care of the dense edges, we sample a collection of $O(n^{2/3}\log n)$ centers and partition the (dense) vertices into \emph{Voronoi cells} around these centers.%

The main challenge is in connecting the Voronoi cells, keeping in mind that taking an edge between every pair of cells adds too many edges to the spanner. 
To get around it, the main contribution of~\cite{LeviLenzen17} was in designing a set of rules for connecting bounded-size sub-structures in Voronoi cells, called {\em clusters}. The high-level description of the rules are as follows\footnote{Here, we state a simplified version of the rules. In particular, the rules are expressed in terms of clusters whose exact definitions are skipped for now. Refer to the longer version of our paper for the precise definitions of the rules.}: \emph{mark} a random subset of $O(n^{1/3}\log n)$ Voronoi cells (among the $n^{2/3}$ Voronoi cells), then connect\footnote{We \emph{connect} two vertex sets by adding the unique lexicographically-first edge between the two vertex sets (if any exists) based on the vertex $\ID$s of the endpoints.} them according to the following rules using $\widetilde{O}(n)$ edges each. Rule (1): connect every marked Voronoi cells to each of its neighboring Voronoi cells. Rule (2): if a Voronoi cell has no neighboring marked Voronoi cells, then connect it to all its neighboring Voronoi cells as well. 
Rule (3): For each pair of (not necessarily adjacent) Voronoi cell $\Vn{a}$ and marked Voronoi cell $\Vn{c}$ sharing common neighboring Voronoi cells $\Gamma(\Vn{a}) \cap \Gamma(\Vn{c})$, keep an edge from $\Vn{a}$ to a single Voronoi cell $\Vn{b^*} \in \Gamma(\Vn{a}) \cap \Gamma(\Vn{c})$ (i.e., $\Vn{b^*}$ has the minimum $\ID$ in $\Gamma(\Vn{a}) \cap \Gamma(\Vn{c})$). 
This last rule handles the edges of (unmarked) Voronoi cells that have some neighboring marked Voronoi cell. 

\stepcounter{ideacounter}\paragraph{Idea (\Roman{ideacounter}) -- Establishing the $O(k^2)$ stretch guarantee} In our implementation, the radius of each  Voronoi cell is $O(k)$ (as opposed to $O(\Delta \log n)$ in~\cite{LeviLenzen17}). Thus, it suffices to show that the spanner path from Voronoi cell supervertices $\Vn{a}$ to $\Vn{b}$ only visits $O(k)$ other Voronoi cells. To this end, we impose a random ordering of the Voronoi cells, by assigning them distinct random \emph{ranks}. We then make the following modification to Rule (3): add an edge from $\Vn{a}$ to $\Vn{b}$ if there exists a marked Voronoi cell $\Vn{c}$ such that the rank $r(\Vn{b})$ of $\Vn{b}$ is among the $O(n^{1/k}\log n)$ lowest ranks in $\Gamma(\Vn{a}) \cap \Gamma(\Vn{c})$, restricted to those discovered by the LCA. 
This modified rule allows us to extend the inductive connectivity argument of \cite{LeviLenzen17} to show that every pair of adjacent cells are connected by a path that goes through $O(k)$ cells -- since each cell has radius $O(k)$, the final stretch is $O(k^2)$. 

\stepcounter{ideacounter}\paragraph{Idea (\Roman{ideacounter}) -- Graph connectivity with bounded independence}
One of our key technical contributions is in showing that one can implement the above randomized random rank assignment using small number of random bits. We show that the ranks of Voronoi cells can be computed using $T=\Theta(k)$ hash functions $h_1,\cdots, h_{T}$ chosen uniformly at random form a family of $O(\log n)$-wise independent hash functions of the form $\{0,1\}^{\log n} \rightarrow \{0,1\}^{O((\log n) /k)}$.
We define our rank function as a concatenation of $h_i$'s on the $\ID$ of the Voronoi cell's center: for the Voronoi cell centered at $v$, $r(v) = h_1(\ID(v))\circ \ldots \circ h_{T}(\ID(v))$.  
We then carefully adopt the inductive stretch argument to this randomized rank assignment with limited independence so that in the $i^{\textrm{th}}$ step, our analysis only relies on the hash function $h_i$.

\subsubsection{Lower Bounds}
To establish the lower bound, we construct two distributions over undirected $d$-regular graph instances that contain a designated edge $e$. For graphs in the first family, it holds that after removing $e$, w.h.p., they remain connected while in the second family, removing $e$ disconnects the endpoints of $e$ and leave them in separate connected components.
We show that for the edge $e$, any LCA that makes $o(\min\{\sqrt{n},n/d\}) = o(\min\{\sqrt{n},n^2/m\})$ probes can only distinguish whether the underlying graph is from the first family or the second family with probability $1/2+o(1)$. 

Our approach mainly follows from the analysis of \citet{kaufman2004tight}, on the lower bound construction of~\cite{LRR14}. While \cite{kaufman2004tight} studies a rather different problem of bipartiteness testing, we consider similar probe types and obtain a similar lower bound as those of \cite{kaufman2004tight}.
On the other hand, the construction of~\cite{LRR14} shows the probe complexity of $\Omega(\sqrt{n})$ for LCAs for spanning graphs that only use \neighborP~probes, not \adjacencyP~probes.

\begin{theorem}[Lower Bound]
\label{them:lowrbound}
Any local randomized LCA that computes, with success probability at least $2/3$, a spanner of the simple undirected $m$-edge input graph $G$ with $o(m)$ edges, has probe complexity $\Omega(\min\{\sqrt{n}, n^2/m\})$.
\end{theorem}

\subsection{Discussion}
We study LCAs for spanners and provide new tools for dealing with large degrees in the local model. We believe these tools should pave the way toward the design of new LCAs for dense graphs. 
We leave a number of remaining open questions, perhaps the most compelling of which is:  Can we provide for \emph{general} graphs, an LCA for $(2k-1)$ spanners, $\widetilde{O}(n^{1+1/k})$ edges and probe complexity $\widetilde{O}(n^{1-1/(2k)})$? Our tools already solve the problems in the dense regime or in the sparse regime\footnote{Up to having stretch of $O(k^2)$ in the latter. Refer to Table~\ref{table:results} for more details on our results in those regimes.}, but there is still an unknown regime to be explored. 

Aside for the LCA setting, our constructions raise some interesting thoughts regarding the notion of optimally in graph spanners. It is folklore to believe that with a budget of $n^{1+1/k}$ edges for our spanner, the best stretch that one can obtain is $(2k-1)$.  However, a deeper look in the girth conjecture of Erd\H{o}s reveals that this tightness holds only when the degrees are of the edge endpoints are at most $c\cdot n^{1/k}$. If one does not care about constant factors in the spanner size, then we can just pick these \emph{tight} edges into the spanner and have a stretch $1$ for them. We then ask: for a given budget of $n^{1+1/k}$, what is the best stretch that can be obtained for an edge $(u,v)$? As we see in this paper, once the degrees of $u$ or $v$ are high, a stretch much better than $2k-1$ can be provided. It would be interesting to further understand the tradeoff between stretch, spanner size and the density of the input graph.

\subsection{Model Definition and Preliminaries}
\paragraph{Graph notation}
Throughout, we consider simple unweighted undirected graphs $G = (V, E)$ on $n = |V|$ vertices and $m = |E|$ edges. 
Each vertex $v$ is labeled by a unique $O(\log n)$-bit value $\ID(v)$\footnote{We do \emph{not} require $\ID$s to be a bijection $V\rightarrow[n]$ as in other LCA papers.}. 
For $u \in V$, let $\Gamma(u,G) = \{v: (u,v) \in E\}$ be the neighbors of $u$, $\deg(u,G) = |\Gamma(u,G)|$ be its degree, and define $\Gamma^+(u,G) = \Gamma(u,G) \cup \{u\}$. Denote $V_{I} = \{v \in V: \deg(v,G) \in I\}$ where $I$ is an interval. For $u,v \in V$, let $\dist(u,v,G)$ be the shortest-path distance between $u$ and $v$ in $G$. Let $\Gamma^{k}(u,G) = \{v: \dist(u,v,G) \leq k\}$ be the $k^\textrm{th}$-neighborhood of $u$, and denote its size $\deg_k(u,G) = |\Gamma^{k}(u,G)|$. For subsets $V_1,V_2 \subseteq V$, let $E(V_1,V_2)=E \cap (V_1 \times V_2)$.
The parameter $G$ may be omitted for the input graph.

We assume that the input graph has an adjacency list representation: each neighbor set has a fixed ordering, $\Gamma(u)=\{v'_{1}, \ldots, v'_{\deg(u)}\}$; this ordering may be arbitrarily (e.g., not necessarily sorted by vertex $\ID$s). Many of the algorithms in this paper are based on partitioning the neighbor-list into balanced-size blocks. For $\Delta \in [n]$ and $u \in V$ such that $\deg(u) \geq \Delta$, let $\Gamma_{\Delta,1}(u), \ldots, \Gamma_{\Delta,\Theta(\deg(u)/\Delta)}(u)$ be blocks of neighbors obtained by partitioning $\Gamma(u)$ into consecutive parts. Each block is of size $\Delta$, except possibly for the last block that is allowed to contain up to $2\Delta$ vertices.

\paragraph{Local Computation Algorithms}
We adopt the definition of LCAs by Rubinfeld et al.~\cite{RubinfeldTVX11}. A local algorithm has access to the \emph{adjacency list oracle} $\mathcal{O}^G$ which provides answers to the following probes (in a single step):
\begin{compactitem}
\item{} \textbf{\neighborP~probes}: Given a vertex $v \in V$ and an index $i$, the $i^\textrm{th}$ neighbor of $v$ is returned if $i \leq \deg(v)$. Otherwise, $\bot$ is returned. The orderings of neighbor sets are fixed in advance, but can be arbitrary.
\item{} \textbf{\degreeP~probes}: Given a vertex $v \in V$, return $\deg(v)$. This probe type is defined for convenience, and can alternatively be implemented via a binary search using $O(\log n)$ \neighborP~probes.
\item{} \textbf{\adjacencyP~probes}: Given an ordered pair $\langle u,v \rangle$, if $v \in \Gamma(u)$ then the index $i$ such that $v$ is the $i^\textrm{th}$ neighbor of $u$. Otherwise, $\bot$ is returned.
\end{compactitem}
\begin{definition}[LCA for Graph Spanners]
An \emph{LCA} $\mathcal{A}$ for graph {\em spanners} is a (randomized) algorithm
with the following properties. $\mathcal{A}$ has access to the adjacency list oracle $\mathcal{O}^G$ of the
input graph $G$, a tape of random bits, and local read-write computation memory. 
When given an input (query) edge $(u,v) \in E$, $\mathcal{A}$ accesses $\mathcal{O}^G$ by making probes, then returns \YES\ if $(u,v)$ is in the spanner $H$, or returns \NO\ otherwise.
This answer must only depend on the query $(u,v)$, the graph $G$, and the random bits. For a fixed tape of random bits, the answers given by $\mathcal{A}$ to all possible edge queries, must be consistent with one particular sparse spanner.
\end{definition}

The main complexity measures of the LCA for graph spanners are the size and stretch of the output spanner, as well as the probe complexity of the LCA, defined as the maximum number of probes that the algorithm makes on $\mathcal{O}^G$ to return an answer for a single input edge.
Informally speaking, imagine $m$ instances of the same LCA, each of which is given an edge of $G$ as a query, while the \emph{shared} random tape is broadcasted to all. Each instance decides if its query edge is in the subgraph by making probes to $\mathcal{O}^G$ and inspecting the random tape, but may not communicate with one another by any means. The LCA succeeds for the input graph $G$ and the random tape if the collectively-constructed subgraph is a desired spanner.
All the algorithms in this paper are randomized and, for any input graph, succeed with high probability $1-1/n^c$ over the random tape.

\paragraph{Paper Organization}
In Section~\ref{sec:three} and~\ref{sec:5spanner} we describe our results for $3$ and $5$-spanners in general graphs.
For simplicity, we first describe all our randomized algorithms as using full independence, then in Section \ref{APPEND:bounds}, we explain how these algorithms can be implemented using a seed of poly-logarithmic number of random bits). Next, in Section~\ref{sec:fullsparse} we show the LCA of the $O(k^2)$-spanners. Finally, in Section~\ref{sec:lowerbound}, we provide a lower bound result for a
simpler task of computing a spanning subgraph with the specified probes.

\paragraph{Clarification}
Throughout we use the term ``spanner construction'' when describing how to construct our spanners. These construction algorithms are used only to define the unique spanner, based on which the LCA makes its decisions: we never construct the full, global spanner at any point.

\section{LCA for $3$-Spanners}\label{sec:three}
In this section, we present the $3$-spanner LCA with probe complexity of $\widetilde{O}(n^{3/4})$. We begin in Section~\ref{sec:3cat} by establishing some observations that allow us to ``take care'' of different types of edges separately based on the degrees of their endpoints.
In Section~\ref{sec:3high}-\ref{sec:3super} we provide constructions that take care of each type of edges; the analysis of stretch, probe complexity and spanner size for each case is included in their respective sections.
We establish our final LCA for 3-spanners in Section~\ref{sec:3final}.

\subsection{Edge classification}\label{sec:3cat}

\begin{definition}[Subgraphs taking care of edges]\label{def:takecare}
For stretch parameter $k$ and set of edges $E' \subseteq E$, we say that the subgraph $H' \subseteq G$ \emph{takes care} of $E'$ if for every $(u,v) \in E'$, $\dist(u,v,H')\leq k$.
\end{definition}

Observe that if we have a collection of subgraphs $H_i$'s such that every edge in $(u,v)\in E$ is taken care by at least one $H_i$, then the union $H$ of the $H_i$'s constitutes a $k$-spanner for $G$.

\begin{observation}[Spanner construction by combining subgraphs]\label{take-care-lca}
For a collection of subsets $E_1, \ldots, E_\ell\subseteq E$ where $\cup_{i\in[\ell]} E_i = E$, if $H_i$ is a subgraph of $G$ that takes care of $E_i$, then $H = \cup_{i\in[\ell]} H_i$ is a $k$-spanner of $G$. Further, if we have an LCA $\mathcal{A}_i$ for computing each $H_i$ (i.e., deciding whether the query edge $(u,v) \in H_i$ and reporting \YES~or \NO~accordingly), we may construct a final LCA that runs every $\mathcal{A}_i$ and answer \YES~precisely when at least one of them does so. The performance of our overall LCA (number of edges, probes, or random bits) can then be bounded by the respective sum over that of $\mathcal{A}_i$'s.
\end{observation}
Note that $H_i$ may contain edges of $E$ that are not in $E_i$, thus it is necessary that the overall LCA invokes every $\mathcal{A}_i$ even if $\mathcal{A}_i$ does not take care of the query edge.

\paragraph{Graph partitioning}
A vertex $v$ is low-degree if $\deg(v)\leq \sqrt{n}$, it is high-degree if $\deg(v)\geq \sqrt{n}$ and it is super-high degree if $\deg(v)\geq n^{3/4}$. Our LCA for $3$-spanner assigns each edge of $E$ into one or more of the subsets $\LowDegEdges$, $\HighDegEdges$, or $\SuperDegEdges$ based on the degrees of its endpoints, where
$$\LowDegEdges=\{(u,v) \in E \mid \min\{\deg(u),\deg(v)\}\leq \sqrt{n}\}, $$
$$\HighDegEdges=\{(u,v) \in E \mid \sqrt{n}< \min\{\deg(u),\deg(v)\}\leq n^{3/4}\},\mbox{~and~}\SuperDegEdges=E \setminus (\LowDegEdges \cup \HighDegEdges).$$

Because vertices of degree at most $\sqrt{n}$ have $O(n\cdot \sqrt{n}) = O(n^{3/2})$ incident edges in total, we may afford to keep all these edges, letting $\LowDegSpanner = (V,\LowDegEdges)$. Thus, an LCA simply needs to check the degrees of both endpoints (via \degreeP~probes), and answer \YES~precisely when both (or in fact, even one) have degrees at most $\sqrt{n}$. From now on, assume that $\deg(u),\deg(v)\geq \sqrt{n}$.

\subsection{$3$-spanner for the edges $\HighDegEdges$}\label{sec:3high} 
We pick a random \emph{center} set $S$ of size $O(\sqrt{n} \log n)$ by sampling vertex $v \in V$ into $S$ independently with probability $p=\Theta((\log n)/\sqrt{n})$. For now, we assume that given an $\ID$ of a vertex $v$, we can decide in $O(1)$ time if $v \in S$. At the end of the section, we describe how to implement this using a seed of $O(\log n)$ random bits. 
For each endpoint $v$ of  $\HighDegEdges$, let $S(v)=\Gamma'(v)\cap S$ where $\Gamma'(v)$ is the set of the first $\sqrt{n}$ neighbors of $v$ in $\Gamma(v)$. By Chernoff bound we have that $|S(v)|=\Theta(\log n)$ (and in particular, $S(v)$ is non-empty). We call $S(v)$ the \emph{multiple-center} set of $v$. 
The algorithm adds to $\HighDegSpanner$ the edges $(v,s)$ connecting $v$ to each of its centers $s \in S(v)$. This adds a total of $O(n\log n)$ edges.

Next, for every $v$ with $\deg(v)=O(n^{3/4})$, the algorithm traverses its neighbor list $\Gamma(v)=\{u_1,\ldots, u_\ell\}$ and adds the edges $(u_i,v) \in \HighDegEdges$ to the spanner $\HighDegSpanner$ only if $u_i$ belongs to a \emph{new} cluster; i.e., $u_i$ has a center $s \in S(u_i)$ that no previous neighbor $u_{j}$, $j< i$, has as its center in $S(u_j)$. Since the algorithm adds an edge whenever a new center is revealed and there are $O(\sqrt{n}\log n)$ centers, the total number of edges added to the spanner is $O(n^{3/2}\log n)$.

We next describe the LCA that, given an edge $(u,v) \in \HighDegEdges$, says $\YES$ iff $(u,v) \in \HighDegSpanner$.
We assume throughout that $\deg(v)\leq \deg(u)$, so $\deg(v)=O(n^{3/4})$.
First, by probing for the first $\sqrt{n}$ neighbors of $u$ and $v$, one can compute the center-sets $S(u)$ and $S(v)$ each containing $O(\log n)$ centers in $S$. 
Next, the algorithm probes for all of $v$'s neighbors $\Gamma(v)=\{u_1,\ldots, u_j=u, \ldots, u_\ell\}$. 
For every neighbor $u_{i}$ appearing before $u$ in $\Gamma(v)$, i.e., for every $i < j$, and for every center $s \in S(u)$, the algorithm makes a \emph{cluster-membership test} for $s$ and $u_i$. This cluster-membership test can be answered by making a single $\adjacencyP$ probe on the pair $\langle u_i, s \rangle$, namely $s\in S(u_i)$ only if $s$ is among the first $\sqrt{n}$ neighbors of $u_i$.
Eventually, the algorithm $\mathcal{A}_{\high}$ answers $\YES$ only if there exists $s' \in S(u)$ such that $s' \notin \bigcup_{i=1}^{j-1}S(u_{i})$. It is straightforwards to verify that the probe complexity is $\widetilde{O}(\deg(u)+\sqrt{n})=O(n^{3/4})$.

Finally, we show that $\HighDegSpanner$ is indeed a $3$-spanner. For every edge $(u,v)$ not added to the spanner, let $s \in S(u)$ and let $u_i$ be the first vertex in $\Gamma(v)$ satisfying $s \in S(u_i)$. By construction, $(u_i,v)\in \HighDegSpanner$ and also the edges $(u_i,s)$ and $(u,s)$ are in the spanner $\HighDegSpanner$, providing a path of length $3$ in $\HighDegSpanner$. See Figure~\ref{fig:3high} for an illustration of $\HighDegSpanner$.

\begin{figure}[h!]
  \begin{center}
    \includegraphics[width=0.55\textwidth]{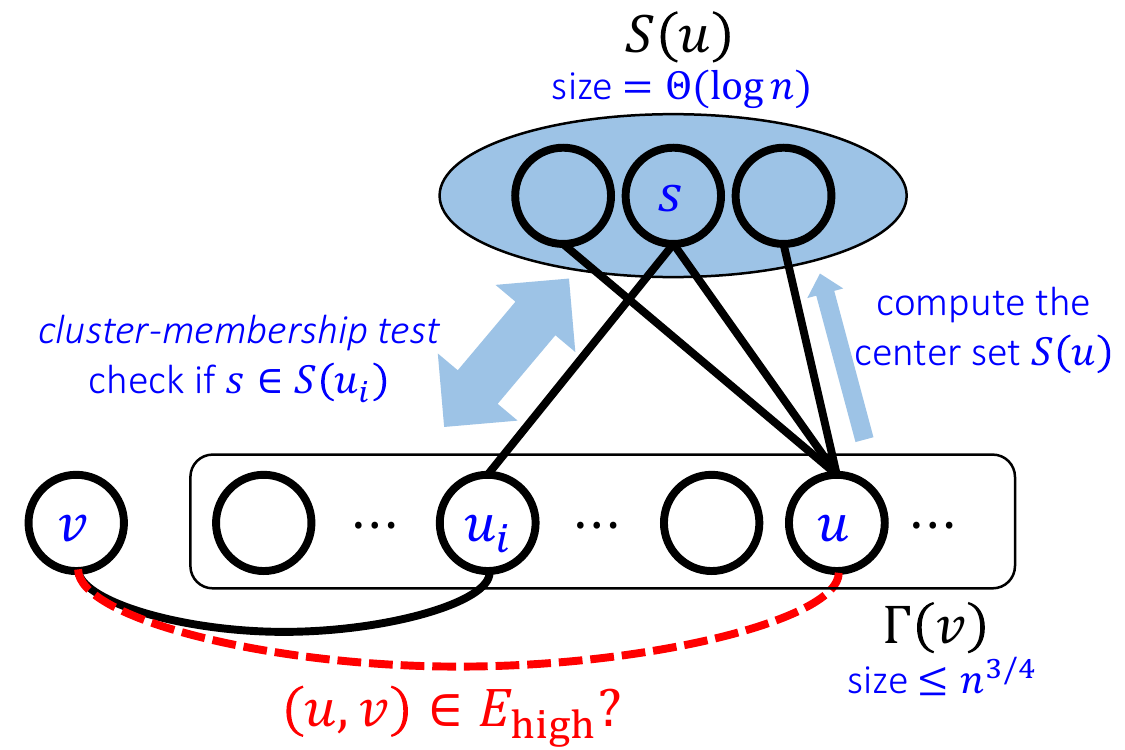}
  \end{center}
  \caption{Illustration for the local construction of $\HighDegSpanner$.}
\label{fig:3high}
\end{figure}

\subsection{$3$-spanner for the edges $\SuperDegEdges$}\label{sec:3super} 
We proceed by describing the construction of the $3$-spanner $\SuperDegSpanner$ that takes care of the edges $\SuperDegEdges$. Let $S'$ be a collection of $O(n^{1/4}\log n)$ centers obtained by sampling each $v \in V$ independently with probability $p'=\Theta((\log n)/n^{3/4})$. For each vertex $v$, define its center set $S'(v)$ to be the members of $S'$ among the first $n^{3/4}$ neighbors of $v$, and if $\deg(v)\leq n^{3/4}$, then $S'(v)=S' \cap \Gamma(v)$. First, as in the construction of $\HighDegSpanner$, the algorithm connects each $v$ to each of its centers by adding the edges $(u,s)$ for every $u$ and $s \in S'(u)$ to the spanner $\SuperDegSpanner$.

Consider a vertex $v$ and divide its neighbor list into consecutive \emph{blocks} $\Gamma_1(v),\ldots, \Gamma_\ell(v)$, each of size $n^{3/4}$ (expect perhaps for the last block). In every block $\Gamma_i(v)=\{u_{i,1}, \ldots, u_{i,\ell'}\}$, the algorithm adds the edge $(v,u_{i,j})$ to the spanner $\SuperDegSpanner$ only if $u_{i,j}$ belongs to a \emph{new} cluster with respect to all other vertices that appear before it in that block. Formally, the edge  $(v,u_{i,j})$ is added iff there exists $s \in S'(u_{i,j})$ such that $s \notin \bigcup_{q \leq j-1}S'(u_{i,q})$. 
This completes the description of the construction.
Observe that within each block, the LCA adds an edge for each new center. W.h.p., there are $O(n/n^{3/4}) = O(n^{1/4})$ blocks and $|S'| = O(n^{1/4} \log n)$ centers, so $O(\sqrt{n} \log n)$ edges are added for each $v$, yielding a spanner of size $O(n^{3/2}\log n)$.

The LCA $\mathcal{A}_{\super}$ is very similar to $\mathcal{A}_{\high}$: the main distinction is that given an edge $(u,v)$ with $\deg(u)\geq n^{3/4}$, the algorithm $\mathcal{A}_{\super}$ will probe only for the block $\Gamma_i(v)=\{u_{i,1}, \ldots, u_{i,j}=u,u_{i,\ell'}\}$ to which $v$ belongs, ad will make its decision only based on that block. 
By probing for the degree of $v$, and the index $j$ such that $u$ is the $j^\textrm{th}$ neighbor of $v$, one can compute the block $\Gamma_i(v)$ by making $n^{3/4}$ $\neighborP$ probes. 
In addition, by probing for the first $n^{3/4}$ neighbors of both $u$ and $v$, one can compute the multiple-center sets $S'(u)$ and $S'(v)$. Finally, the algorithm applies a cluster-membership test for each pair $s \in S'(u)$ and $u_{i,q}$ for $q \leq j-1$. It returns $\YES$ only if there exists $s \notin \bigcup_{q \leq j-1}S'(u_{i,q})$.  Hence, the number of probes made by the LCA is w.h.p.~bounded by $|\Gamma_i(v)| \cdot |S'(u)| = O(n^{3/4}\log n)$.

We now show that $\SuperDegSpanner$ is a $3$-spanner for the edges $\SuperDegSpanner$. Let $(u,v)$ be such that $\deg(u)\geq n^{3/4}$ and let $\Gamma_i(v)$ be the block in $\Gamma(v)$ to which $u$ belongs. Since $\deg(u)\geq n^{3/4}$, w.h.p. $|S'(u)|=\Theta(\log n)$.  Assume that $(u,v) \notin \SuperDegSpanner$. Fix $s \in S'(u)$ and let $u_{i,q}$ be the first vertex in $\Gamma_i(v)$ that belongs to the cluster of $s$. Since $(u,v) \notin \SuperDegSpanner$, such a vertex $u_{i,q}$ is guaranteed to exist. The spanner $\SuperDegSpanner$ contains the edges $(s,u), (s,u_{i,q})$ and $(v,u_{i,q})$, thus containing a path of length $3$ between $u$ and $v$. See Figure~\ref{fig:3super} for an illustration of $\SuperDegSpanner$.

\begin{figure}[!h]
  \begin{center}
    \includegraphics[width=0.75\textwidth]{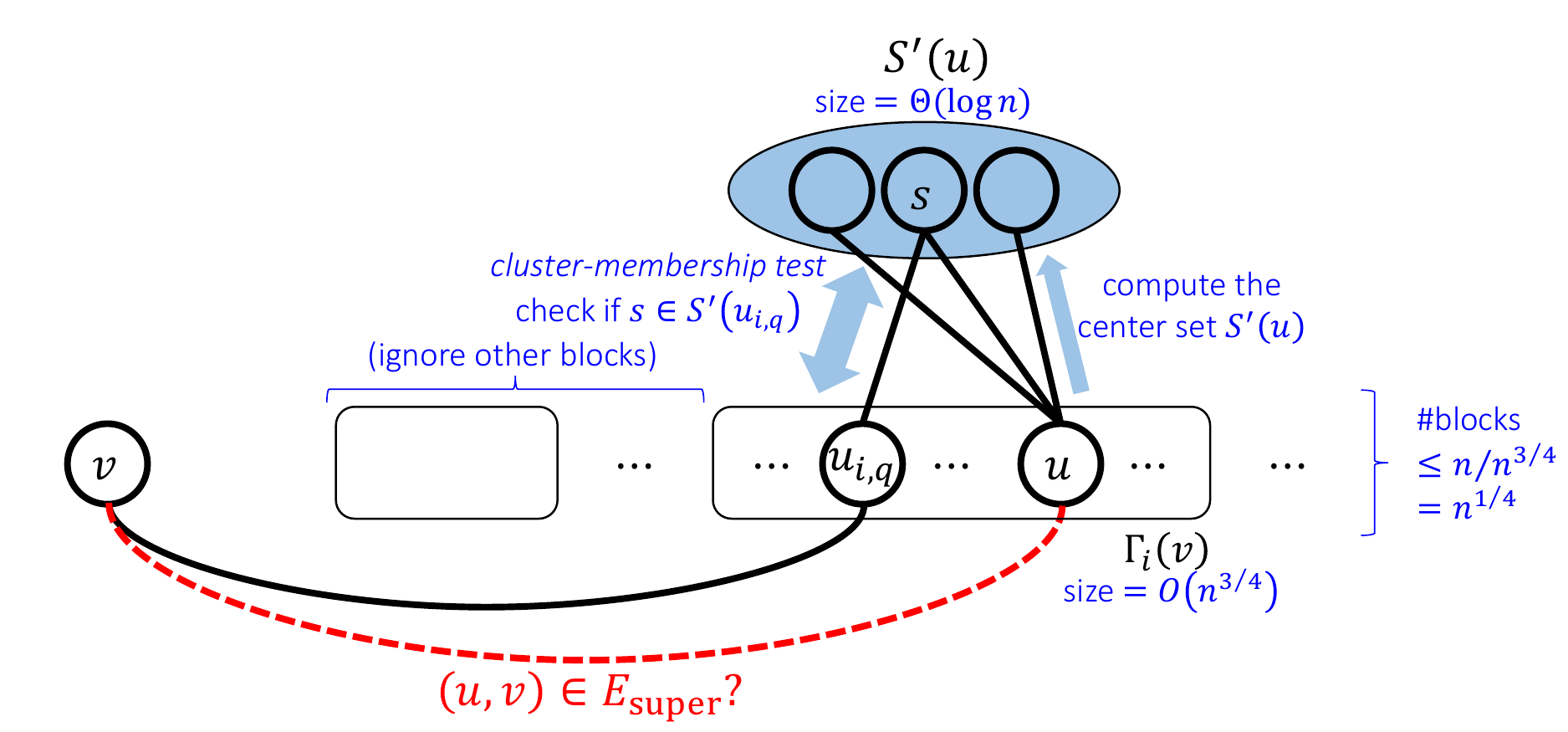}
  \end{center}
  \caption{Illustration for the local construction of $\SuperDegSpanner$.}
\label{fig:3super}
\end{figure}

\subsection{The Final LCA}\label{sec:3final} Given an edge $(u,v)$ the algorithm says $\YES$ if one of the following holds:
\begin{compactitem}
\item
$\deg(u),\deg(v)\leq \sqrt{n}$.
\item
$u \in S(v) \cup S'(v)$ (or vice versa).
\item
the local algorithm $\mathcal{A}_{\high}$ says $\YES$ on edge $(u,v)$.
\item
the local algorithm $\mathcal{A}_{\super}$ says $\YES$ on edge $(u,v)$.
\end{compactitem}
This completes the $3$-spanner LCA from Theorem \ref{thm:35-main}. 

\paragraph{Missing piece: computing centers in the LCA model}
In the LCA model, we do not generate the entire set $S$ (or $S'$) up front. Instead, we may verify whether $v \in S$ on-the-fly using $v$'s $\ID$ by, e.g., applying a random map (chosen according to the given random tape) from $v$'s $\ID$ to $\{0, 1\}$ with expectation $p$. In fact, this hitting set argument does not require full independence -- the discussion on reducing the amount of random bits is given in Section~\ref{APPEND:bounds}, but for now we formalize it as the following observation.
\begin{observation}[Local Computation of Centers]\label{obs:centercheck}
Let $S$ be a center set obtained by placing each vertex into $S$ independently with probability $p = \Theta(\log n/\Delta)$. W.h.p., $S$ forms a hitting set for the collection of neighbor sets of all vertices of degree at least $\Delta$. Further, under the LCA model, we may check whether $v \in S$ locally without making any probes.
\end{observation}

\section{LCA for $5$-Spanners} \label{sec:5spanner}

We now consider LCAs for $5$-spanners, aiming for spanners of size $\widetilde{O}(n^{4/3})$ with probe complexity $\widetilde{O}(n^{5/6})$. 
We start by noting that the construction of $\SuperDegSpanner$ for the $3$-spanners in fact gives for every $r\geq 1$, a $3$-spanner of size $\widetilde{O}(n^{1+1/r})$ for the subset of edges $(u,v)$ with $\min\{\deg(u),\deg(v)\}\geq n^{1-1/(2r)}$: this is achieved by instead setting the threshold for super-high degree at $n^{1-1/(2r)}$, pick $|S'| = \widetilde{O}(n^{1/(2r)})$ centers, and use block size $n^{1-1/(2r)}$. The probe complexity for querying the spanner is $\widetilde{O}(n^{1-1/(2r)})$.
For $5$-spanner, by taking $r = 3$, one takes care of all edges $(u,v)$ with $\max\{\deg(u),\deg(v)\}\geq n^{5/6}$.

Let $\LowDeg=n^{1/r}$, $\MedDeg = n^{1/2 - 1/(2r)}$ and $\SuperDeg = n^{1-1/(2r)}$. For the purpose of constructing $5$-spanners for general graphs, we let $r=3$, simplifying the thresholds to $\LowDeg=\MedDeg = n^{1/3}$ and $\SuperDeg=n^{5/6}$.) Again, we may afford to keep all edges incident to some vertex of degree at most $\LowDeg$.

For integers $a\leq b$, let $V_{[a,b]}=\{ v \in V(G) ~\mid~ \deg(v) \in [a,b]\}$.
We will design a subgraph $H\subseteq G$ that will take care of the remaining edges $\MedDegEdges = E(V_{[\MedDeg,\SuperDeg]}, V_{[\MedDeg,\SuperDeg]})$.

\begin{definition}[Deserted and Crowded vertices]
A vertex $v \in V_{[\MedDeg,\SuperDeg]}$ is \emph{deserted} if at least half of its neighbors in $\Gamma_{\MedDeg,1}(v)$ are of degree at most $\SuperDeg$; i.e., $|\Gamma_{\MedDeg,1}(v)\cap V_{[1,\SuperDeg]}|\geq \MedDeg/2.$ Otherwise, the vertex is \emph{crowded}.
\end{definition}

\paragraph{Criteria for edges}
We aim to take care of edges for which both endpoints are in $V_{[\MedDeg,\SuperDeg]}$. To categorize our edges for the purpose of constructing $5$-spanners, we need the following partition of these vertices.

Let $\DesertedVer$ (resp., $\CrowdedVer$) be the set of deserted (resp., crowded) vertices in $V_{[\MedDeg,\SuperDeg]}$.
Given a vertex, we can verify whether it is in any of these sets using $O(\MedDeg)$ probes by checking the degrees of $v$ and each vertex in $\Gamma_{\MedDeg,1}(v)$.
We then assign each $(u,v) \in E$ into one of the four cases $\{\low,\cell,\rep,\super\}$ as given in the table below. It is straightforward to verify that when $\LowDeg = \MedDeg$ (namely when we choose $r = 3$, which also yields the required performance), these four cases take care of all edges in $E$.
We note that $\RepSpanner$ assumes that $\SuperDegSpanner$ is included: $\RepEdges$ is taken care by $\RepSpanner \cup \SuperDegSpanner$, not by $\RepSpanner$ alone.

\begin{table*}[!h]
\centering
\renewcommand{\arraystretch}{1.75} 
\resizebox{\textwidth}{!}{%
\begin{tabular}{|c||c|c|c|} \hline
  {\bf Subset} & {\bf Criteria} & {\bf \# Edges} & {\bf Probe Complexity} \\ \hline\hline
  $\LowDegEdges$ & $(u,v) \in E(V,V_{[1, \LowDeg]})$ & $O(n\cdot\LowDeg)=O(n^{1+\frac{1}{r}})$ & $O(1)$ \\ \hline
  $\CellEdges$ & $(u,v) \in E(\DesertedVer, \DesertedVer)$ & $O(\frac{n^2 \log^2 n}{\MedDeg^2})=O(n^{1+\frac{1}{r}} \log^2 n)$ & $O((\SuperDeg + \MedDeg^2) \log^2 n)=O(n^{1-\frac{1}{2r}} \log^2 n)$ \\ \hline
  $\RepEdges$ & $(u,v) \in E(V_{[\MedDeg, \SuperDeg]}, \CrowdedVer)$ & $O(\frac{n^2}{\SuperDeg} \cdot \log n)=O(n^{1+\frac{1}{r}} \log n)$ & $O(\SuperDeg \log^3 n)=O(n^{1-\frac{1}{2r}} \log^3 n)$ \\ \hline
  $\SuperDegEdges$ & $(u,v) \in E(V, V_{[\SuperDeg, n)})$ & $O(\frac{n^3\log n}{\SuperDeg^2})=O(n^{1+\frac{1}{r}} \log n)$ & $O(\SuperDeg \log n)=O(n^{1-\frac{1}{2r}} \log n)$ \\ \hline
\end{tabular}
}
\caption{Edge categorization for the construction of $5$-spanners.}
\end{table*}

\paragraph{LCA for $\CellEdges$: the cluster partitioning method}\label{sec:5cell}
The algorithm is as follows.
\begin{compactitem}
\item Only vertices of degree at most $\SuperDeg$ are chosen to be in $S$ with probability $p = \Theta((\log n)/\MedDeg)$. Since at least half the vertices in $\Gamma_{\MedDeg,1}(v)$ for any $v\in \DesertedVer$ have degree smaller than $\SuperDeg$, we have that w.h.p.  $|S(v)|=\Theta(\log n)$ the cluster-membership test can be done with constant number of probes. Let us denote by $\ClusterOf(s) = \{s\} \cup \{v: s \in S(v)\}$ the cluster of center $s$.
\item The partitioning of clusters into buckets is defined in a consistent way (regardless of the given query edge); for instance, create a list of vertices in the cluster, sort them according to their $\ID$s, divide the list into buckets of size $\MedDeg$ possibly except for the last one. Note that we partition $\ClusterOf(s)$ and $\ClusterOf(t)$ separately -- we do not combine their elements. Similarly, once we obtain buckets containing $u$ and $v$, the order in which we check the adjacency of $u'$ and $v'$ must be consistent. To this end, define the $\ID$ of an edge $(u,v)$ as $($\ID$(u), $\ID$(v))$, where the comparison between edge $\ID$s is lexicographic. Thus, this step only adds the edge of minimum $\ID$ between the two clusters.
\item We also set the precondition $(u,v)\in E(V_{[\MedDeg,n)}, V_{[\MedDeg,n)})$, and consistently only allow candidate pairs $(u',v')\in E(V_{[\MedDeg,n)}, V_{[\MedDeg,n)})$, to ensure that the lexicographically first edge of this exact specification is added if one exists. We do not restrict to $\CellEdges$, which require both endpoints to be deserted vertices, because checking whether $(u',v') \in \CellEdges$ would take $\Theta(\MedDeg)$ probes instead of constant probes. We restrict to edges whose endpoints have degrees at least $\MedDeg$ instead of considering the entire $E$ so that $\CentersOf$ would be well-defined. \end{compactitem}

\vspace{4pt}\noindent\fbox{\begin{minipage}{\dimexpr\textwidth-2\fboxsep-2\fboxrule\relax}
\textbf{Local construction of $\CellSpanner$.} Each $v\in V_{[1,\SuperDeg]}$ is added to $S$ with probability $p = \Theta(\log n/\MedDeg)$.
\begin{compactenum}[]
\item \textbf{(A)} If $u \in \CentersOf(v)$ or $v \in \CentersOf(u)$, answer \YES.
\item \textbf{(B)} If $(u,v)\in E(V_{[\MedDeg,n)}, V_{[\MedDeg,n)})$:
\begin{compactitem}
\item Compute $S(u)$ and $S(v)$ by iterating through $\Gamma_{\MedDeg,1}(u)$ and $\Gamma_{\MedDeg,1}(v)$.
\item For each pair of $s \in S(u)$ and $t \in S(v)$:
\begin{compactitem}
\item Partition each of the clusters $\ClusterOf(s)$ and $\ClusterOf(t)$ into buckets of size (mostly) $\MedDeg$. Denote the buckets containing $u$ and $v$ by $\CellOf(u,s)$ and $\CellOf(v,t)$, respectively.
\item Iterate through each pair of $u' \in \CellOf(u,s)$ and $v' \in \CellOf(v,t)$ and check if $(u',v')\in E(V_{[\MedDeg,n)}, V_{[\MedDeg,n)})$. Answer \YES~if the edge of minimum $\ID$ found is $(u',v') = (u,v)$.
\end{compactitem}
\end{compactitem}
\end{compactenum}
\end{minipage}}\vspace{2pt}

\begin{lemma}\label{lem:cellsp}
For $1 \leq \MedDeg \leq \sqrt{n} \leq \SuperDeg \leq n$, there exists a subgraph $\CellSpanner \subseteq G$ such that w.h.p.:
\begin{compactenum}[\hspace{\parindent}(i)]
\item $\CellSpanner$ has $O(\frac{n^2 \log^2 n}{\MedDeg^2})$ edges, 
\item $\CellSpanner$ takes care of $\CellEdges$; that is, for every $(u,v)\in \CellSpanner$, $\dist(u,v,\CellSpanner)\leq 5$, and 
\item for a given edge $(u,v) \in E$, one can test if $(u,v) \in \CellSpanner$ by making $O((\SuperDeg + \MedDeg^2) \log^2 n)$ probes.
\end{compactenum}
\end{lemma}

\begin{proof}

\paragraph{(i) Size} In (A) we add $|S(v)| = \Theta(\log n)$ edges for each $v \in \DesertedVer$, which constitutes to $O(n \log n)$ edges in total. In (B), we add one edge between each pair of buckets. We now compute the total number of buckets. The total size of clusters $\sum_{s\in S}|\ClusterOf(s)| \leq |S| + \sum_{v \in V_{[\MedDeg,n)}} |S(v)| = O(n \log n)$, so there can be up to $O((n \log n)/\MedDeg)$ full buckets of size $\MedDeg$. As buckets are formed by partitioning $|S|$ clusters, there are up to $|S| = \Theta((n \log n)/\MedDeg)$ remainder buckets of size less than $\MedDeg$. Thus, there are $\Theta((n \log n)/\MedDeg)$ buckets, and $O(((n \log n)/\MedDeg)^2)$ edges are added in (B).

\paragraph{(ii) Stretch} Suppose that $(u,v)$ is omitted. Fix centers $s \in S(u)$ and $t \in S(v)$, then the lexicographically-first edge $(u',v') \in E(\CellOf(u,s),\CellOf(v,t))$ must have been added to $\CellSpanner$, forming the path $\langle u, s, u', v', t, v\rangle$ (or shorter, if there are repeated vertices), yielding $\dist(u,v,\CellSpanner)\leq 5$.

\paragraph{(iii) Probes} Computing $S(u)$ and $S(v)$ takes $O(\MedDeg)$ probes. For each pairs of centers, we scan through the entire neighbor-lists $\Gamma(s)$ and $\Gamma(t)$ and collect all vertices in their respective clusters. This takes $O(\SuperDeg)$ probes each because we restrict to centers of degree at most $\SuperDeg$. Given the clusters, we identify the buckets containing $u$ and $v$ each of size $O(\MedDeg)$. We then check through candidates $(u',v')$ between these buckets, taking $O(\MedDeg^2)$ \adjacencyP~probes. So, each pair of centers requires $O(\SuperDeg + \MedDeg^2)$ total probes. We repeat the process for $|S(u)| \cdot |S(v)| = O(\log^2 n)$ pairs of centers w.h.p., yielding the claimed probe complexity.
\end{proof}
\begin{figure}[!h]
  \begin{center}
    \includegraphics[width=0.7\textwidth]{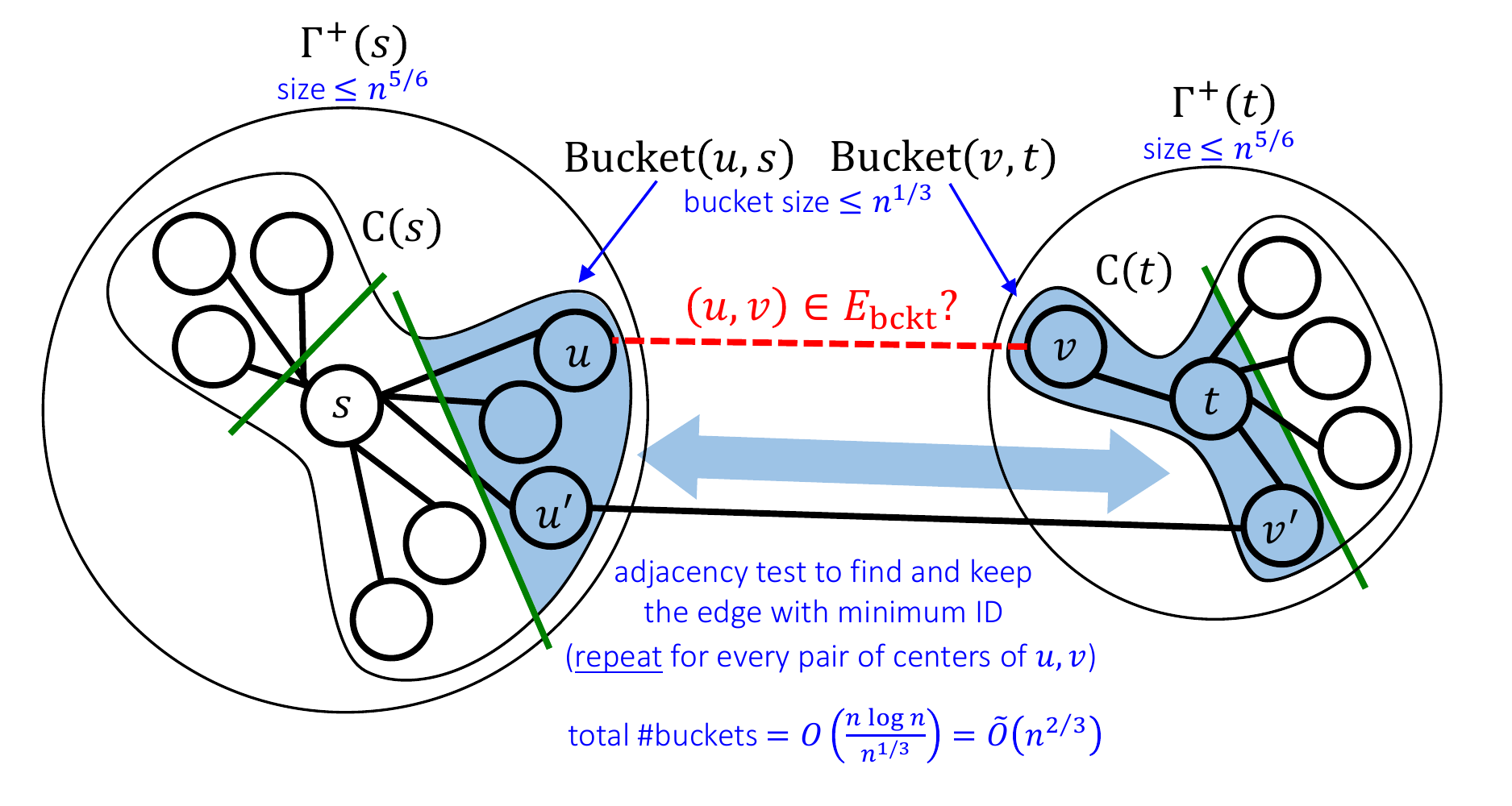}
  \end{center}
  \caption{Illustration for the local construction of $\CellSpanner$. Green lines show the partition of clusters into buckets.}
\label{fig:5cell}
\end{figure}

\paragraph{LCA for $\RepEdges$: the Representative method}\label{sec:5rep}
We first explain the computation of the representative set $\RepsOf(v)$ for a croweded vertex $v \in \CrowdedVer$, i.e., a collection of neighbors of $v$ that have  degree at least $n^{5/6}$. Using the random bits and the vertex $\ID$, we sample a set $R_v$ of $\Theta(\log n)$ (not necessarily distinct) indices in $[\MedDeg]$ at random (for details, see Sec.~\ref{APPEND:bounds}). Denote the neighbor-list of $v$ by $\{x'_1, \ldots, x'_{\deg(v)}\}$, then define $\RepsOf(v) = \{x'_i: i \in R_v \textrm{ and } \deg(x'_i) \geq \SuperDeg\}$. Then since at least half of the vertices in $\Gamma_{\MedDeg,1}(v)$ are of degree at least $\SuperDeg$, w.h.p.~$\RepsOf(v) \neq \emptyset$. For consistency, we allow the same definition for $\RepsOf(v)$ for any $v \in V_{[\MedDeg,n)}$ as well, even if it may result in empty sets of representatives.
Hence computing $\RepsOf(v)$ takes $O(\log n)$ probes\footnote{The na\"{\i}ve solution traverses the entire $\MedDeg$ first neighbors of $v$ which is too costly.}.

Let $\SuperDegEdges=\{(u,v) \in E ~\mid~ \max\{\deg(u),\deg(v)\}\geq n^{5/6}\}$ and apply the $3$-spanner algorithm algorithm of Sec.~\ref{sec:three} to construct a subgraph $\SuperDegSpanner$ that takes care of the edges $\SuperDegEdges$. To construct $\SuperDegSpanner$ the algorithm (fully described\footnote{Upon replacing the degree threshold of $n^{3/4}$ with $n^{5/6}$.} in Sec. \ref{sec:three}) samples a set 
$S'$ of centers by picking each $v \in V$ independently with probability $O(\log n/n^{5/6})$. For every $v$ with $\deg(v)\geq n^{5/6}$, let $S'(v)$ be the sampled neighbors in $S'\cap \Gamma_1(v)$ where $\Gamma_1(v)$ is the first block of size $n^{5/6}$ in $\Gamma(v)$. This allows us to check membership to a cluster of $s \in S'$ using a single adjacency probe. 
The idea would be to extend the $1$-radius clusters of $S'$ by one additional layer consisting of the crowded vertices connected to the cluster via their representatives.

For convenience, for a crowded $v$, define $\CentersOfRepsOf(v) = \cup_{x' \in \RepsOf(v)} S'(x')$, the set of (multiple) centers of any of $v$'s representatives. Observe that by adding the edge $(v,x')$ to $\RepSpanner$ for every $x' \in \RepsOf(v)$, it yields that $\dist(v,s,\RepSpanner \cup \SuperDegSpanner) \leq 2$ for any $s \in \CentersOfRepsOf(v)$.

Consider the query $(u,v)$, and suppose that $v = v'_i$ is the $i^\textrm{th}$ neighbor in $u$'s neighbor-list, $\Gamma(u) = \{v'_1, \ldots, v'_{\deg(u)}\}$. We then add $(u,v)$ to $\RepSpanner$ if and only if $v$ introduces a new center through some representative; that is, $\CentersOfRepsOf(v'_i) \setminus \cup_{j<i} \CentersOfRepsOf(v'_j) \neq \emptyset$. To verify this condition locally, we first compute $\CentersOfRepsOf(v)$, and for each of $\{v'_j\}_{j<i}$, $\RepsOf(v'_j)$. Then, we discard $(u,v)$ if for every center $s \in \CentersOfRepsOf(v)$, there exists $x$ and $v'_j$ where $x \in \RepsOf(v'_j)$ and $s \in S'(x)$; the last condition takes constant probes to verify. This gives the full LCA for constructing $\RepSpanner$ below.

\begin{figure}[!h]
\vspace{2pt}\noindent\fbox{\begin{minipage}{\dimexpr\textwidth-2\fboxsep-2\fboxrule\relax}
\textbf{Local construction of $\RepSpanner$.} Each $v\in V$ is added to $S'$ with probability $p = \Theta((\log n)/\SuperDeg)$.
\begin{compactenum}[]
\item \textbf{(A)} If $v \in V_{[\MedDeg,\SuperDeg]}$ and $u \in \RepsOf(v)$, answer \YES.
\item \textbf{(B)} If $u, v \in V_{[\MedDeg,\SuperDeg]}$:
\begin{compactitem}
\item Compute $\CentersOfRepsOf(v)$.
\item Denote the neighbor-list of $u$ by $\{v'_1, \ldots, v'_{\deg(u)}\}$; identify $i$ such that $v = v'_i$.
\item For each vertex $w \in \{v'_1, \ldots, v'_{i-1}\}$, if $w \in V_{[\MedDeg,\SuperDeg]}$, compute $\RepsOf(w)$.
\item For each $s \in \CentersOfRepsOf(v)$, iterate to check for a vertex $x$ in any of the $\RepsOf(w)$'s obtained above, such that $s \in S'(x)$. Answer \YES~if there exists a vertex $s$ where no such $x$ exists.
\end{compactitem}
\end{compactenum}
\end{minipage}}
\caption[Procedure for the local construction of $\RepSpanner$]{Procedure for the local construction of $\RepSpanner$.}\end{figure}

\begin{lemma}\label{lem:repsp}
For $1 \leq \MedDeg \leq \SuperDeg \leq n$, there exists a subgraph $\RepSpanner \subseteq G$ such that w.h.p.:
\begin{compactenum}[\hspace{\parindent}(i)]
\item $\RepSpanner$ has $O(n^2/\SuperDeg \cdot \log n)$ edges, 
\item $\RepSpanner \cup \SuperDegSpanner$ takes care of $\RepEdges$; that is, for every $(u,v)\in \RepEdges$, $\dist(u,v,\RepSpanner \cup \SuperDegSpanner)\leq 3$, and 
\item for a given edge $(u,v) \in E$, one can test if $(u,v) \in \RepSpanner$ by making $O(\SuperDeg \log^3 n)$ probes.
\end{compactenum}
\end{lemma}

\begin{proof}
\paragraph{(i) Size} W.h.p., in (A) we add at most $\sum_{v\in V_{[\MedDeg,\SuperDeg]}}|\RepsOf(v)| \leq n \cdot O(\log n) = O(n \log n)$. Similarly to the analysis of $\HighDegSpanner$, in (B) we add $|S'| = O((n \log n)/\SuperDeg)$ edges per vertex $u$, so $|E(\RepSpanner)| = O(n^2/\SuperDeg \cdot \log n)$.

\paragraph{(ii) Stretch} This claim follows from the argument given in the overview, and is similar to the analysis of $\HighDegSpanner$.

\paragraph{(iii) Probes} Computing $\CentersOfRepsOf(v)$ takes $O(\log n) \cdot \SuperDeg = O(\SuperDeg \log n)$ (recall that we only check $\Gamma_{\SuperDeg,1}$ of each reprsentative). Note also that $|\CentersOfRepsOf(v)| = O(\log^2 n)$ since $v$ has $O(\log n)$ representative, each of which belongs to $\Theta(\log n)$ clusters. Computing $\RepsOf$ for each neighbor $w \in \{v'_j\}_{j<i}$ of $u$ takes $O(\log n)$ probes each, which is $O(\SuperDeg \log n)$ in total since $\deg(u) \leq \SuperDeg$. This also introduces up to $\SuperDeg \cdot O(\log n)$ representatives in total. Checking whether each of the $O(\log^2n )$ centers in $\CentersOfRepsOf(v)$ is a center of each of these $O(\SuperDeg\log n)$ representative takes, in total w.h.p., $O(\SuperDeg \log^3 n)$ probes.
\end{proof}

\begin{figure}[!h]
  \begin{center}
    \includegraphics[width=0.7\textwidth]{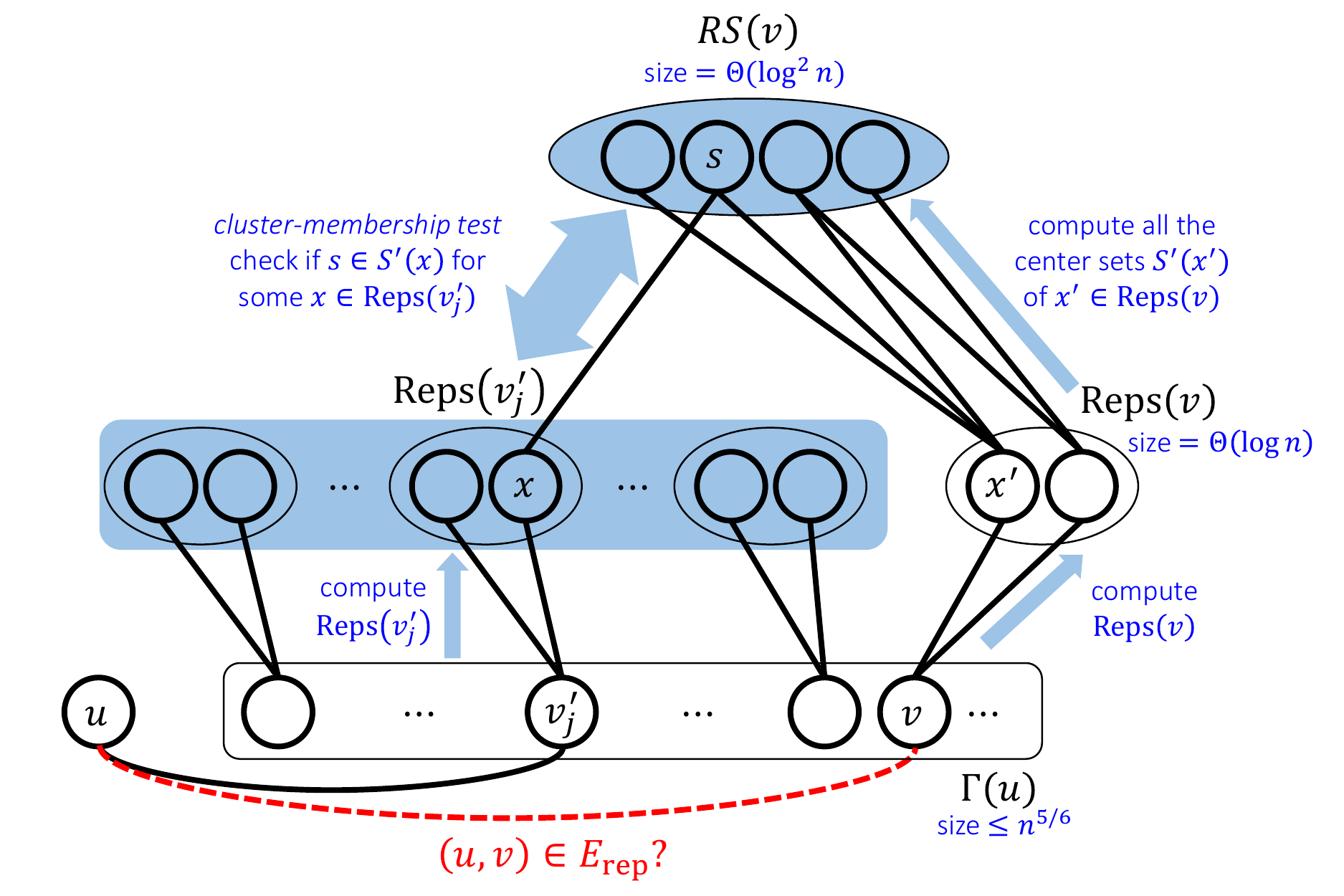}
  \end{center}
  \caption{Illustration for the local construction of $\RepSpanner$.}
\label{fig:5rep}
\end{figure}

\paragraph{Final $5$-spanner results}\label{sec:5results}
To obtain an LCA for $5$-spanners, we again invoke all of our LCAs for the four cases. Applying Lemma~\ref{lem:cellsp} and~\ref{lem:repsp}, we obtain the following LCA result for $5$-spanner in general graphs.

\begin{theorem}\label{thm:fivespannergen}
For every $n$-vertex simple undirected graph $G=(V,E)$ there exists an LCA for $5$-spanner with $O(n^{4/3}\log^2 n)$ edges and probe complexity $O(n^{5/6} \log^3 n)$. 
\end{theorem}

Again, by combining results for larger degrees, we obtain an LCA for $5$-spanners with smaller sizes on graphs with minimum degree at least $n^{1/2-1/(2r)}$.
\begin{theorem}\label{thm:fivespannerk}
For every $r\geq 1$ and $n$-vertex simple undirected graph $G=(V,E)$ with minimum degree at least $n^{1/2-1/(2r)}$, there exists a (randomized) LCA for $5$-spanner with $O(n^{1+1/r} \log^2 n)$ edges and probe complexity of $O(n^{1-1/(2r)} \log^3 n)$. 
\end{theorem}

\section{LCA for $O(k^2)$ Spanners}\label{sec:fullsparse}

In this section, we prove Theorem~\ref{thm:lowdegree} by showing LCAs for $O(k^2)$ spanners and $\widetilde{O}(n^{1+1/k})$ edges.
The solution is inspired by the result of \citet{LeviLenzen17} and it is extended in two major aspects. First, we improve upon the stretch factor of the constructed spanner from $O(\log n \cdot(\Delta + \log n))$ down to $O(k^2)$ for any $k\geq 1$, thereby removing the dependencies on $\Delta$ and $n$ completely, at the cost of increasing the number of spanner edges from $\Theta(n)$ to $\widetilde{O}(n^{1+1/k})$. Second, we show how to implement a key part of their algorithm using a collection of $k$ bounded independence hash functions to 
reduce the number of random bits (kept at each machine) from linear to only polylogarithmic in $n$. We also remark that the probe complexity in our construction is improved by a factor of $\Delta$ compared to \cite{LeviLenzen17}.

\subsection{High-level Overview}

We now provide some preliminaries and an outline of our $O(k^2)$-spanner construction. Throughout the main part of this section, we fix two parameters $L=\Theta(n^{1/3})$ and $p=\Theta(\log n/L)$. We only need to consider $k=O(\log n)$ because, by the size-stretch tradeoff of spanners, any $k=\Omega(\log n)$ yields a spanner of (roughly) linear size, $\widetilde{O}(n^{1+1/k}) = \widetilde{O}(n)$. We note that LCAs in this section only make use of the \neighborP~probes. 

\paragraph{Sparse and dense vertices} We first sample a collection $S$ of $O((n\log n)/L) = \widetilde{O}(n^{2/3})$ centers, which is implemented locally by having each vertex elect itself as a center with probability $p$. We remark that we never explicitly enumerate the entire set $S$, but only rely on the fact that we may locally determine whether a given vertex $v$ is a center based on its \ID~and the randomness, without using any probes. Next, we partition our vertices into {\em sparse} and {\em dense} vertices with respect to the center set $S$ based on their distances to the respective closest centers: a vertex $v$ is considered sparse if it is at distance more than $k$ away from all centers, and it is dense otherwise. By a hitting set argument, if the $k^\textrm{th}$-neighborhood of $v$ is of size at least $L$, then it most likely contains a center, making $v$ a dense vertex. This observation suggests that to verify that a vertex is dense, we do not necessarily need to find some center in $v$'s potentially large $k^\textrm{th}$-neighborhood: it also suffices to confirm that the neighborhood itself is large.

\begin{definition}[Sparse and dense]
A vertex $v$ is \emph{sparse} in $G$ if $\Gamma^{k}(v,G) \cap S=\emptyset$ and otherwise, it is \emph{dense}. Denote the sets of sparse vertices and dense vertices by $\Vsparse$ and $\Vdense$, respectively.
\end{definition}
We next partition the edge set of $G$ into $\Esparse = E(V,\Vsparse)$ and $\Edense = E(\Vdense,\Vdense)$, then take care\footnote{As a reminder, to ``take care'' of an edge $(u,v)$, we ensure that in the constructed spanner, there is a $u$-$v$ path whose length is at most the desired stretch factor. See Definition~\ref{def:takecare} for its formal definition.} of them by constructing $\Hsparse \subseteq \Esparse$ and $\Hdense \subseteq \Edense$, so that $H=\Hsparse \cup \Hdense$ gives a spanner for all edges of $G$. See Table~\ref{table:k2-summary} for a summary of the properties of each spanner.

\begin{table*}[!h]
\begin{center}
\renewcommand{\arraystretch}{1.5} 
\resizebox{\textwidth}{!}{%
\begin{tabular}{|c||c|c|c|c|} \hline
{\bf Subset} & {\bf Criteria} & {\bf Spanner Edges} & {\bf \# Edges} & {\bf Probe Complexity} \\ \hline\hline
$\Esparse $& at least one endpoint is sparse & $\Hsparse$ & $O(kn^{1+1/k})$ & $O(\Delta^2 L^2)$ \\ \hline
\multirow{2}{*}{$\Edense$} & \multirow{2}{*}{both endpoints are dense} & $\HdenseI$ & $O(n)$ & $O(\Delta^2 L^2)$ \\ \cline{3-5}
 & & $\HdenseB$ & $O(n^{1+1/k} \log^4 n)$ & $O(p\Delta^4 L^3 \log n)$ \\ \hline
\end{tabular}
}
\caption[Edge categorization for the construction of $O(k^2)$-spanners]{Edge categorization for the construction of $O(k^2)$-spanners, with respective spanner sizes and probe complexities.}
\label{table:k2-summary}
\end{center}
\end{table*}

The outline of the construction is given as follows. For convenience, tables of various probe complexities for computing $\Hsparse$ and $\Hdense$ are provided: Table~\ref{table:sparse-probes} (page \pageref{table:sparse-probes}) and Table~\ref{table:dense-probes} (page \pageref{table:dense-probes}), respectively.

\paragraph{Taking care of $\Esparse$} (Section~\ref{sec:k2-sparse})
Attempting to leverage the clustering approach, we need to partition our vertices based on their distances to $S$. However, some vertices can be very far from all centers: connecting them to their respective closest centers would still incur a large stretch factor. We observe that every sparse vertex $v$ has a small $k^\textrm{th}$-neighborhood: $\deg_k(v,G)=|\Gamma^k(v,G)| = O(L)$ (hence the name ``sparse''). Thus, we may test whether some vertex $v$ is sparse by simply examining up to $O(L)$ vertices closest to it, using $O(\Delta L)$ probes. To take care of sparse vertices' incident edges $\Esparse$, we can then afford to identify the query edge's endpoints' $k^\textrm{th}$-neighborhoods and simulate a $k$-round distributed $(2k-1)$-spanner algorithm on the subgraph $\Gsparse = (V, \Esparse)$. We locally obtain our spanner $\Hsparse$ of $\Gsparse$ using $O(\Delta^2 L^2)$ probes.

\paragraph{Partitioning of dense vertices into Voronoi cells} (Section~\ref{sec:k2-dense-voronoi})
In the subgraph induced by dense vertices $\Gdense = (\Vdense, \Edense)$, all vertices are at distance at most $k$ from some center. We partition them into \emph{Voronoi cells} by connecting each of them to its closest center. We show that each dense vertex can find its shortest path to its center in $O(\Delta L)$ probes. Building on this subroutine, we straightforwardly connect vertices within each Voronoi cell to their center via these shortest paths, forming a \emph{Voronoi tree} of depth at most $k$, which in turn bounds the diameter of every Voronoi cell in our spanner by $2k$. In particular, our construction improves upon the construction of \cite{LeviLenzen17} that provides a diameter bound of $O(\Delta + \log n)$. We denote by $\HdenseI$ the set of Voronoi tree edges, as each tree spans vertices \textbf{i}nside the same Voronoi cell.

\paragraph{Refining Voronoi cells into small clusters} (Section~\ref{sec:k2-dense-cluster})
Naturally as our next step, we would like to consider our Voronoi cells as ``supervertices,'' and connect them via an $O(k)$-spanner with respect to this ``supergraph''. However, determining the connectivity in this supergraph is impossible in sub-linear probes, as a Voronoi cell may contain as many as $\Theta(n)$ vertices. To handle this issue, we define a local rule based on the subtree sizes of the Voronoi tree, which \emph{refines} our Voronoi cells. We show that this rule partitions the dense vertices into $\widetilde{O}(n/L)$ clusters of size $O(L)$ each, such that each vertex can identify its \emph{entire} cluster using $O(\Delta^3 L^2)$ probes.

\paragraph{Connecting between Voronoi cells through clusters} (Section~\ref{sec:k2-dense-connect-overview})
We then formalize local criteria for connecting Voronoi cells (through clusters), forming the set of spanner edges between clusters, $\HdenseB$, using $\widetilde{O}(\Delta^4 L^2)$ probes: the union $\Hdense = \HdenseI \cup \HdenseB$ is the desired spanner of $\Gdense$. For any omitted edge between clusters, $\HdenseB$ contains a path connecting the endpoints' Voronoi cells that, w.h.p., visits only $O(k)$ other Voronoi cells along the way. Since each Voronoi cell has a $2k$-diameter spanning Voronoi tree in $\HdenseI$, $\Hdense$ achieves the desired $O(k^2)$ stretch factor. The rules for choosing $\HdenseB$ are based on marking $\widetilde{O}(n^{1/3})$ random Voronoi cells along with the clusters therein, then adding at most $\widetilde{O}(n^{1/k})$ edges per each pair of cluster and marked cluster, using a total of $\widetilde{O}(n^{1+1/k})$ edges. Sections~\ref{sec:k2-dense-connect-algo}-\ref{sec:k2-dense-connect-proof} formalize these ideas into an efficient LCA, then show the desired properties of the constructed $\Hdense$ and wrap up the proof, respectively.


\paragraph{Reducing the required amount of independent random bits} For simplicity, our analysis in this section uses a linear number of independent random bits. This assumption for the above construction is deferred to Section~\ref{APPEND:bounds}, where we provide an implementation using only $O(\log^2 n)$ independent random bits.


\subsection{LCA for computing a $(2k-1)$-spanner $\Hsparse$ for $\Esparse$}\label{sec:k2-sparse}

\paragraph{Checking if a vertex is sparse or dense}
We first propose a variant of the breadth-first search (BFS) algorithm that, when executed starting from a vertex $v$, either finds $v$'s center or verifies that $v$ is sparse. We justify the necessity to employ a different BFS variant from that of the prior works, namely \cite{levi2016local,LeviLenzen17}, as follows. In these prior works, the BFS algorithm explores \emph{all} vertices in an entire level of the BFS tree in each step until some center is encountered, and chooses the center with the lowest $\ID$ among them. This distance tie-breaking rule via $\ID$~directly ensures that the set of vertices choosing the same center induces a connected component in $G$\footnote{If $v$ chooses $s$ at distance $d$ as its center, and another vertex $u$ is at distance $d'<d$ from $s$, then $u$ must also choose $s$ because $s$ is the center of minimum $\ID$ in $\Gamma^d(v,G)\supset \Gamma^{d'}(u,G)$, and there are no other centers in $\Gamma^{d-1}(v,G)\supset\Gamma^{d'-1}(u,G)$.}.

We have shown before that it suffices to explore $L$ vertices closest to a dense vertex $v$ in order to \emph{discover} some center. However, to \emph{choose} $v$'s center via the above approach, we must explore the entire last level of the BFS tree in order to apply the tie-breaking rule: this last level may contain as many as $\Theta(\Delta L)$ vertices. Instead, we aim to further reduce a factor of $\Delta$ from the probe complexity by designing a BFS algorithm that picks the first center it discovers as $v$'s center: this center may not be the lowest-$\ID$ center in that level. The desired connectivity guarantee does not trivially follow under this rule, and will be further discussed in Section~\ref{sec:k2-dense-voronoi}; for now we focus on $\Gsparse$.

We provide our BFS variant as follows. Note that $Q$ denotes a first-in first-out queue, and $D$ denotes the set of discovered vertices. We say that the BFS algorithm \emph{discovers} a vertex $w$ when $w$ is added to $D$.

\begin{figure}[!h]
\vspace{2pt}\noindent\fbox{\begin{minipage}{\dimexpr\textwidth-2\fboxsep-2\fboxrule\relax}
\textbf{BFS variant} of a search for centers starting at vertex $v$
\begin{compactitem}[\quad]
\item $Q.\textrm{enqueue}(v)$, $D.\textrm{add}(v)$
\item \textbf{while} $Q$ is not empty
\begin{compactitem}[\quad]
\item $u \leftarrow Q.\textrm{dequeue}$
\item probe for all neighbors $\Gamma(u,G)$ of $u$
\item \textbf{for each} $w \in \Gamma(u,G) \setminus D$ in the increasing order of $\ID$s
\begin{compactitem}[\quad]
\item $Q.\textrm{enqueue}(w)$, $D.\textrm{add}(w)$ $\quad\rhd$ $w$ is \emph{discovered}
\end{compactitem}
\end{compactitem}
\end{compactitem}
\end{minipage}}\vspace{4pt}
\caption[BFS variant for finding centers]{BFS variant for finding centers.}
\end{figure}

Denote by $D^k_L(v)$ the set of the first $L$ vertices discovered by the BFS variant, restricting to vertices at distance at most $k$ from $v$. (Equivalently speaking, if we adjust the BFS algorithm above so that it also terminates as soon as we have discovered $L$ vertices or dequeued a vertex at distance $k$ from $v$, then $D^k_L(v)$ would be the set $D$ upon termination.) Note that $D^k_L(v,G) \subseteq \Gamma^k(v,G)$, and the containment is strict when $\deg_k(v,G) > L$.

\paragraph{BFS probe complexity} Recall that each vertex elects itself as a center with probability $p_{\textrm{center}}=(c_{\textrm{center}}\log n)/L$. We choose a sufficiently large constant $c_{\textrm{center}}$ so that, by the hitting set argument, w.h.p., $D^k_L(v,G) \cap S \neq \emptyset$ for every $v$ with $|D^k_L(v,G)| = L$. That is, w.h.p., every vertex $v$ with $\deg_k(v,G) \geq L$ must be dense. Equivalently:
\begin{observation}\label{obs:sparsedeg}
W.h.p., for every sparse vertex $v$, $\deg_k(v,G) < L$.
\end{observation}
This observation leads to a subroutine for verifying whether a vertex $v$ is sparse or dense based on $D^k_L(v,G)$: 
\begin{claim}
$v$ is sparse if and only if both of the following holds: $|D^k_L(v,G)|<L$ and $D^k_L(v,G) \cap S = \emptyset$.
\end{claim}
\begin{proof} \textbf{(Sparse)} If $v$ is sparse ($\Gamma^k(v,G)\cap S = \emptyset$), then by Obs.~\ref{obs:sparsedeg}, $\deg_k(v,G) < L$, so $D^k_L(v,G)=\Gamma^k(v,G)$ and both conditions follow.
\textbf{(Dense)} If $v$ is dense ($\Gamma^k(v,G)\cap S \neq \emptyset$), we assume $|D^k_L(v,G)| < L$, then $D^k_L(v,G) = \Gamma^k(v,G)$ and hence $D^k_L(v,G) \cap S=\Gamma^k(v,G) \cap S \neq \emptyset$.
\end{proof}

To compute $D^k_L(v,G)$ we must discover (up to) $L$ distinct vertices. Recall that we always probe for all neighbors of a vertex at a time. Observe that for any positive integer $\ell$, among the neighbor sets of $\ell-1$ vertices in the same connected component of size at least $\ell$, at least one must necessarily contain an $\ell^\textrm{th}$ vertex from the component. Inductively, probing for all neighbors of $\ell-1$ vertices during the BFS algorithm must reveal at least $\ell$ vertices unless the entire component containing $v$ is exhausted. Hence, we conclude that we only need to probe for all neighbors of $L-1$ vertices during our BFS in order to compute $D^k_L(v,G)$, requiring $O(\Delta L)$ probes in total.

\paragraph{Local simulation of a distributed spanner algorithm}
We construct a $(2k-1)$-spanner $\Hsparse \subseteq \Esparse$ via a local simulation of a $k$-round distributed algorithm for constructing spanners on the subgraph $\Gsparse$. Since we also want the randomized algorithm to operate on $O(\log n)$-wise independence random bits, we will use the distributed construction of \citet{baswana07} with bounded independence \cite{Censor-HillelPS16}:

\begin{theorem}[From \cite{baswana07,Censor-HillelPS16}]
There exists a randomized $k$-round distributed algorithm for computing a $(2k-1)$-spanner $H$ with $O(kn^{1+1/k})$ edges for the unweighted input graph $G$. More specifically, for every $(u,v)\in H$, at the end of the $k$-round procedure, at least one of the endpoints $u$ or $v$ (but not necessarily both) has chosen to include $(u,v)$ in $H$. Moreover, this algorithm only requires $O(\log n)$-wise independence random bits.
\end{theorem}

For a query edge $(u,v)$, we first verify that at least one of $u$ or $v$ is sparse; otherwise we handle it later during the dense case. Without loss of generality, assume that $v$ is sparse. To simulate the distributed algorithm on $\Gsparse$ for vertex $v$, we first learn its $k^\textrm{th}$-neighborhood $\Gamma^k(v, G)$, and collect all the induced edges therein. We then verify every vertex in $\Gamma^k(v,G)$ whether it is dense or sparse, so that we can determine the edges that also appear in $\Esparse$, and simulate the distributed algorithm as if it is executed on $\Gsparse$ accordingly.

According to the description of the distributed algorithm's behavior, for a query edge $(u,v)$, we need to simulate this algorithm on both $u$ and $v$, requiring the knowledge of both $\Gamma^k(u, G)$ and $\Gamma^k(v, G)$. Since $v$ is sparse and $\Gamma^k(u, G) \subseteq \Gamma^{k+1}(v,G)$, we have $|\Gamma^k(u, G)| \leq |\Gamma^{k+1}(v,G)| \leq \Delta \cdot |\Gamma^{k}(v,G)| < \Delta L$ by Obs.~\ref{obs:sparsedeg}. So, we need $O(\Delta^2 L)$ \neighborP~probes to compute the subgraph of $G$ induced by $\Gamma^k(u, G)$ and $\Gamma^k(v, G)$. We must also test up to $O(\Delta L)$ vertices to determine whether they are sparse or not, so our simulation process requires $O(\Delta^2 L^2)$ probes in total. We conclude the analysis of our LCA for computing $\Hsparse$ as the following lemma; see Table~\ref{table:sparse-probes} for a summary of probe complexities.

\begin{lemma}[$\Hsparse$ properties and probe complexity]\label{lem:k2sparse}
For any stretch factor $k\geq1$, there exists an LCA that w.h.p., given an edge $(u,v)\in E$, decides whether $(u,v) \in \Hsparse$ using probe complexity $O(\Delta^2 L^2)$, where $\Hsparse$ is a $k$-spanner of $\Gsparse$ with $O(kn^{1+1/k})$ edges.
\end{lemma}

\begin{table*}
\centering
\renewcommand{\arraystretch}{1.5} 
\begin{tabular}{ l|l }
\textbf{Subroutine} & \textbf{Probe complexity} \\ \hline \hline
determine whether $v$ is a center & none \\ \hline
compute $D^k_L(v,G)$, and test whether $v \in \Vsparse$ or $v \in \Vdense$ & $O(\Delta L)$ \\\hline
for $(u,v) \in \Esparse$, compute $\Gamma^k(u,G)$ and $\Gamma^k(v,G)$ & $O(\Delta^2 L)$ \\ \hline
test $(u,v)\in\Esparse$ whether $(u,v) \in \Hsparse$ & $O(\Delta^2 L^2)$
\end{tabular}
\caption[Probe complexities for computing $\Hsparse$]{Probe complexities of various subroutines used for computing $\Hsparse$.}
\label{table:sparse-probes}
\end{table*}

\subsection{LCA for computing an $O(k^2)$-spanner $\Hdense$ for $\Edense$}

Recall that $\Vdense=V\setminus \Vsparse$ is the collection of dense vertices characterized as $\Gamma^{k}(v,G) \cap S\neq\emptyset$, and can be verified by computing $D^k_L(v,G)$ with $O(\Delta L)$ probes.
We will now take care of $\Edense=E(\Vdense,\Vdense)$ by constructing an $O(k^2)$-spanner $\Hdense \subseteq \Edense$ so that $H=\Hsparse \cup \Hdense$ becomes the desired spanner of $G$. To do so, we follow the general approach of \citet{LeviLenzen17} with several keys modifications along the way.
Table~\ref{table:dense-probes} keeps track of the probe complexities for various useful operations for constructing $\Hdense$.

In the following, we show how to partition the dense vertices into Voronoi cells, and connect vertices in each cell via a low-depth tree structure in Section~\ref{sec:k2-dense-voronoi}. We then show how to subdivide Voronoi cells into clusters of size $O(L)$ in Section~\ref{sec:k2-dense-cluster}, and discuss how we connect them into the desired spanner in Sections~\ref{sec:k2-dense-connect-overview}-\ref{sec:k2-dense-connect-proof}. We denote the set of spanner edges connecting vertices \emph{inside} Voronoi cells by $\HdenseI$, and edges connecting \emph{between} clusters by $\HdenseB$, so $\Hdense = \HdenseI \cup \HdenseB$.

\subsubsection{Partitioning of dense vertices into Voronoi cells}\label{sec:k2-dense-voronoi}

We partition the dense vertices into $|S|=O((n\log n)/L)=O(n^{2/3}\log n)$ \emph{Voronoi cells} with respect to centers $s_i \in S$, where each dense vertex $v$ chooses the first center $s_i$ that it discovers when executing the proposed BFS variant. We denote by $c(v)$ the center of $v$, and $\Vor(s)$ the Voronoi cell centered at $s$, consisting of all vertices that choose $s$ as its center. 

\paragraph{Order of vertex discovery in BFS} Clearly, the vertices are discovered in increasing distance from $v$. We claim that the distance ties are broken according to their lexicographically-first shortest path \emph{from} $v$ (with respect to vertex $\ID$s).\footnote{Between paths of the same length $d$, $\langle v_0, \ldots, v_d \rangle \prec \langle u_0, \ldots, u_d \rangle$ if, for the minimum index $i$ such that $v_i \neq u_i$, $\ID(v_i) < \ID(u_i)$.} More formally, let $\pi(v,u)$ denote the lexicographically-first shortest path from $v$ to $u$ in $G$, and $|\pi(v,u)|$ denote its length (namely $\dist(v,u,G)$, the number of edges in the shortest $v$-$u$ path). We claim that the BFS from $v$ discovers $u$ before $u'$ if either $|\pi(v,u)|<|\pi(v,u')|$, or $|\pi(v,u)|=|\pi(v,u')|$ and $\pi(v,u)\prec\pi(v,u')$: assuming the induction hypothesis that vertices at the same distance $d$ from $v$ are discovered (enqueued) in this lexicographical order, we dequeue them in the same order, then enqueue the neighbors of each vertex in the order of their $\ID$s, proving the hypothesis for distance $d+1$.

\paragraph{Connectedness of each Voronoi cell on $G$} To prove that every $\Vor(s_i)$ induces a connected component in $G$, consider a vertex $v$ and its shortest path $\pi(v,c(v)) = \langle v_0=v, v_1, \ldots, v_{d-1}, v_d=c(v)\rangle$: we show that \emph{all} vertices in this path are in $\Vor(s_i)$. Assume the contrary: let $u = v_i$ be the first vertex on $\pi(v,c(v))$ choosing a different center $c(u)$ via $\pi(u,c(u))= \langle v_i=u, v'_{i+1}, \ldots, v'_{d-1}, v'_d=c(u)\rangle$; note that $|\pi(u,c(u))|=d-i+1$ because there is no center in $\Gamma^{d-1}(v,G) \supset \Gamma^{d-i}(u,G)$. Then we have that $\langle v_i, v'_{i+1}, \ldots, v'_d\rangle \prec \langle v_i, v_{i+1}, \ldots, v_d\rangle$, yielding $\langle v_0=v, \ldots, v_i=u, v'_{i+1}, \ldots, v'_d=c(u)\rangle \prec \pi(v,c(v))$, a contradiction.

\paragraph{Construction of depth-$k$ trees spanning Voronoi cells} We straightforwardly connect each $v$ to its center $s=c(v)$ via the edges of $\pi(v, s)$. Observe that due to the lexicographic condition, the vertex after $v$ on $\pi(v,s)$ must be the vertex of the minimum $\ID$ in $\Gamma(v,G) \cap \Gamma^{|\pi(v,s)|-1}(s,G)$; that is, each vertex $v \in \Vor(s)$ has a fixed ``next vertex'' to reach $s$. Consequently, the union of edges in $\pi(v, s)$ for every $v \in \Vor(s)$ forms a \emph{tree} rooted at $s$, where every level $d$ contains vertices at distance exactly $d$ away from $s$.

Due to the resulting tree structure, we henceforth refer to the constructed subgraphs spanning the Voronoi cells as \emph{Voronoi trees}. The union of these trees forms the spanner edge set $\HdenseI$. As our BFS variant for finding a center terminates after exploring radius $k$, our Voronoi trees are also of depth at most $k$, or diameter at most $2k$, as desired. Lastly, by augmenting our proposed BFS variant to record the BFS tree edges, we can also retrieve the Voronoi tree path $\pi(v,s)$ using $O(\Delta L)$ probes. In particular, $(u,v)$ is a Voronoi tree edge if $u$ is on $\pi(v,c(v))$ or $v$ is on $\pi(u,c(u))$, implying the following lemma.

\begin{lemma} [$\HdenseI$ properties and probe complexity] \label{lem:k2denseI}
There exists a partition of dense vertices $v \in \Vdense$ into $O((n\log n)/L)$ Voronoi cells $\{\Vor(s)\}_{s\in S}$ according to their respective first-discovered centers $c(v)$ under the provided BFS variant. The set of edges $\HdenseI$, defined as the collection of lexicographically-first shortest paths $\pi(v,c(v))$, forms Voronoi trees, each of which spans its corresponding Voronoi cell and has diameter at most $2k$. Further, there exists an LCA that w.h.p., given an edge $(u,v)\in E$, decides whether $(u,v)\in\HdenseI$ using $O(\Delta L)$ probes.
\end{lemma}

\begin{table*}
\centering
\renewcommand{\arraystretch}{1.5} 
\resizebox{\textwidth}{!}{%
\begin{tabular}{ l|l }
\textbf{Subroutine} & \textbf{Probe complexity} \\ \hline \hline
verify that $v \in \Vdense$, choose $c(v)$, and compute $\pi(v,c(v))$ & \multirow{2}{*}{$O(\Delta L)$} \\ 
verify if a given edge $(u,v)$ is a Voronoi tree edge (i.e., $(u,v) \in \HdenseI$)&  \\ \hline
compute all children of $v$ in the Voronoi tree $T(c(v))$ & $O(\Delta^2 L)$ \\ \hline
verify whether $v$ is heavy or light, and determine $|T(v)|$ when $v$ is light & $O(\Delta^2 L^2)$ \\ \hline
compute the entire cluster containing $v$ & $O(\Delta^3 L^2)$\\ \hline
given an entire cluster $A$, compute $c(\partial A)$ and $E(A,\Vor(s))$ for any $s\in c(\partial A)$ & $O(\Delta^2 L^2)$\\ \hline
test $(u,v)\in\Edense$ whether $(u,v) \in \Hdense$ & $O(p\Delta^4 L^3 \log n)$
\end{tabular}
}
\caption[Probe complexities for computing $\Hdense$]{Probe complexities of various subroutines used for computing $\Hdense$. This table addresses $u,v \in \Vdense$, but these probe complexities do not assume that the LCA originally knows that $u$ and $v$ are dense.}
\label{table:dense-probes}
\end{table*}

\subsubsection{Refinement of the Voronoi cell partition into clusters}\label{sec:k2-dense-cluster}

We now further partition the Voronoi cells into \emph{clusters}, each of size $O(L)$. Our cluster structure is based on the construction of \cite{LeviLenzen17} but has two major differences. First, whereas in \cite{LeviLenzen17} the Voronoi cells are partitioned into $\Theta(\Delta n/L)$ clusters, in our algorithm we need the number of clusters to be independent of $\Delta$, and more specifically bounded by $O((n\log n)/L)$. Second, unlike the clusters in \cite{LeviLenzen17} that are always connected in $G$, each of our clusters may not necessarily induce a connected subgraph of $G$; they are still connected in the spanner via $\HdenseI$, namely by the Voronoi tree of diameter at most $2k$.

\paragraph{Refinement of Voronoi cells into clusters} For $s\in S$, let $T(s)$ denote the Voronoi tree spanning $\Vor(s)$. We extend this notation for non-centers, so that $T(v)\subseteq T(s)$ denote the subtree of $T(s)$ rooted at $v \in \Vor(s)$. For every $v \in \Vor(s)$, let $p(v)$ denote the \emph{parent} of $v$ in $T(s)$, and $|T(v)|$ be the number of vertices in the subtree. We define heavy and light vertices as follows.

\begin{definition}[Heavy and light vertices]
A dense vertex $v$ is \emph{heavy} if $|T(v)|>L$ and otherwise, it is \emph{light}.
\end{definition}

We are now ready to define the cluster of $v \in \Vor(s)$ using the heavy and light classification.
\begin{compactenum}[(a)]
\item \textbf{$s$ is light:} That is, the Voronoi cell containing $v$, $\Vor(s)$, contains at most $L$ vertices. Then, all vertices in $\Vor(s)$ form the cluster centered at $s$. 
\item \textbf{$v$ is heavy:} Then the cluster of $v$ is the singleton cluster $\{v\}$.
\item \textbf{$s$ is heavy and $v$ is light:} Let $u$ be the first heavy vertex on $\pi(v,s)$, and $W=\{w: p(w)=u \textrm{ and } w \textrm{ is light}\}$ be the set of $u$'s \emph{light} children on the Voronoi tree. Consistently ordering the vertices $W=\{w_1, \ldots, w_\ell\}$ (e.g., according to the adjacency-list order from $u$), we iterate through these $w_i$'s, grouping $T(w_i)$'s into clusters of sizes between $L$ and $2L$; the last remaining cluster is allowed to have size strictly less than $L$. See Figure~\ref{fig:clusters} for an illustration of this rule.
\end{compactenum}
Clearly each cluster contains at most $2L =O(L)$ vertices, and any pair of vertices in the same cluster has a path of length at most $2k$ on $T(s)$ because they belong to the same Voronoi cell. Next, we show that the number of clusters resulting from this refinement is not asymptotically larger than the number of Voronoi cells.

\begin{figure}[!h]
  \begin{center}
    \includegraphics[width=0.6\textwidth]{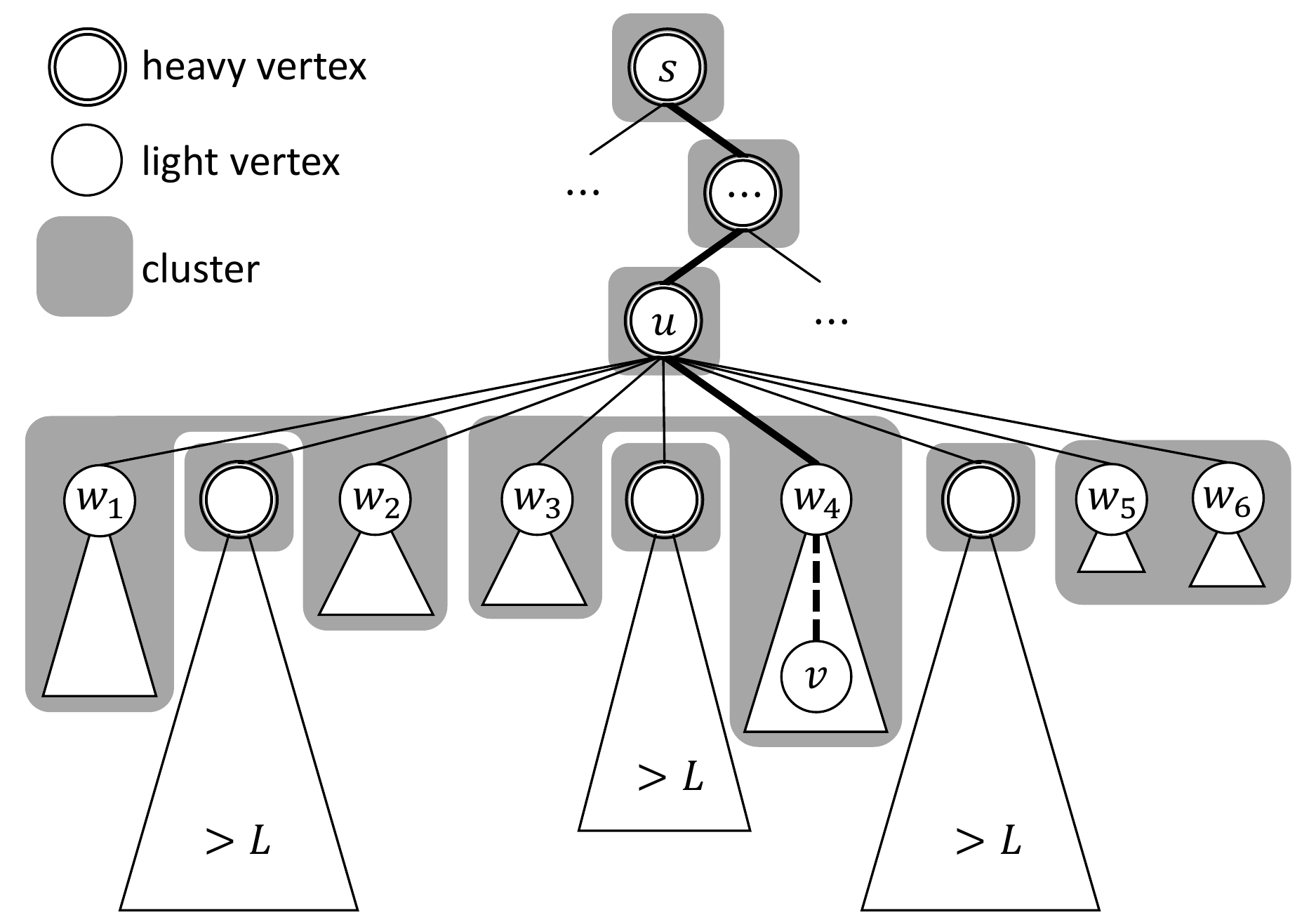}
  \end{center}
  \caption[Illustration for cluster partitioning rule (c)]{Illustration for cluster partitioning rule (c). The interesting part of the Voronoi tree $T(s)$ is shown: heavy vertices are denoted with double borderlines, and thick edges are edges of $\pi(v, s)$. Shaded areas are clusters: observe that heavy vertices form singleton clusters, while many clusters do not induce a connected subgraph of $T(s)$. In this example, $u$ is the first heavy ancestor of $v$, so we compute all light children $W=\{w_1, \ldots, w_6\}$ of $u$, along with their subtree sizes $|T(w_i)|$'s. (For the heavy children, it suffices to only verify that they are heavy.) We group these $T(w_i)$'s into clusters of sizes in $[L, 2L]$, except possibly for the remainder cluster ($T(w_5)\cup T(w_6)$ in this case). Here, $v$'s cluster is $T(w_3) \cup T(w_4)$.}
\label{fig:clusters}
\end{figure}

\begin{claim}\label{obs:nclusters}
The number of clusters is $O((n\log n)/L)=O(n^{2/3}\log n)$.
\end{claim}
\begin{proof}
Recall that there are $|S| = O((n\log n)/L)$ Voronoi cells: this bounds the number of clusters of type (a). Observe that in any fixed level, among \emph{all} Voronoi trees, there can be at most $n/L$ heavy vertices because these heavy vertices' subtrees are disjoint. Since the Voronoi tree has depth $k$, there are at most $kn/L$ heavy vertices, bounding the number of clusters of type (b). 

We only subdivide the subtrees of heavy vertices into clusters, and within each such subtree, all clusters, except for at most one, have size at least $L$. Hence, there can be up to $n/L$ clusters of size at least $L$, and $kn/L$ clusters of smaller sizes (one for each heavy parent), establishing the bound for clusters of type (c). Thus, there are in total at most $O((n\log n)/L) + (2k+1)n/L = O((n\log n)/L)$ clusters (as we only consider $k = O(\log n)$).
\end{proof}

\paragraph{Probe complexity for identifying a vertex's cluster} Recall that via our BFS variant we can find the center $s$ and the path $\pi(v, s)$ for a dense vertex $v \in \Vor(s)$ using $O(\Delta L)$ probes. We begin by establishing the probe complexity for deciding whether $v$ is light or heavy. Observe that we can find all children of $v$ on $T(s)$ using $O(\Delta^2 L)$ probes: run the BFS on all neighbors of $v$, then any $w$ with center $s$ such that $\pi(w,s)$ passes through $v$ is a child of $v$. Using this subroutine, we traverse the subtree $T(v)$ to compute $|T(v)|$ if $v$ is light, or stop after $L+1$ and declare that $v$ is heavy. Since $O(L)$ vertices are investigated, the probe complexity for this process is $O(L)\cdot O(\Delta^2 L) = O(\Delta^2 L^2)$.

We can then compute $v$'s cluster as follows. If $v$ is heavy then we have $\{v\}$ as the cluster of type (b). Otherwise, we follow the path $\pi(v, s)$ up the Voronoi tree, one vertex at a time, and check each vertex's subtree size until we reach some heavy ancestor $u$ of $v$; if there is no such $u$ then the entire $\Vor(s)$ is the cluster of type (a). During this process of traversing up the Voronoi tree, we also record every computed subtree size, so that we do not need to revisit any subtree. Hence, finding the first heavy ancestor $u$ essentially only requires visiting $O(L)$ descendants of $u$, which only takes $O(\Delta^2 L^2)$ probes. Once we detect $u$, we check each of $u$'s children if it is light, and compute its subtree size correspondingly. Using this information, we determine all subtrees that form the cluster of type (c) containing $v$, as desired. This last case dominates the probe complexity: since we must check whether each of $u$'s children is heavy or light, our algorithm require $\Delta \cdot O(\Delta^2 L^2) = O(\Delta^3 L^2)$ probes to identify $v$'s entire cluster. The following lemma concludes the properties of the cluster partitioning of dense vertices.

\begin{lemma}[Probe complexity for computing clusters]\label{lem:compute-cluster}
There exists a refinement of the Voronoi cell partition into $O((n \log n)/L)$ clusters of size $O(L)$ each. Further, there exists an LCA that w.h.p., given a dense vertex, compute all vertices in the cluster containing $v$ using $O(\Delta^3 L^2)$ probes.
\end{lemma}

\subsubsection{Overview: connecting Voronoi cells}\label{sec:k2-dense-connect-overview}

\paragraph{The supergraph intuition} To establish some intuition for connecting the Voronoi cells while maintaining a low stretch factor, let us imagine constructing an LCA for a \emph{supergraph}, where each of the $|S| = \widetilde{O}(n/L) = \widetilde{O}(n^{2/3})$ Voronoi cells is a supervertex, and all edges between the same pair of Voronoi cells are merged into a single superedge. Leveraging the classic clustering approach, to compute a spanner on this supergraph, we mark each supervertex independently with probability $p=n^{-1/3}$, so roughly $\widetilde{O}(pn/L) = \widetilde{O}(n^{1/3})$ supervertices are marked. These marked supervertices now act as the centers in this supergraph.

In the constructed spanner, we keep superedges between adjacent Voronoi cells according to the following three rules. Rule (1): we keep all superedges incident to a marked supervertex. There are $\widetilde{O}(n^{2/3})$ supervertices in total, and $\widetilde{O}(n^{1/3})$ supervertices are marked, contributing to $\widetilde{O}(n)$ total superedges. Rule (2): we can also keep incident superedges of supervertices without any marked neighbors: if they had more than $\widetilde{O}(n^{1/3})$ neighboring Voronoi cells, then w.h.p., one of them would have been marked. 
Lastly, rule (3): for each (not necessarily adjacent) pair of a supervertex $\Vn{a}$ and a \emph{marked} supervertex $\Vn{c}$, we keep a superedge from $\Vn{a}$ to a \emph{single} common neighbor $\Vn{b^*} \in \Gamma(\Vn{a})\cap\Gamma(\Vn{c})$ -- by consistently choosing $\Vn{b^*}$ with the lowest $\ID$, for instance. The number of added superedges is $\widetilde{O}(n)$ via the same analysis as that of rule (1).

We claim that connectivity is preserved: consider an omitted superedge $(\Vn{a}, \Vn{b})$. Since rule (2) does not keep $(\Vn{a}, \Vn{b})$, $\Vn{b}$ has some marked neighbor $\Vn{c}$. By rule (3), there exists some $\Vn{b^*} \in \Gamma(\Vn{a})\cap\Gamma(\Vn{c})$ with lower $\ID$ than $\Vn{b}$, such that $(\Vn{a}, \Vn{b^*})$ is kept by the LCA. Recall that $\Vn{c}$ is marked, so combining with rule (1), the spanner path $\langle \Vn{a}, \Vn{b^*}, \Vn{c}, \Vn{b}\rangle$ connects $\Vn{a}$ and $\Vn{b}$, as desired. Thus, an LCA, given a query $(\Vn{a}, \Vn{b})$, keeps this superedge if there exists a supervertex $\Vn{c}\in \Gamma(\Vn{b})$ where $\Vn{b}$ has the minimum $\ID$ among $\Gamma(\Vn{a})\cap\Gamma(\Vn{c})$, producing a $3$-spanner of the supergraph with $\widetilde{O}(n)$ superedges. 

However, such a supergraph-level approach cannot be implemented efficiently under the cluster refinement in the original input graph. Recall the original graph before the Voronoi cell contraction: the LCA is only given a vertex (query edge's endpoint) in the Voronoi cell, and we cannot afford to enumerate all vertices in the \emph{entire} Voronoi cell (supervertex $\Vn{b}$) and identify \emph{all} of its neighboring Voronoi cells (supervertex $\Vn{c}$) -- finding the Voronoi cell $\Vn{b^*}$ of minimum $\ID$ is outright impossible in sub-linear probes. Nonetheless, we construct an LCA based on this approach despite incomplete information of the supergraph.

\paragraph{Local implementation based on clusters} Employing the developed cluster refinement, as we mark a Voronoi cell, we also mark the clusters therein. We will show that the number of clusters (resp., marked clusters), do not significantly increase from the number of Voronoi cells (resp., marked Voronoi cells); hence, we may still add an edge from every cluster that is (1) marked, or (2) not adjacent to any marked clusters, to all adjacent Voronoi cells, modularly imitating the corresponding supergraph rules while still using $\widetilde{O}(n)$ edges. Nonetheless, attempting to implement rule (3) poses a problem because the LCA can only see the \emph{clusters} containing the query edge's endpoints (while keeping the desired probe complexity). From them, we can only find out the Voronoi cells neighboring these clusters -- not all Voronoi cells neighboring to the current Voronoi cell may be visible to the LCA. Due to this limitation, we cannot implement rule (3) which requires knowing \emph{all} of $\Vn{b}$ and $\Vn{c}$'s neighboring Voronoi cells.

To resolve this problem, \cite{LeviLenzen17} observes that the desired connectivity is still preserved if the LCA implements a variation of rule (3) that only checks the neighboring Voronoi cells of the queried cluster in $\Vn{b}$ and a canonical cluster in $\Vn{c}$. Recall that the LCA must answer ``is the superedge $(\Vn{a}, \Vn{b})$ in the spanner?'' We need to show that $\Vn{a}$ and $\Vn{b}$ are connected under this rule, so if the LCA keeps $(\Vn{a}, \Vn{b})$ then we are done. Otherwise $(\Vn{a}, \Vn{b})$ is omitted, which implies that there exists a marked Voronoi cell $\Vn{c}$ and a Voronoi cell $\Vn{b'} \in \Gamma(\Vn{a}) \cap \Gamma(\Vn{c})$ with $\ID(\Vn{b'})<\ID(\Vn{b})$, such that there exists a path $\langle \Vn{b}, \Vn{c}, \Vn{b'}\rangle$ in the spanner thanks to rule (1). Hence, it suffices to show that $\Vn{a}$ and $\Vn{b'}$ are connected in the spanner. Since the supergraph contains the superedge $(\Vn{a}, \Vn{b'})$ (because $\Vn{b'} \in \Gamma(\Vn{a})$), we will inductively rely on how the LCA ensures connectivity between $\Vn{a}$ and $\Vn{b'}$ when it handles the query $(\Vn{a}, \Vn{b'})$.

So far, we have only managed to defer the original burden of proving the connectivity between $\Vn{a}$ and $\Vn{b}$ to the LCA's answer to the question ``is the superedge $(\Vn{a}, \Vn{b'})$ in the spanner?'' Again, even if $\Vn{b'}$ indeed has the minimum $\ID$ among $\Gamma(\Vn{a}) \cap \Gamma(\Vn{c})$, the LCA may not perceive this fact when it cannot see all of $\Vn{b'}$'s neighboring Voronoi cells, notably $\Vn{c}$. Still, we have progress: the Voronoi cell $\Vn{b'}$ in question has a \emph{lower} $\ID$~than $\Vn{b}$. Thus we may repeat this same argument inductively on the $\ID$ of $\Vn{a}$'s neighbor, which strictly decreases at each step -- this argument will eventually terminate (albeit possibly in as many as $\Theta(|S|)$ steps), establishing the desired connectivity guarantee. Moreover, \cite{LeviLenzen17} enhances the LCA further by assigning random \emph{ranks} on the Voronoi cells instead of using $\ID$s directly, showing that $a$'s neighbor's rank is halved at each inductive step in expectation, so the stretch of the constructed spanner (on this supergraph) is, w.h.p., $O(\log n)$.

\begin{figure}[!h]
  \begin{center}
    \includegraphics[width=0.6\textwidth]{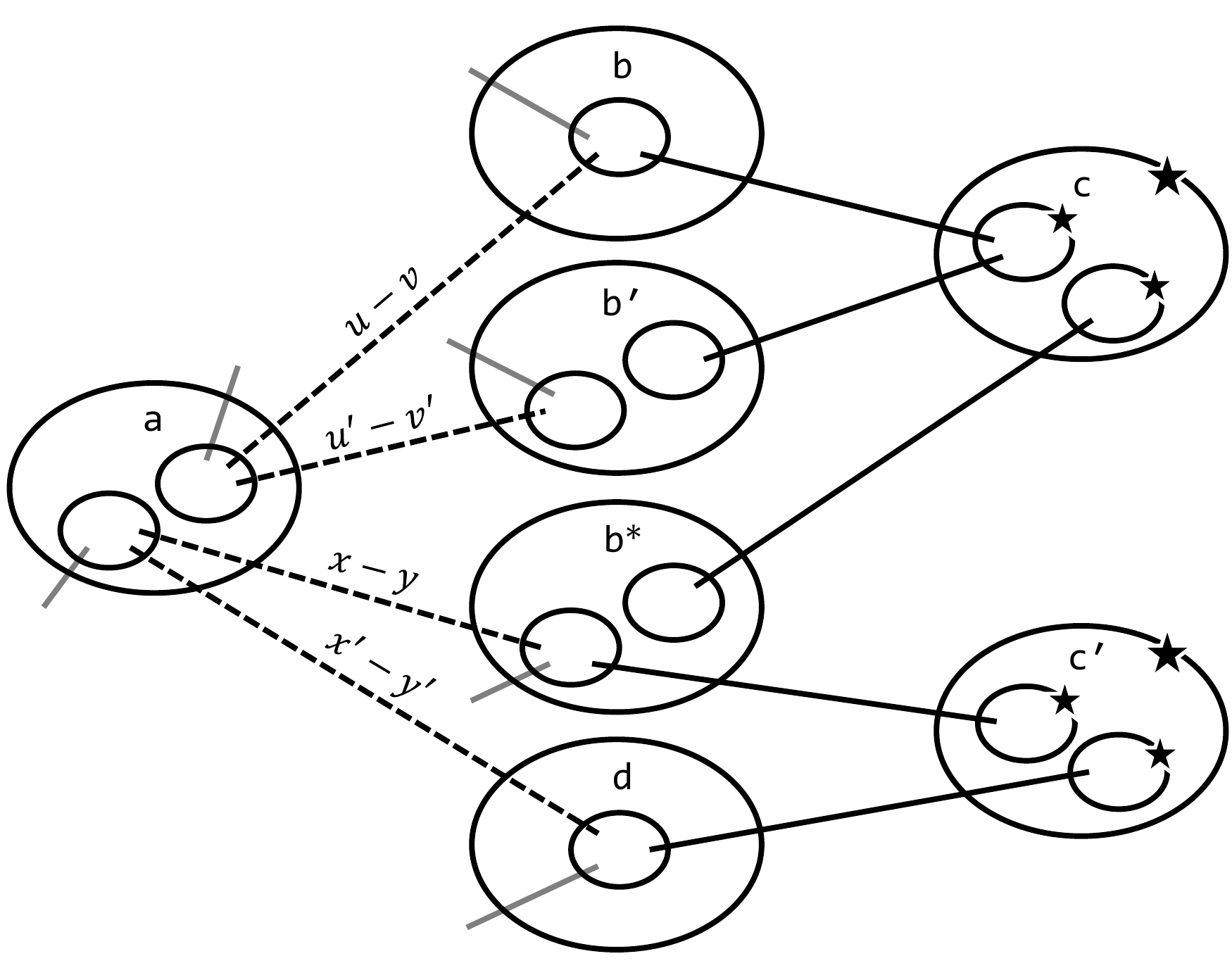}
  \end{center}
  \caption[Illustration accompanying the example of clusters connection rule (3)]{Illustration accompanying the example of clusters connection rule (3): Large and small ovals denote Voronoi cells and clusters; marked Voronoi cells and clusters therein are marked with stars. Dashed edges are query edges we consider in the example -- their labels shows the names of the endpoints (vertices) inside the clusters they connect.}
\label{fig:vorconnect}
\end{figure}

\paragraph{Illustrated example} \label{par:example} Consider Figure~\ref{fig:vorconnect}. All solid edges are added by rule (1). We focus on rule (3), so to prevent an application of rule (2), we add solid grey lines to indicate that all incident clusters are adjacent to some marked Voronoi cells. Let $\ID(\Vn{b})>\ID(\Vn{b'})>\ID(\Vn{b^*})>\ID(\Vn{d})$. The ``supergraph-level'' \emph{Voronoi cell} connection rule (3) would add $(x,y)$ and $(x',y')$ because $\Vn{b^*}$ and $\Vn{d}$ are Voronoi cells of minimum $\ID$s in $\Gamma(\Vn{a})\cap\Gamma(\Vn{c})$ and $\Gamma(\Vn{a})\cap\Gamma(\Vn{c'})$, respectively. Instead, consider now the \emph{cluster} connection rule (3).
\begin{compactitem}
\item Query edge $(x',y')$: The LCA applies the cluster connection rule (3) w.r.t.~$\Vn{c'}$ and keeps $(x',y')$.
\item Query edge $(u,v)$: This edge may be omitted because the LCA finds the Voronoi cell $\Vn{b'}$ also adjacent to $\Vn{c}$ with lower $\ID$ than $\Vn{b}$, so rule (3) w.r.t.~$\Vn{c}$ does not keep this edge. The inductive argument turns to consider $(u',v')$ (not $(x,y)$, even if $\Vn{b^*}$ actually has the lowest $\ID$ among $\Gamma(\Vn{a})\cap\Gamma(\Vn{c})$).
\item Query edge $(u',v')$: This edge may also be omitted because the LCA cannot reach $\Vn{c}$ from $v'$ despite the fact that $\Vn{c} \in \Gamma(\Vn{b})$; hence it cannot apply rule (3) w.r.t.~$\Vn{c}$. Note that $(u',v')$ is engaged in another application of rule (3) w.r.t.~the (undepicted) other marked endpoint of the grey edge incident to $v'$'s cluster -- $(u',v')$ may indeed be kept by this application, and if not, the inductive argument will continue.
\item Query edge $(x,y)$: This edge is kept, but \emph{not} because $\Vn{b^*}$ is the minimum-$\ID$ Voronoi cell of $\Gamma(\Vn{a})\cap\Gamma(\Vn{c})$: the LCA exploring the graph from $y$ could not have found $\Vn{c}$. Instead, it finds $\Vn{c'}$, but still cannot find $\Vn{d}$. Apparently, $\Vn{b^*}$ becomes the Voronoi cell of minimum $\ID$ among $\Gamma(\Vn{a})\cap\Gamma(\Vn{c'})$ that it actually finds (here, the only one, in fact). Hence, the LCA applies rule (3) w.r.t.~$\Vn{c'}$ and keeps $(x,y)$.
\end{compactitem}

\paragraph{Reducing the stretch factor} Unlike the scenario of \cite{LeviLenzen17}, we aim for an $O(k^2)$-spanner of size $\widetilde{O}(n^{1+1/k})$ (in the original graph); in particular, we are allowed an extra factor of $\widetilde{O}(n^{1/k})$ in the number of edges. So, between each pair of a cluster (in Voronoi cell $\Vn{a}$) and a marked cluster (in Voronoi cell $\Vn{c}$), and we add edges from the cluster in $\Vn{a}$ to $\widetilde{\Theta}(n^{1+1/k})$ lowest-rank Voronoi cells $\Vn{b}$, instead of just the lowest-rank one. This adjustment reduces the ranks in the inductive argument much more rapidly: w.h.p., the argument terminates in only $O(k)$ steps, yielding an $O(k)$-spanner on this supergraph. Since each Voronoi cell has diameter at most $2k$, as we expand back our supervertices into Voronoi cells, we obtain the desired $O(k^2)$ stretch factor. 

\subsubsection{Implementation details and probe complexity analysis}\label{sec:k2-dense-connect-algo}

\paragraph{Marked Voronoi cells and clusters} Recall that we randomly choose a set $S$ of $O((n \log n)/L)$ centers, and mark each Voronoi cell \emph{center} independently with probability $p=1/L=n^{-1/3}$. For each marked center $s_i$, we also mark \emph{all} the clusters in $\Vor(s_i)$. We claim that the number of marked clusters is not significantly more than the number of marked centers.

\begin{claim}
The number of marked clusters is $O((pn \log^2 n)/L) = O(n^{1/3} \log^2 n)$.
\end{claim}
\begin{proof}
Since there are $O(\frac{n \log n}{L})$ clusters, then for any value $x>0$, there are $O(\frac{n \log n}{xL})$ Voronoi cells with $t\in [x,2x]$ clusters. So, we have at most $O(\frac{pn \log n}{xL})$ marked Voronoi cells with at most $2x$ clusters, yielding $O(\frac{pn \log n}{L})$ such clusters. Applying the argument for $O(\log n)$ different values of $x$ yields $O(\frac{pn \log^2 n}{L})$ total marked clusters.
\end{proof}

\paragraph{Random ranks} We assign each center $s \in S$ an independent random \emph{rank} $r(s) \in [0, 1)$ (e.g., a random hash function applied to their $\ID$s): these random ranks implicitly impose a random ordering of the centers. We sometimes refer to the rank of a Voronoi cell's center simply as the rank of that Voronoi cell. We remark that in Section~\ref{sec:boundedindepsparse}, we will show that $\Theta(\log n)$-independence random bits suffice for our purpose of choosing centers and assigning random ranks: our algorithm can be implemented with $O(\log^2 n)$ random bits.

\paragraph{Adjacent clusters and Voronoi cells} The following definitions are as in \cite{LeviLenzen17}. We say that clusters $A$ and $B$ are \emph{adjacent} if there exists $u \in A$ and $v \in B$ which are neighbors. In the same manner, cluster $A$ is \emph{adjacent} to $\Vor(s)$ if there exists $B \in \Vor(s)$ such that $A$ and $B$ are adjacent. For a cluster $A$, let $\Vor(A)$ denote the Voronoi cell containing $A$. Define the \emph{adjacent centers} of a cluster $A$ as $c(\partial A)=\{c(v): \Gamma(v) \cap A \neq \emptyset\}\setminus \{c(A)\}$. Roughly speaking, this is a partial collection of neighboring Voronoi cell centers of $\Vor(A)$, restricted to those visible to the LCA from $A$.

\paragraph{Connecting clusters and Voronoi cells} By ``connecting'' two adjacent subsets of vertices $A$ and $B$, we refer to the process of adding the edge of minimum $\ID$ in $E(A,B)$ to $\HdenseB$, where the $\ID$ of an edge $(u,v) \in E(A,B)$ is given by $(\ID(u),\ID(v))$. The comparison is lexicographic: first compare against $\ID(u)$, break ties with $\ID(v)$.

For every marked cluster $C$, define the \emph{cluster of clusters} of $C$, denoted by $\mathcal{C}(C)$, as the set of all clusters consisting of $C$ and all other clusters which are adjacent to $C$. A cluster $B \in \mathcal{C}(C)$ is \emph{participating} in $\mathcal{C}(C)$ if the edge of minimum $\ID$ in $E(B,\Vor(C))$ also belongs to $E(B,C)$. That is, if we want to connect the cluster $B$ to a certain marked Voronoi cell by choosing the edge $(u,v)$ of minimum $\ID$ (where $u \in B$ and $v$ is in that Voronoi cell), then ``$B$ is participating in $\mathcal{C}(C)$'' means that, $C$ is the (unique) cluster in the Voronoi cell containing $v$.

\paragraph{Constructing $\HdenseB$}
Adjacent clusters are connected in $\Hdense$ using the following rules, where $A$ and $B$ denote the clusters containing the two respective endpoints of the query edges $(u,v)$. It suffices to apply these rules when $u$ and $v$ belong to different Voronoi cells, $c(u)\neq c(v)$; otherwise $\HdenseI$ spans them already. Note that these conditions as written are not symmetric: we must also verify them with the roles of the $u$ and $v$ (e.g., $A$ and $B$) switched.

\begin{figure}[!h]
\vspace{5pt}\noindent\fbox{\begin{minipage}{\dimexpr\textwidth-2\fboxsep-2\fboxrule\relax}
\textbf{Global construction of $\HdenseB$ for edges between clusters.}
\begin{compactenum}[(1)]\label{desc:edges}
\item Every marked cluster is connected to each of its adjacent clusters. 
\item Each cluster $B$ that is not participating in any cluster-of-clusters (i.e., no cell adjacent to $B$ is marked), is connected to each of its adjacent \emph{Voronoi cells}.
\item For each pair of cluster $A$ and marked cluster $C$, consider the centers of clusters adjacent to both $A$ and $C$, namely $c(\partial A)\cap c(\partial C)$. If the rank $r(s)$ of the center $s \in c(\partial A)\cap c(\partial C)$ is among the $q = \Theta(n^{1/k}\log n)$ lowest ranks of centers in $c(\partial A)\cap c(\partial C)$, then $A$ is connected to $\Vor(s)$.
\end{compactenum}
\end{minipage}}\vspace{4pt}
\caption[Procedure for the global construction of $\HdenseB$]{Procedure for the global construction of $\HdenseB$.}\end{figure}

\paragraph{The local algorithm and its probe complexity} We now describe the local algorithm that decides whether a query edge $(u,v)\in \HdenseB$. Using the subroutines constructed so far, assume that the LCA has verified that $(u,v)\in\Edense$, identified their centers $c(u)\neq c(v)$, and computed the entire respective clusters $A$ and $B$, using $O(\Delta^3 L^2)$ probes according to Lemma~\ref{lem:compute-cluster}. We then verify the global rules of $\HdenseB$ in a local fashion, answering \YES~indicating that $(u,v)\in\HdenseB$ if any of the following condition holds.

\begin{figure}[!h]
\vspace{5pt}\noindent\fbox{\begin{minipage}{\dimexpr\textwidth-2\fboxsep-2\fboxrule\relax}
\textbf{Local construction of $\Hdense$ for edges between clusters.}
\begin{compactenum}[(1)]
\label{desc:whenyes}
\item
$A$ is a marked cluster and $(u,v)$ has the minimum edge $\ID$ amongst the edges in $E(A,B)$.
\item
$B$ is not adjacent to any of the marked clusters, and $(u,v)$ has the minimum edge $\ID$ among all edges in $E(B,\Vor(A))$. 
\item
There exists a marked cluster $C$ such that all of the following holds:
\begin{compactitem}
\item
$B$ is participating in $\mathcal{C}(C)$,
\item
The rank of $c(B)$ is amongst the $q=\Theta(n^{1/k}\log n)$ lowest ranks in $c(\partial A) \cap c(\partial C)$,
\item
The edge $(u,v)$ has the minimum $\ID$ among all the edges in $E(A,\Vor(B))$.
\end{compactitem}
\end{compactenum}
\end{minipage}}\vspace{4pt}
\caption[Procedure for the local construction of $\HdenseB$]{Procedure for the local construction of $\HdenseB$.}
\end{figure}

As we have already computed the entire clusters $A$ and $B$, we may verify condition (1) by checking all incident edges of $A$ for those with the other endpoints in $B$. For condition (2), we compute the set $c(\partial A)$ of Voronoi cell centers $c(w)$ for neighboring vertices $w$ of $A$; note that $|c(\partial A)|\leq \Delta L$. Then we check whether any of them is marked using $O(\Delta L)$ probes each. The edge of minimum $\ID$~in $E(A,\Vor(B))$ is among these $O(\Delta L)$ edges incident to $A$, allowing us to check whether $(u,v)=E(A,\Vor(B))$ as well. Overall condition (2) can be verified with $O(\Delta^2 L^2)$ probes.

For condition (3), we instead consider the neighboring vertices of $B$ and compute their centers. During the process, we also keep track of the edge of minimum $\ID$~in $E(B,\Vor(s_i))$ of each encountered \emph{marked} center $s_i$. There are up to $\Delta L$ neighboring Voronoi cells of $B$ in total, but w.h.p., only $O(p \cdot \Delta L \cdot \log n)$ of them are marked. For each marked $\Vor(s_i)$, starting from the recorded endpoint in there, we compute the entire cluster $C_i$ such that $B$ is participating $\mathcal{C}(C_i)$ using $O(\Delta^3 L^2)$ probes. Then, we compute the centers' $\ID$s of all neighboring vertices of $C_i$, namely $c(\partial C)$, spending another $O(\Delta^2 L^2)$ probes for each $C_i$. Combining with $c(\partial A)$ computed earlier, we can deduce if the rank of $c(B)$ is sufficiently low that $E(A,\Vor(B))$ must be added. In total, we require $O(p \Delta L \log n) \cdot (O(\Delta^3 L^2)+O(\Delta^2 L^2)) = O(p\Delta^4 L^3 \log n)$ probes, as desired:

\begin{lemma}[$\HdenseB$ probe complexity]\label{lem:k2denseprobe}
There exists an LCA that w.h.p., given an edge $(u,v)\in E$, decides whether $(u,v) \in \HdenseB$ using probe complexity $O(p \Delta^4 L^3 \log n)$, where $\HdenseB$ is as defined in Section~\ref{sec:k2-dense-connect-algo}.
\end{lemma}

\subsubsection{Proof of connectivity, stretch, and size analysis}\label{sec:k2-dense-connect-proof}

\paragraph{Stretch and size analysis of $\Hdense$}
Denote by $G_{\Vor}$ the supergraph obtained from $G$ by merging vertices within each Voronoi tree into a supervertex (e.g., by contracting $\HdenseI$), and by $H_{\Vor}$ its subgraph obtained by applying the same operation in the spanner $\Hdense$ (e.g., the same edges as $\HdenseB$ but joining corresponding supervertices instead). Since we add strictly more edges than the algorithm of \cite{LeviLenzen17} does, the connectivity follows by the exact same argument (see Lemma 4 of \cite{LeviLenzen17}); for completeness, we provide it here (with only slightly modifications). See Figure~\ref{fig:spsp} for an illustration.
\begin{lemma}[Connectivity by $\HdenseB$]\label{lem:connected}
$H_{\Vor}$ preserves the connectivity of the Voronoi cells: if $\Vor$ and $\Vor_0$ are connected in $G_{\Vor}$, they remain connected in $H_{\Vor}$.
\end{lemma}

\begin{figure}[!h]
  \begin{center}
    \includegraphics[width=0.5\textwidth]{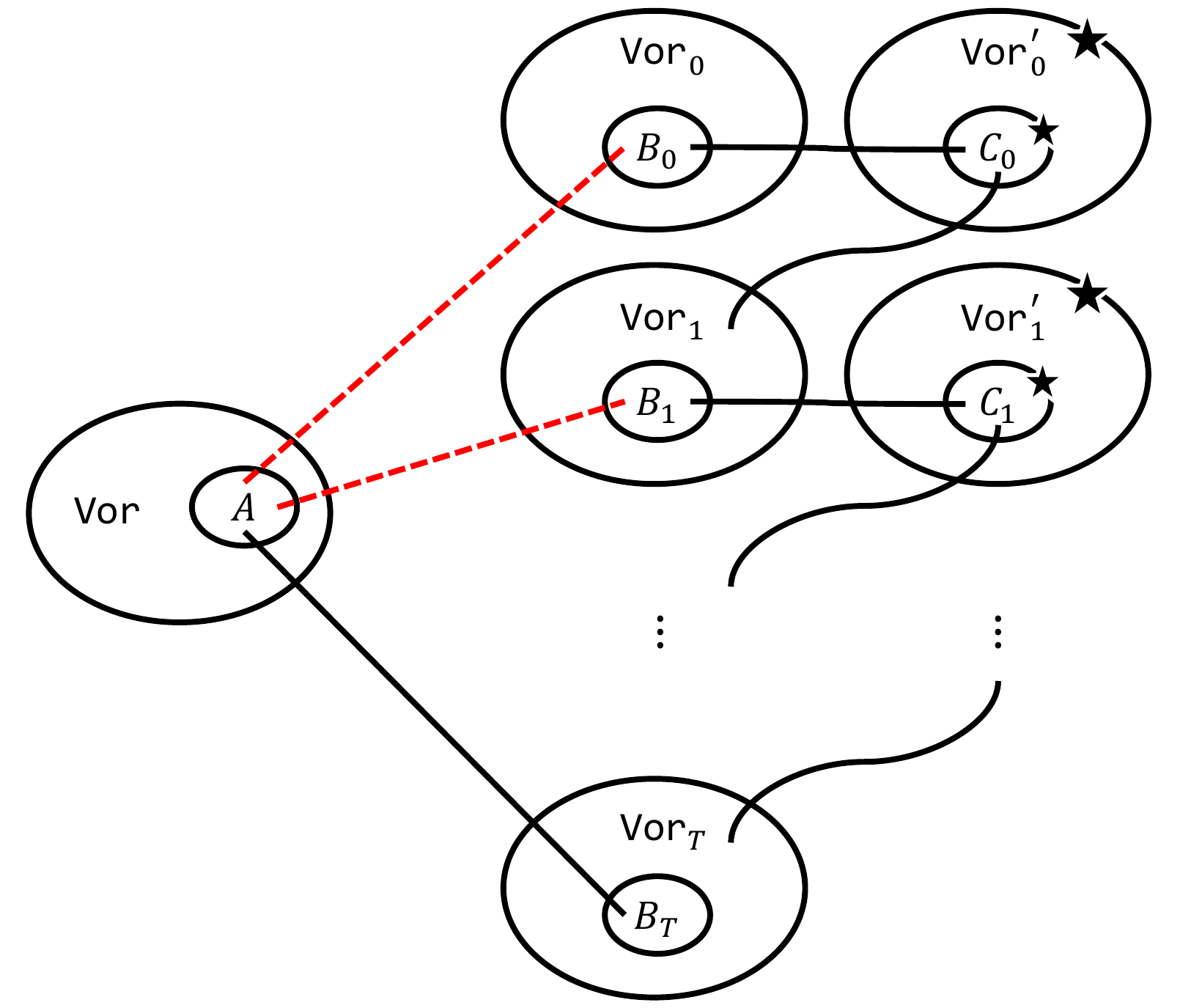}
  \end{center}
  \caption[Illustration for the proof of connectivity and stretch for $\Hdense$]{Illustration for the proof of connectivity and stretch for $\Hdense$: Dashed red edges show the edges of interest at each inductive step, where the top one joining clusters $A$ and $B_0$ represents the original query. Solid black edges show the path of length $2T+1 = O(k)$ in $H_{\Vor}$ between $\Vor$ and $\Vor_0$.}
\label{fig:spsp}
\end{figure}

\begin{proof}
Consider clusters $A \subseteq \Vor$ and $B=B_0 \subseteq \Vor_0$ such that the edge $e$ of minimum $\ID$ in $E(\Vor,\Vor_0)$ is in $E(A,B_0)$. If $B_0$ is not adjacent to any marked cell, then by condition (2) there is an edge between $\Vor$ and $\Vor_0$ in $H_{\Vor}$. Hence, we assume that $B_0$ is adjacent to a marked cell $\Vor'$. Let $C_0 \subseteq \Vor'_0$ be the cluster such that $B_0$ is participating in $\mathcal{C}(C_0)$. 

Let $s_0$ be the center of $\Vor_0$. If the rank $r(s_0)$ is among the $q$ lowest ranks of the centers $c(\partial A)\cap c(\partial C_0)$, then $e$ is added to $\HdenseB$ by condition (3). Otherwise, $\Vor_0$ is connected to $\Vor'_0$ in $H_{\Vor}$ as the edge of minimum $\ID$ in $E(B,C_0)$ is added to $\HdenseB$ by condition (1), since $C_0 \subseteq \Vor'_0$ is marked. Let $\Vor_1$ be the cell whose center has the minimum rank in $c(\partial A)\cap c(\partial C_0)$, and let $B_1 \subseteq \Vor_1$ be the cluster such that the edge of minimum $\ID$ in $E(A,\Vor_1)$ is in $E(A,B_1)$. Again by condition (1), $\Vor_1$ is also connected to $\Vor'_0$ in $H_{\Vor}$.

At that point, it suffices to show that $\Vor$ is connected to $\Vor_1$ in $H_{\Vor}$, where the rank of $\Vor_1$ is strictly smaller than the rank of $\Vor_0$. We may proceed with the proof by induction, with the hypothesis that all $\Vor_i$'s are connected in $H_{\Vor}$. Since the ranks of $\Vor_i$'s are strictly decreasing, the inductive argument halts after $T < |S|$ steps: at this point, $A$ is connected to $B_T \subseteq \Vor_T$ in $H_{\Vor}$, as desired.
\end{proof}

We next claim that stretch of the our spanner $\Hdense$ is $O(k^2)$, while \cite{LeviLenzen17} provides a stretch factor of $O(\log n \cdot (\Delta + \log n))$. The second factor of $O(\Delta \log n)$ has been reduced down to $O(k)$ thanks to the new partitioning criteria and algorithms described so far. To remove the remaining factor of $O(\log n)$, we leverage the fact we may add a factor of $O(n^{1/k}\log n)$ more edges to the spanner $\Hdense$, allowing the ranks in the inductive argument to decrease more rapidly. For simplicity, we assume now that the ranks of the centers are fully independent. In Section \ref{sec:boundedindepsparse} (Theorem~\ref{lem:stretchk_bounded}) we extend the following claim to the case where the ranks of the centers are formed by short random seed of $O(\log^2 n)$ bits.

\begin{lemma}[Stretch guarantee by $\HdenseB$]\label{lem:stretchk}
If the ranks of the centers are assigned independently, uniformly at random from $[0, 1)$, then w.h.p., the stretch of $H_{\Vor}$ w.r.t.~$G_{\Vor}$ is $O(k)$.
\end{lemma}
\begin{proof}
The connectivity proof in Lemma~\ref{lem:connected} uses an inductive argument, where each step in the induction, increases the length of the path in $G_{\Vor}$ by $2$. Thus it suffices to show that the induction of Lemma~\ref{lem:connected} halts, w.h.p., after $O(k)$ steps. In comparison, in Lemma 4 of \cite{LeviLenzen17}, the induction uses $O(\log n)$ steps and hence the stretch in $H_{\Vor}$ is also $O(\log n)$. 

Observe that while the construction of $G_{\Vor}$ heavily relies on the $\ID$s of vertices, the rank assignment of vertices is random and independent of $G_{\Vor}$.
Again, let $A \subseteq \Vor, B=B_0 \subseteq \Vor_0$ be two adjacent clusters of interest. Folowing the argument of Lemma~\ref{lem:connected}, at each step $i\geq 0$, we consider $\Vor_i$ which by the inductive hypothesis satisfies the following.
\begin{compactenum}[(a)]
\item $A$ and $\Vor_i$ are adjacent.
\item The distance between $\Vor_0$ and $\Vor_i$ in $H_{\Vor}$ is at most $2i$.
\item The rank of $c(\Vor_i)$ is the minimum rank among those of all centers in the collection $\{c(\partial A)\cap c(\partial C_j)\}_{j<i}$.
\end{compactenum}

It is straightforward to verify that these conditions hold for the base case $i = 0$. For the inductive step, hypothesis (a) holds because we choose $\Vor_i$ with center $s_i \in c(\partial A)\cap c(\partial C_{i-1})$, so $\Vor_i$ is adjacent to $\Vor$. It is also connected to the marked cluster $C_{i-1}$, which in turn is connected to $B_{i-1}$ in $\Vor_{i-1}$ by rule (1), thereby proving condition (b). Lastly, condition (c) follows, because $c(\Vor_i)$ is the center of minimum rank in the set of centers $c(\partial A)\cap c(\partial C_{i-1})$, which contains $c(\Vor_{i-1})$.

It remains to show that the induction terminates after $O(k)$ steps with high probability. Let $r_i=r(c(\Vor_i))$. We claim that in each step, either the process terminates or, w.h.p., chooses a center of rank $r_{i+1} \leq r_i / n^{1/k}$. Suppose that the process does not terminate at step $i$. Observe that at this point, all ranks ever ``revealed'' by our algorithm so far are of the centers in condition (c): no rank lower than $r_i$ has been encountered. Then in the beginning of step $i$, there are at least $q$ cluster centers in $c(\partial A)\cap c(\partial C_i)$ whose ranks are uniformly distributed in $[0, r_i)$ (since we assume that ranks are chosen independently).
For each of these $q = \Theta(n^{1/k} \log n)$ unrevealed ranks, the probability that the rank is at most $r_i/n^{1/k}$ is at least $n^{-1/k}$.
By the Chernoff bound we obtain that, w.h.p., at least one of these ranks turns out to be at most $r_i / n^{1/k}$.
Similarly, w.h.p.,~no center has rank below $\Theta(1/(n \log n))$. Thus, the algorithm terminates in $\log_{n^{1/k}} (n \log n) = \Theta(k)$ steps, as desired.
\end{proof}

Next, we proceed to bounding the size of $\HdenseB$.
\begin{lemma}[Size of $\HdenseB$]\label{lem:denseb-size}
W.h.p., $\HdenseB$ contains $O(\frac{pn^{2+1/k} \log^4 n}{L^2}+\frac{n \log^2 n}{pL})=O(n^{1+1/k}\log^4 n)$ edges.
\end{lemma}
\begin{proof}
Recall that there are $O((n \log n)/L)=O(n^{2/3}\log n)$ clusters and $((pn \log^2 n)/L)=O(n^{1/3}\log^2 n)$ marked clusters, bounding the number of edges from condition (1) by $O((pn^2 \log^3 n)/L^2)=O(n \log^3 n)$. For condition (3), the algorithm adds $O(n^{1/k}\log n)$ edges for each such pair, which is $O((pn^{2+1/k} \log^4 n)/L^2)=O(n^{1+1/k} \log^4 n)$ edges in total. Lastly for conditon (2), every cluster that is not participating in any cluster of clusters (i.e., not adjacent to any marked Voronoi cell) w.h.p.~has $O((\log n)/p) = O(n^{1/3} \log n)$ adjacent Voronoi cells, because these cells are independently marked with probability $p = n^{-1/3}$. (On the other hand, clusters are not marked independently, so in condition (2) we add one edge from $A$ to every adjacent Voronoi cell, rather than every adjacent cluster.) Hence, the number of edges added by condition (c) is $O((n \log^2 n)/(pL))=O(n\log^2 n)$.
\end{proof}

\paragraph{Putting everything together} Recall that our overall spanner is $H=\Hsparse\cup\Hdense$ where $\Hdense = \HdenseI\cup\HdenseB$. Combining all results so far in this section, we achieve at our main result, Theorem~\ref{thm:lowdegree}, as follows.
\begin{proof}[Proof of Theorem~\ref{thm:lowdegree}]
\paragraph{(i) Size} The size of $\Hsparse$, $\HdenseI$ and $\HdenseB$ are $O(kn^{1+1/k})$, $O(n)$ and $O(n^{1+1/k} \log^4 n)$ due to Lemma~\ref{lem:k2sparse}, Lemma~\ref{lem:k2denseI} (from the fact that $\HdenseB$ is a forest), and Lemma~\ref{lem:denseb-size}, respectively. More precisely, for parameters $L$ and $p$, we present a spanner with $O(kn^{1+1/k}+\frac{pn^{2+1/k} \log^4 n}{L^2}+\frac{n \log^2 n}{pL})$ edges.

\paragraph{(ii) Stretch} The case of $\Hsparse$ taking care of $\Esparse$ is immediate by Lemma~\ref{lem:k2sparse}, hence we focus on $\Hdense$. The stretch argument follows by Lemma~\ref{lem:stretchk} for $\HdenseB$ together with the fact that in each Voronoi cell we have a Voronoi tree of depth $O(k)$ in $\HdenseI$ by \ref{lem:k2denseI}. That is, between two adjacent Voronoi cells, the spanner has a path of length $O(k)$ in the Voronoi graph $H_{\Vor}$. Within each Voronoi cell (supervertex in $G_{\Vor}$) there exists a path of length $2k$ connecting any pair of vertices. Thus, there is a path of length $O(k^2)$ in $\Hdense$ between any pair of neighboring dense vertices.
 
\paragraph{(iii) Probes} The LCA can verify whether $(u,v) \in \Esparse$, and if so, check if $(u,v) \in \Hsparse$ using $O(\Delta L)$ probes by Lemma~\ref{lem:k2sparse} using $O(\Delta^2 L^2)$ total probes. Otherwise, Lemma~\ref{lem:k2denseI} allows the LCA to verify whether $u$ and $v$ belongs to the same Voronoi cell, and if so, check whether $(u,v)\in\HdenseI$ using $O(\Delta^2 L^2)$ probes. Lastly for $u$ and $v$ from different Voronoi cells, the LCA can check whether $(u,v)\in\HdenseB$ using $O(p\Delta^4 L^3 \log n)$ probes via Lemma~\ref{lem:k2denseprobe}. Substituting $L=n^{1/3}$ and $p=1/L$ yields the desired result.
\end{proof}

Theorem~\ref{thm:lowdegree} implies that there exists an LCA with sub-linear probe complexity for any $\Delta = O(n^{1/12-\epsilon})$. In fact, we remark that by using the argument of Lemma \ref{lem:stretchk}, we can achieve a spanner $H$ with $\widetilde{O}(n^{1+1/k} + nq)$ edges with stretch $O(k \log_q n) = O((k \log n)/ \log q)$. As a reminder, the theorem above does not show that the LCA uses a polylogarithmic number of independent random bits. To complete the proof of Theorem~\ref{thm:lowdegree}, Section~\ref{sec:boundedindepsparse} describes the required adaptation for working with only $O(\log^2 n)$ random bits. 

\paragraph{A summary of the differences between our algorithm and the algorithm of \citet{LeviLenzen17}}
Our algorithm can be considered as an extension of~\cite{LeviLenzen17} that provides a trade-off between the stretch factor and the size of the subgraph. In particular, we show that the stretch factor's dependency on $\Delta$ and $n$ can be removed completely. We conclude by summarizing several key differences between our approaches. 
\begin{compactitem}
\item
In \cite{LeviLenzen17}, the distinction between dense and sparse vertices depends on a radius $\ell$ sampled uniformly at random from a given range that depends on $\Delta$. In our construction, the radius is $k$, the stretch parameter.
\item
In \cite{LeviLenzen17}, the sparse and dense graphs are vertex disjoint and the parameter $\ell$ guarantees that the number of edges between these graphs is small. In contrast, in our construction the sparse and dense graphs share vertices and in fact, these graphs are only {\em edge-disjoint}.
\item
The BFS algorithm of \cite{LeviLenzen17} for detecting a center explores an entire level of the BFS tree in each step, choosing the closest center with minimum $\ID$. We provide a more efficient variant that explores the neighborhood of one vertex at a time, and chooses the closest center with lexicographically-first shortest path, improving the probe complexity by a factor of $\Delta$.
\item
For the sparse case, \cite{LeviLenzen17} uses the distributed algorithm of \citet{elkin2017efficient}, whereas we use the algorithm of \citet{baswana07} since it has been proved to work with $O(\log n)$-wise independence~\cite{Censor-HillelPS16}. 
\item
For the dense case, in \cite{LeviLenzen17}, the radius of the Voronoi cells is $\ell=\Theta(\Delta+\log n)$ and in our case, it is $k$. 
\item 
The number of clusters in \cite{LeviLenzen17} depends on $\ell$ and $\Delta$. In our construction, the number of clusters is $\widetilde{O}(n^{2/3})$, each containing $O(n^{1/3})$ vertices. 
\item 
We allow $O(n^{1/k}\log n)$ edges between a cluster and neighboring clusters of a given marked clusters, whereas \cite{LeviLenzen17} only adds a single such edge.
\item
The algorithm of \cite{LeviLenzen17} uses random seed of size $O(\Delta \cdot n^{2/3})$. However, our algorithm only uses a poly-logarithmic number of random bits.
\end{compactitem}

\section{Bounded Independence}\label{APPEND:bounds}
In this section, we show that all our LCA constructions succeed w.h.p.~using {\em $\Theta(\log n)$-wise independent hash functions} which only require $\Theta(\log^2 n)$ random bits. 
We use the following standard notion of $d$-wise independent hash functions as in \cite{Vadhan12}. In particular, our algorithms use the explicit construction of $\Hcal$ by \cite{Vadhan12}, with the parameters as stated in Lemma~\ref{lem: d-wise independent}.
\begin{definition}
	\label{def: d-wise independent}
	For	$N,M,d \in \mathbb{N} $ such that $d \leq N$, a family of functions $\Hcal = \set{h : [N] \rightarrow
		[M]}$ is $d$-wise independent if for all distinct $x_1,...,x_d \in [N],$ the
	random variables $h(x_1),...,h(x_d)$ are independent and uniformly distributed
	in $[M]$ when $h$ is chosen randomly from $\Hcal$.
\end{definition}

\begin{lemma}[{Corollary 3.34 in \cite{Vadhan12}}]
	\label{lem: d-wise independent}
	For every $\gamma,\beta,d \in \mathbb{N},$ there is a family of $d$-wise independent functions $\mathcal{H}_{\gamma,\beta} = \set{h : \set{0,1}^\gamma \rightarrow \set{0,1}^\beta}$ such that choosing a random function from $\mathcal{H}_{\gamma,\beta}$ takes $d \cdot \max \set{\gamma,\beta}$ random bits, and evaluating
	a function from $\mathcal{H}_{\gamma,\beta}$ takes time $\poly(\gamma,\beta,d)$.
\end{lemma}

Then, we exploit the following result to show the concentration of $d$-wise independent random variables:

\begin{fact}[Theorem 5(III) in \cite{schmidt1995chernoff}]
\label{fc:bounded}
If $X$ is a sum of $d$-wise independent random variables, each of which is in the interval $[0,1]$ with $\mu=\mathbb{E}(X)$, then:
\begin{compactitem}
\item{(I)}
For $\delta\leq 1$ and $d \leq \lfloor \delta^2 \mu e^{-1/3}\rfloor$, it holds that 
$\Pr[|X-\mu|\geq \delta \mu]\leq e^{-\lfloor d/2\rfloor}.$
\item{(II)}
For $\delta\geq 1$ and $d=\lceil \delta \mu\rceil$, it holds that:
$\Pr[|X-\mu|\geq \delta \mu]\leq e^{-\delta\mu/3}.$ 
\end{compactitem}
\end{fact}

\paragraph{Bounded independence for hitting set procedures}
Most of our algorithms are based on the following hitting set procedure. For a given threshold $\Delta \in [1,n]$, each vertex flips a coin with probability $p=(c\log n)/\Delta$ of being head and the set of all vertices with head outcome join the set of centers $S$. Assuming the outcome of coin flips are fully independent, by the Chernoff bound, the followings hold w.h.p.:
\begin{description}
\item{(HI)}
There are $\Theta(p n)$ sampled vertices $S$.
\item{(HII)}
For each vertex of degree at least $\Delta$, it has $\Theta(\log n)$ centers among its first $\Delta$ neighbors.   
\end{description}
Here we show that to satisfy properties (HI) and (HII), it is sufficient to assume that the outcomes of the coin flips are $d$-wise independent. By Lemma~\ref{lem: d-wise independent}, to simulate $d$-wise independent coin flips for all vertices, the algorithm only requires $t = \Theta(d (\log n + \log 1/p))$ random bits: more precisely, setting $\gamma=\Theta(\log n)$ and $\beta=\log {1/p}$ (for simplicity, lets assume that $\log 1/p$ is an integer), there exits a family of $d$-wise independent functions $\Hcal=\set{h : \set{0,1}^{\Theta(\log n)} \rightarrow \set{0,1}^{\log (1/p)}}$ such that a random function $h \in \Hcal$ can be specified by a string of random bits of length $t$. In other words, each function $h\in \Hcal$ maps the $\ID$ of each vertex to the outcome of its coin flip according to a coin with bias $p$. 
Then, from a string $\mathcal{R}$ of $t$ random bits, the algorithm picks a function $h_{\mathcal{R}} \in \Hcal$ at random to simulate the coin flips of the vertices accordingly: the outcome of the coin flip of $v$ is head if $h_{\mathcal{R}}(\ID(v)) = 0$ (which happens with probability $p$) and the coin flips are $d$-wise independent.
Setting $d=c \log n$ for some constant $c>1$, we prove the following:

\begin{claim}\label{cl:bounded}
If the coin flips are $d$-wise independent then properties (HI) and (HII) holds. Furthermore, the sequence of $n$ $d$-wise independent coin flips can be simulated using a string of $O(\log^2 n)$ random bits.
\end{claim}

\subsection{Construction of representatives in Section~\ref{sec:5rep}} 
The analysis above (for hitting set procedures) also extends to the process of computing $\RepsOf$. Each crowded vertex chooses values $c\log n$ random indices (of its neighbor-list) in $[\MedDeg]$, each of which has probability $1/2$ of hitting a neighbor of degree at least $\SuperDeg$. Let $\{Z_i\}_{i \in [c \log n]}$ be indicators for these events and $Z$ denote their sum, then the expected sum $\mathbb{E}(Z) \geq (c/2)\log n$. Imposing $d$-wise independence, Fact \ref{fc:bounded}(I) implies that w.h.p., $Z>0$, so the representative set is non-empty. We apply the union bound to show that $\RepsOf(v) \neq \emptyset$ for every $v\in\CrowdedVer$, as desired.

\subsection{Bounded independence for Section \ref{sec:fullsparse}} \label{sec:boundedindepsparse}
To define the $\ell=\lceil \log n \rceil$-bit random rank $r(v)$, we will use a collection of $k$ hash functions (where $k$ is the stretch parameter). Letting $N=\lceil \log n /k \rceil$, each function $h_i:\{0,1\}^\ell\to \{0,1\}^N$ is an $O(\log n)$-wise independent hash function for $i \in \{1, \ldots, k\}$.

To do so, we view the rank $r(v)$ as consisting of $k$ blocks, each with $N$ bits. 
Specifically, let $r(v)=[b_1, \ldots, b_{\ell}] \in \{0,1\}^\ell$ and let $R_i(v)=[b_{(i-1)\cdot N}, \ldots, b_{i\cdot N-1}]$ be the $i^{th}$ block of $N$ bits in $r(v)$. For every center $v$, define
$$R_i(v)=h_i(\ID(v)) ~\mbox{and}~ r(v)=R_{1}(v) \circ R_{2}(v) \circ \ldots \circ R_{k}(v)~.$$
The collection of these $h_1, \ldots, h_k$ functions are obtained by a uniform sampling from a family $\mathcal{H}=\{ h:\{0,1\}^\ell\to \{0,1\}^N\}$ of $O(\log n)$-wise independent hash functions.

Our goal is prove Lemma \ref{lem:stretchk} using these ranks instead of fully independent random ranks.
\begin{lemma}[Stretch guarantee by $\HdenseB$]\label{lem:stretchk_bounded}
If the ranks of the centers are generated according to the above construction, then w.h.p., the stretch of $H_{\Vor}$ w.r.t.~$G_{\Vor}$ is $O(k)$.

\end{lemma}
\begin{proof}
Note that $G_{\Vor}$ is independent of the rank assignment. Consider any pair of adjacent cells $\Vor,\Vor_1$ (i.e., neighbors in $G_{\Vor}$) and let $A \subseteq \Vor, B \subseteq \Vor_1$ be two adjacent clusters of interest in these Voronoi cells. 

At the beginning all vertices are unrevealed and throughout the process some of them will get revealed by exposing \emph{one} $N$-size block $R_j$ of their rank. Let $q=\lceil  c \log n \cdot n^{1/k}\rceil$ for some large enough constant $c$, as used by our spanner construction algorithm. 
In each inductive step $i$, we either halt or we reveal the $i^\textrm{th}$ block $R_i(v)$ in the ranks of at least $q$ oblivious \emph{unrevealed} centers $v$. At that point, we will also reveal the $i^\textrm{th}$ block in the rank of all the centers $w$ with $R_i(w)\neq \bar{0}$ (where $\bar{0} = [0,\ldots, 0]$). 

We now describe this induction process in details. At the beginning of step $i\geq 0$, we look at $c(\Vor_i)$ which by induction assumption satisfies the following.
\begin{compactenum}[(a)]
\item $A$ and $\Vor_i$ are adjacent.
\item The distance between $\Vor_0$ and $\Vor_i$ in $H_{\Vor}$ is at most $2i$.
\item The rank of $c(\Vor_i)$ is the minimum rank among those of all centers in the collection $\{c(\partial A)\cap c(\partial C_j)\}_{j<i}$.
\end{compactenum}
Observe that all vertices whose ranks are revealed are precisely those included in property (c).
In particular, we will show property (c) as a result of two sub-properties: 
\begin{compactenum}[(c1)]
\item The first $i$ blocks in the rank of $c(\Vor_i)$ are all \emph{zeros}.
\item For every center $v$ whose rank is revealed, there is exists $j \leq i$ such that $R_{j}(v) \neq \bar{0}$.
\end{compactenum}


For the base case, at the beginning of step $i$, all claims hold. 

Assume that the claims hold up to the beginning of step $i\geq 1$. We will show that either we halt at that step or that all properties hold at the beginning of step $i+1$. By property (c2), each revealed center $v$ at the beginning of step $i$ has at least one non-zero block among the first $i$ blocks of $r(v)$. 
Or, in other words, the first $i$ blocks in the ranks of all the unrevealed vertices at the beginning of step $i$, are \emph{all-zeros}.

We may assume that there is a marked cluster $C_i$ such that $B_i$ (the cluster in $\Vor_i$ such that the edge of minimum $\ID$ in $E(A, \Vor_i)$ is in $E(A, B_i)$) participates in $\mathcal{C}(C_i)$ (as otherwise, we halt). If there are less than $q$ unrevealed centers in $c(\partial A)\cap c(\partial C_i)$, then the process terminates: by property (c), all revealed centers have a strictly larger rank than $c(\Vor_i)$. Otherwise, (i.e., there are at least $q$ unrevealed centers in $c(\partial A)\cap c(\partial C_i)$), we probe the $i^\textrm{th}$ block (using the hash function $h_i$) in the rank of these $q$ unrevealed centers in $c(\partial A)\cap c(\partial C_i)$. We let $\Vor_{i+1}$ be a cell with a center $s_{i+1}=c(\Vor_{i+1})$ satisfying that $s_{i+1} \in c(\partial A)\cap c(\partial C_i)$ and
$R_i(s_{i+1})=\bar{0}$. If there are several such centers that satisfy these two conditions, we pick one arbitrarily. We now claim:
\begin{claim}\label{cl:lognindepranks}
W.h.p., there exists at least one $s_{i+1}\in c(\partial A)\cap c(\partial C_i)$ such that $R_i(s_{i+1})=\bar{0}$.
\end{claim}
\begin{proof}
Let $S'$ a subset of $q$ unrevealed centers in $c(\partial A)\cap c(\partial C_i)$.
For every $s_j \in S'$, let $X_j \in \{0,1\}$ be the event that $R_i(s_j)=\bar{0}$. Since $R_i(s_j)=h_i(\ID(s_j))$, we have that $\mathbb{E}(X_j)=1/2^N$ and $\mathbb{E}(X)=q/2^N=\Theta(\log n)$ where $X=\sum_{j=1}^q X_j$.
Since the $X_j$ variables are $O(\log n)$-independent, using the Chernoff bound from Fact \ref{fc:bounded}(I), we obtain that w.h.p.~$X \geq 1$ and hence there exists $s_j \in S'$ that satisfies the above. The claim follows.
\end{proof}
The proofs of the first two properties remain unchanged.
Property (3a) holds by induction and by the selection of $\Vor_{i+1}$. In particular, by induction, all the first $i$ blocks of the rank $r(s_{i+1})$ are all zeros (as $s_{i+1}$ is unrevealed at the beginning of step $i$) and we select $s_{i+1}$ since $R_{i}(s_{i+1})=\bar{0}$. Property (c2) holds by induction and by the fact that the $i^\textrm{th}$-block in the ranks of all those centers that got revealed in step $i$ is \emph{nonzero}. By combining (c1) and (c2), property (c) holds as well since $s_{i+1}$ has the minimum rank among all those that got revealed so far. 

Finally, we claim that w.h.p., the process terminates after $O(k)$ induction steps. We will show that by claiming that in every step $i$, at least a $(1-c'\cdot n^{-1/k})$ fraction of the remaining unrevealed centers are revealed for some constant $c'>0$. Let $U_i$ be the number of unrevealed centers at the beginning of step $i$. Hence, $U_1=n$. 
If we did not halt at step $i$, it means that $U_i\geq q=\Omega(\log n \cdot n^{1/k})$. 
We now bound the number $UZ_i$ of unrevealed centers at the beginning of step $i$ whose $i^\textrm{th}$ block is all-zero. The probability of having an all-zero block for a single center is $1/2^N$ and hence in expectation there are $U_i/2^N$ such centers. Since $U_i\geq q$, and since the ranks are $O(\log n)$-wise independent, using Chernoff bound of Fact \ref{fc:bounded}(I), with get that w.h.p. $UZ_i \in [c_1 \cdot U_i/2^N, c_2 \cdot U_i/2^N]$ for some constants $0<c_1<c_2$. Hence, w.h.p., $U_{i+1}=U_{i}-UZ_i \geq (1-c'/2^N)U_i$. 
Overall, after $O(k)$ induction steps, there are at most $q$ unrevealed vertices and at that point we halt. The lemma follows.
\end{proof}

\section{Lower Bounds}\label{sec:lowerbound}
\def\q{\mathcal{L}}
In this section, we establish lower bounds for the problem of locally constructing a spanner consisting of an asymptotically sub-linear number of edges from the input graph. Our results largely follows from the analysis of~\cite{kaufman2004tight} on the lower bound construction of~\cite{LRR14}; a compact version of this proof is given here for completion.

For simplicity, we assume that each vertex occupies a unique ID from $\{1, \ldots, n\}$; this assumption may only strengthen our lower bound. We define an \emph{instance} of a $d$-regular graph on $n$ vertices as a perfect matching between cells of a table of size $n \times d$: a matching between the cells $(u, i)$ and $(v, j)$ indicates that $v$ is the $i^\textrm{th}$ neighbor of $u$ and $u$ is the $j^\textrm{th}$ neighbor of $v$. An edge can be then expressed as a quadruple $(u,i,v,j)$; note that the endpoints are always interchangeable. For consistency with this notation, we let the \neighborP~probe with parameter $\langle u,i \rangle$ for $i \leq \deg(u)$ return $(v, j)$ (instead of only $v$) -- this change can only provide more information to the algorithm. We say that an instance $G$ and the edge $(u,i,v,j)$ are \emph{compatible} if $G$ contains $(u,i,v,j)$. Our lower bounds are established for sufficiently large $n \equiv 2 \textrm{ mod } 4$ and odd integer $d$.

\paragraph{The overall argument} First, we construct two distributions $\mathcal{D}_{(x,a,y,b)}^+$ and $\mathcal{D}_{(x,a,y,b)}^-$ over undirected $d$-regular graph instances for $x,y \in V$ and $a, b \in [d]$. Any graph instance $G^+$ in the support of $\mathcal{D}_{(x,a,y,b)}^+$ contains the edge $(x,a,y,b)$ such that with high probability, removing this edge does not disconnect $x$ and $y$. In particular, $\mathcal{D}_{(x,a,y,b)}^+$ is the uniform distribution over all instances compatible with $(x,a,y,b)$. On the other hand, any graph instance $G^-$ in the support of $\mathcal{D}_{(x,a,y,b)}^-$ contains the edge $(x,a,y,b)$ such that removing this edge disconnects $x$ and $y$ (leaving them in separate connected components). 

We show that when given the query $(x,a,y,b)$, any \emph{deterministic} LCA $\mathcal{A}$ that only makes $o(\min\{\sqrt{n},{n\over d}\})$ probes can only distinguish whether the underlying graph is a graph randomly drawn from $\mathcal{D}_{(x,a,y,b)}^+$ or $\mathcal{D}_{(x,a,y,b)}^-$ with probability $o(1)$. We prove this claim by defining two processes $\mathcal{P}_{(x,a,y,b)}^+$ and $\mathcal{P}_{(x,a,y,b)}^-$ which interact with $\mathcal{A}$ and generate a random subgraph from $\mathcal{D}_{(x,a,y,b)}^+$ and $\mathcal{D}_{(x,a,y,b)}^-$ respectively. We then argue that for each probe the answers that these two processes return are nearly identically distributed, and so are their \emph{probe-answer histories}.

Aiming for an overall success probability of $2/3$, $\mathcal{A}$ must keep the edge $(x,a,y,b)$ in its spanner with probability $\frac{2}{3}(1-o(1)) > 1/2$. Since an instance in $\mathcal{D}_{(x,a,y,b)}^+$ is chosen uniformly at random, then for more than half of the instances in the support of $\mathcal{D}_{(x,a,y,b)}^+$, which are exactly the instances compatible with $(x,a,y,b)$, $\mathcal{A}$ returns \YES~on query $(x,a,y,b)$. Applying this argument for all possible edges (quadruples $(x,a,y,b)$), we obtain that $\mathcal{A}$ returns \YES~on at least half of all compatible instance-query pairs. Consequently, over the uniform distribution over all instances, in expectation any deterministic algorithm $\mathcal{A}$ must return \YES~on more than $m/2$ edges. Employing Yao's principle, we conclude that any (randomized) LCA cannot compute a spanning subgraph with $o(m)$ edges using $o(\min\{\sqrt{n},n/d\})$ probes. 

\subsection{Analysis of the probe-answer histories}
Similarly to the work of \cite{LRR14}. we construct our distributions as follow.
\begin{itemize}
\item \textbf{Distribution $\mathcal{D}_{(x,a,y,b)}^+$.} $\mathcal{D}_{(x,a,y,b)}^+$ is a uniform distribution over all $d$-regular graph instances, conditioned that $(x,a,y,b)$ is in the instance. More precisely, the edges of $G$ in the family is determined by the following process. Consider a two-dimensional table of size $n\times d$ which is called {\em matching table} and is denoted by $M$. Any perfect matching between cells in this table corresponds to a graph in $\mathcal{D}_{(x,a,y,b)}^+$. Note that the generated graphs are not necessarily simple.
 
\item \textbf{Distribution $\mathcal{D}_{(x,a,y,b)}^-$.} Let $V = V_0 \uplus V_1$ be a {\em random} partition of the vertex set into two equal sets such that $x \in S$ and $y \in T$. Now consider two matching tables of each of size ${n/2} \times d$ denoted by $M_1$ and $M_2$. For a graph $G$ in this family, besides the edge $(x,a,y,b)$, the rest of edges are determined by choosing a random perfect matching \emph{within} each of $M_1$ and $M_2$ (over the remaining cells). Thus, $(x,a,y,b)$ is the only edge connecting between $M_1$ and $M_2$.
\end{itemize}
For brevity we drop the subscript $(x,a,y,b)$ for now as it is clear from the context. For sufficiently large values of $d = \Omega(1)$, w.h.p, each instance $G$ from $\mathcal{D}^+$ is connected even when $(x,y)$ is removed (see e.g., \cite{bollobas2001random}). On the other hand, removing $(x,y)$ from any $G^-\in \mathcal{D}^-$ clearly disconnects $x$ and $y$. Thus, unless a \emph{deterministic} algorithm $\mathcal{A}$ can determine whether it is given $(x,a,y,b)$ of an instance from $\mathcal{D}^+$ or $\mathcal{D}^-$, it must return \YES~on $(x,a,y,b)$ for a $(2/3)$-fraction of these instances. For simplicity we assume that $\mathcal{A}$ has a knowledge of the construction (including the degree $d$), and never makes a probe that does not reveal any new information. 

Let $\q$ denote the number of probes made by the algorithm, and $Q$ denote the set of probes performed by $\mathcal{A}$. Observe that $\mathcal{A}$ is a deterministic mapping from the probe-answer histories $\langle (q_1,a_1), \cdots, (q_t,a_t) \rangle \mapsto q_{t+1}$ for $t<\q$ and to $\{\YES,\NO\}$ for $t=\q$. Each probe $q_i$ is either a \neighborP~probe or an \adjacencyP~probe.

Next, similarly to~\cite{kaufman2004tight}, we define two processes $\mathcal{P}^+$ and $\mathcal{P}^-$ which interact with an arbitrary algorithm $\mathcal{A}$ and respectively construct a random graph from $\mathcal{D}^+$ and $\mathcal{D}^-$. Defining $D_t^+$ and $D_t^-$ to be the distribution of the probe-answer histories of the interaction of $\mathcal{P}^+$ and $\mathcal{P}^-$ respectively with $\mathcal{A}$ after $t$ probes, we show that if $\q = o(\min\{\sqrt{n}, n/d\})$, then the statistical distance of $D_\q^+$ and $D_\q^-$ is $o(1)$. We now give the formal description of $\mathcal{P}^s$ for $s\in\{+,-\}$:
\begin{compactitem}
\item Let $R^s$ be the set of all graphs in the support of $\mathcal{D}^s$. Let $R^s_{(u,v)}$ and $R^s_{(u,i,v,j)}$ be the set of all graphs in the support of $\mathcal{D}^s$ that are compatible $(u,v)$ and $(u,i,v,j)$ respectively. In the former case, we require at least one matching $(u,i',v,j')$ for some $i',j'\in [d]$ between cells in the rows of $u$ and $v$ in the matching table; however, in the latter case, we only allow a fixed matching $(u,i,v,j)$. We also write $R^s_{\overline{(u,v)}}$ to denote the set of all graphs in the support of $\mathcal{D}^s$ that are not compatible with $u,v$.

\item{Starting from $R_0^s = R^s_{(x,a,y,b)}$, for any $t>0$, $R_t^s$ denotes the set of all graphs in the support of $\mathcal{D}^s$ that are compatible with the first $t$ probes and answers}.
	\begin{compactitem}
		\item{ If $q_t$ is an \adjacencyP~probe of the form $\langle u_t,v_t \rangle$:} We choose whether to add an edge between $u$ and $v$ with probability ${|R^s_{(u_t,v_t)}\cap R^s_{t-1}| / |R^s_{t-1}|}$. If so, we match a pair of cells between the rows of $u$ and $v$: sample $(i_t, i_t) = (i,j)$ with probability ${|R^s_{(u_t,i_t,v_t,j_t)}\cap R^s_{t-1}| / |R^s_{t-1}|}$ set $R^s_t = R^s_{(u_t,i_t,v_t,j_t)}\cap R^s_{t-1}$, and answer $a_t = i_t$. Otherwise, we simply set $R^s_t = R^s_{\overline{(u_t,v_t)}}\cap R^s_{t-1}$ and answer $a_t = \bot$.
		\item{ If $q_t$ is a \neighborP~probe of the form $\langle u_t, i_t \rangle$: For each $v\in V$ and $j_t\in\{1,\ldots,d\}$, we choose a cell to match with $(u_t,i_t)$: sample the answer $a_t = (v_t,j_t)$ with probability ${|R^s_{(u_t,i_t,v_t,j_t)}\cap R^s_{t-1}| / |R^s_{t-1}|}$ and set $R^s_t = R^s_{(u_t,i_t,v_t,j_t)}\cap R^s_{t-1}$.}
	\end{compactitem}
\item{After $\q$ probes, return a random graph uniformly sampled from $R^s_{\q}$.}
\end{compactitem} 

\begin{lemma}[{\bf Lemma 10 in~\cite{kaufman2004tight}}]
For any deterministic algorithm $\mathcal{A}$, the process $\mathcal{P}^s$ {\em ($s\in\{+,-\}$)} when interacting with $\mathcal{A}$, uniformly generates a graph from the support of $\mathcal{D}^{s}_{(x,a,y,b)}$.
\end{lemma}

Next, we show that the probability that $\mathcal{A}$ can detect an edge with \adjacencyP~probe (asking probe $q=(u,v)$ for which the answer is positive; an edge exists between $u$ and $v$) after performing only $o(n/d)$ is small. We can define $R^s_t$, the set of all graphs in $\mathcal{D}^s$ as $R^s_{B, \overline{D}}$ where $B$ is the set of edges that the graphs in $R^{s}_t$ must contain (namely, all pairs of cells $(u,i,v,j)$ created in some previous probes) and $D$ is the set of edges that the graphs in $R^{s}_t$ must not contain (namely, all pairs $(u,v)$ disallowed by \adjacencyP~probes with negative answer).

Assuming that the algorithm makes $\q=o(n/d)$ probes, we establish the following lemmas that will be useful in bounding the difference between the distributions of probe-answer histories generated by the two processes. In particular, assume the number of conditions $|B|,|D|= o(n/d)$, and the initial conditions $(x,a,y,b) \in {B}$ and $(x,y) \notin D$, in the following three lemmas.
\begin{lemma}\label{lem:edge-presence-density}
For every $(u,i,v,j) \neq (x,a,y,b)$,
${|R^s _{(u,i,v,j)} \cap R^s_{B, \overline{D}}| \over |R^s_{B, \overline{D}}|} = O({1\over nd})$.
\end{lemma}
\begin{proof}
For process $\mathcal{P}^+$, the proof is the same as the proof of the similar statement in Lemma 11 of~\cite{kaufman2004tight}. 
Here, we show that the argument holds for $\mathcal{P}^-$.
\begin{align*}
{|R^- _{(u,i,v,j)} \cap R^-_{B, \overline{D}}| \over |R^-_{B,\overline{D}}|} &= {|R^-_{(u,i,v,j)} \cap R^-_B|\over |R^-_B|} \cdot {|R^-_{(u,i,v,j)} \cap R^-_{B,\overline{D}}|\over |R^-_{(u,i,v,j)} \cap R^-_B|} \cdot {|R^-_B| \over |R^-_{B,\overline{D}}|}\\
&\leq {1 \over \Omega(nd)} \cdot 1 \cdot O(1) = O\left({1\over nd}\right), 
\end{align*}
where the bounds on ${|R^-_{(u,i,v,j)} \cap R^-_B|\over |R^-_B|}$ and ${|R^-_B| \over |R^-_{B,\overline{D}}|}$ are shown in Claim~\ref{cl:edge-port-presence} and~\ref{cl:ratio}.
\end{proof}

\begin{claim}\label{cl:edge-port-presence}
For every $(u,i,v,j) \neq (x,a,y,b)$,
${|R^-_{(u,i,v,j)} \cap R^-_B| \over |R^-_B|} \leq {2\over nd}$.
\end{claim}
\begin{proof}
If $u$ and $v$ belong to different partitions or at least one of them is already matched in $B$, then $R^-_{(u,i,v,j)}=\emptyset$. Otherwise, let $w$ denote the lower bound on the number of unmatched cells in the matching table containing rows of $u$ and $v$ in any instance of $R^-_B$. Recall that the number of matched cells is bounded by $o(nd)$, so $w \geq nd-2|B|-1\geq (1-o(1))\cdot nd$.  The probability that cells $(u,i)$ is matched to $(v,j)$ is given by
\begin{align*}
{|R^-_{(u,i,v,j)} \cap R^-_B| \over |R^-_B|} = {1 \over w-1} \leq {2\over nd}
\end{align*}
for sufficiently large $n$ and $d$.
\end{proof}


\begin{claim}\label{cl:ratio}
${|R^-_B| \over |R^-_{B,\overline{D}}|} =O(1)$.
\end{claim}
\begin{proof}
As we consider $(u,v) \neq (x, y)$, we have
\begin{align*}
{|R^-_{(u,v)} \cap R^-_B| \over |R^-_B|} = {|\big(\bigcup_{i,j\in[d]} R^-_{(u,i,v,j)}\big) \cap R^-_B|\over |R^-_B|} &\leq {\sum_{i,j\in[d]} |R^-_{(u,i,v,j)} \cap R^-_B|\over |R^-_B|}\leq d^2 \cdot {2\over nd} = O\left({d\over n}\right).
\end{align*}
Then by the union bound,
\begin{align*}
{|R^-_{B, \overline{D}}| \over |R^-_B|} = {|R^-_{B} \cap \big(\bigcap_{r\leq |D|} R^-_{\overline{e_r}}\big)| \over |R^-_{B,\overline{D}}|} = {|R^-_{B} \setminus \big(\bigcup_{r\leq |D|} R^-_{{e_r}}\big)| \over |R^-_{B,\overline{D}}|} &\geq 1 - \sum_{r\in[|D|]} {|R^-_{e_r} \cap R^-_B| \over |R^-_B|}\\ 
&= 1- o\left({n\over d}\right) \cdot O\left({d\over n}\right) = 1- o(1).
\end{align*}
Hence, ${|R^-_B| \over |R^-_{B,\overline{D}}|} = O(1)$.
\end{proof}

Recall again the assumption that $\mathcal{A}$ does not make probes that do not reveal any new information about the instance. Now, we are ready to formally prove the following claim on the \adjacencyP~probes, that with $\q = o(n/d)$ probes, $\mathcal{A}$ is unlikely to obtain any positive answer.
\begin{lemma}\label{lem:no-edge}
Let $\mathcal{A}$ be an arbitrary deterministic algorithm interacting with process $\mathcal{P}^{s}$ {\em ($s\in \{+,-\}$)} and that has probed $o(n/d)$ times. The probability that $\mathcal{A}$ detects an edge with an \adjacencyP~probe {\em of the form $\langle u_t,v_t\rangle$} during the interaction is $o(1)$.    
\end{lemma}
\begin{proof}
Consider an arbitrary step $t$ in the interaction of $\mathcal{A}$ and $\mathcal{P}^s$ in which the algorithm performs an \adjacencyP~probe. Since, $t=o(n/d)$, by the description of $\mathcal{P}^s$ and applying Lemma~\ref{lem:edge-presence-density}, the probability that the answer to $q_t$ is not $\bot$ is bounded by:
\begin{align*}
{|R^-_{(u_t,v_t)} \cap R^-_{t-1}|\over |R^-_{t-1}|} \leq {\sum_{i_t,j_t\in[d]}|R^-_{(u_t,i_t,v_t,j_t)} \cap R^-_{t-1}|\over |R^-_{t-1}|} \leq d^2 \cdot O\left({1\over nd}\right) = O\left({d\over n}\right).
\end{align*}
Since the total number of probes is $o(d/n)$, by the union bound, the probability that $\mathcal{A}$ detects an edge with an \adjacencyP~probe during its interaction with $\mathcal{P}^s$ is $o(1)$.  
\end{proof}

Next, we similarly show that if $\q=o(\sqrt{n})$, $\mathcal{A}$ is likely to obtain a new vertex from every \neighborP~probe it performs.

\begin{lemma}\label{lem:repetitive-vertex}
Let $\mathcal{A}$ be an arbitrary deterministic algorithm interacting with process $\mathcal{P}^{s}$ {\em ($s\in \{+,-\}$)} and that has probed $o(\sqrt{n})$ times. With probability $1-o(1)$, all \neighborP~probes of $\mathcal{A}$ receive distinct vertices in their answers.    
\end{lemma}
\begin{proof}
Consider step $t$ in the interaction of $\mathcal{A}$ and $\mathcal{P}^s$ and let $V_{t-1}$ denote the set of vertices seen by $\mathcal{A}$ so far (i.e., participate in some $q_{t'}$ or $a_{t'}$ where $t' \leq t-1$); thus $|V_{t-1}| \leq 2t$. In what follows we bound the probability $p_t$ that $a_t$ (the answer to of the form $(u_t, i_t)$) corresponds to a vertex $v$ which belong to $V_{t-1}$.
\begin{align*}
p_t &= {|\bigcup_{v\in V_{t-1}, j\in [d]} (R^s_{(u_t,i_t,v,j)} \cap R^s_{t-1})|\over |R^s_{t-1}|} \\&\leq \sum_{v\in V_{t-1}, j\in [d]} {|(R^s_{(u_t,i_t,v,j)} \cap R^s_{t-1})|\over |R^s_{t-1}|} \leq 2t \cdot d\cdot O\left({1\over nd}\right) = O\left({1\over\sqrt{n}}\right).
\end{align*}
where the last inequality is implied by Lemma~\ref{lem:edge-presence-density}. Hence, if the total number of probes is $o(\sqrt{n})$, with probability $1-o(1)$ the answer to every \adjacencyP~probe introduces a new vertex
\end{proof}
 
Next, we prove the main result of this section. Lets $D^s_t$ denotes the distribution over the probe-answer histories of $t$ rounds of the interaction of $\mathcal{A}$ and $\mathcal{P}^s$.  
\begin{lemma}\label{lem:distance}
For any arbitrary deterministic $\mathcal{A}$ and $\q=o(\min\{\sqrt{n},n/d\})$, the statistical distance between $D^+_\q$ and $D^-_\q$ is $o(1)$. 
\end{lemma}
\begin{proof}
Let $\Pi$ be the set of all valid probe-answer histories of length $\q$ and let $\Pi' \subset \Pi$ denote the set of all histories in which every \adjacencyP~probe returns $\bot$ and no \neighborP~probe returns an already-discovered vertex.

Observe that conditioned on $\pi \in \Pi'$, the answers to all \adjacencyP~probes by both $\mathcal{P}^-$ and $\mathcal{P}^+$ are $\bot$. Moreover, the answers to each \neighborP~probe by both processes are chosen uniformly at random among all cells from the rows corresponding to the set of all vertices not visited so far, which is the same for the both processes. That is, $D^+_{\q}(\pi)$ and $D^-_\q(\pi)$ are proportional to each other for every probe-answer history $\pi \in \Pi$. Hence the difference between the probe-answer histories for $\pi \in \Pi'$ in both processes are bounded simply by the difference in their total probabilities:
\begin{align*}\sum_{\pi\in \Pi'} |D_\q^+(\pi) - D_\q^-(\pi)| = \left|\sum_{\pi\in \Pi'} D_\q^+(\pi) - \sum_{\pi\in \Pi'} D_\q^-(\pi)\right| = \left|\sum_{\pi\in \Pi\setminus\Pi'} D_\q^+(\pi) - \sum_{\pi\in \Pi\setminus\Pi'} D_\q^-(\pi)\right|.\end{align*}
Putting everything together, we bound the difference between the distributions of probe-answer histories when $\q = \min\{\sqrt{n},n/d\}$:
\begin{align*}
\sum_{\pi\in \Pi} |D_\q^+(\pi) - D_\q^-(\pi)| &= \sum_{\pi\in \Pi'} |D_\q^+(\pi) - D_\q^-(\pi)| +\sum_{\pi\in \Pi\setminus\Pi'} |D_\q^+(\pi) - D_\q^-(\pi)|\\ 
&\leq 2\left|\sum_{\pi\in \Pi\setminus\Pi'} D_\q^+(\pi) - \sum_{\pi\in \Pi\setminus\Pi'} D_\q^-(\pi)\right|\\
&\leq 2\left|\sum_{\pi\in \Pi\setminus\Pi'} D_\q^+(\pi)\right|+2\left|\sum_{\pi\in \Pi\setminus\Pi'} D_\q^-(\pi)\right| = o(1).
\end{align*}
where the last equation follows as a result of Lemma~\ref{lem:no-edge} and Lemma~\ref{lem:repetitive-vertex}. 
\end{proof}
Finally we turn to complete the proof of Thm. \ref{them:lowrbound} which holds also for simple graphs.
%
\begin{proof}[Proof of Theorem \ref{them:lowrbound}]
As outlined earlier, for the $1-o(1)$ fraction of the instances in $\mathcal{D}^+_{(x,a,y,b)}$, a deterministic $\mathcal{A}$ must keep the edge $(x,a,y,b)$ in its spanner with probability $\frac{2}{3}(1-o(1)) > 1/2$ because, due to Lemma~\ref{lem:distance}, with probability $1-o(1)$ it cannot distinguish whether the given instance is from $\mathcal{D}^+_{(x,a,y,b)}$ or $\mathcal{D}^-_{(x,a,y,b)}$. Since $\mathcal{D}_{(x,a,y,b)}^+$ is the uniform distribution over instances compatible with $(x,a,y,b)$, then for more than half of these instances, $\mathcal{A}$ returns \YES~on query $(x,a,y,b)$. Applying this argument for all $(x,a,y,b)$, we obtain that $\mathcal{A}$ returns \YES~on at least half of all compatible instance-query pairs. Nonetheless, our generated instances in $\mathcal{D}^+_{(x,a,y,b)}$, $\mathcal{D}^-_{(x,a,y,b)}$ often contains parallel edges and self-loops.

In order to remove these non-simple graphs, as similarly noted in \cite{kaufman2004tight}, we observe that our constructed graphs only have $O(d^2)$ parallel edges and $O(d)$ self-loops in expectation. Thus, we may simply fix each instance by modifying $O(d^2)$ matchings so that all instances become simple (assuming sufficiently large $n$ and $d$). Observe that by doing so, the connectivity of the graph strictly increases: the required condition of $\mathcal{D}^+_{(x,a,y,b)}$ that $x$ and $y$ must be connected even when $(x,a,y,b)$ is absent is still upheld. Similarly, for $\mathcal{D}^-_{(x,a,y,b)}$ the modifications must still respect the restriction that no edge other than $(x,a,y,b)$ has endpoints on different tables, so that removing $(x,a,y,b)$ disconnects them.

Due to similarly arguments as Lemma~\ref{lem:no-edge} and Lemma~\ref{lem:repetitive-vertex}, the probability that $\mathcal{A}$ detects these modifications are $o(1)$, and therefore Lemma~\ref{lem:distance} still holds under these changes, as long as the query to $\mathcal{A}$ itself is not one of the modified edges. On the other hand, if a modified edge is given as a query to $\mathcal{A}$, then we do not assume anything about the algorithm's answer for this edge. As the modified edges constitute a fraction of up to $O(d^2)/nd = O(d/n)$ of the total number of edges on each instance on average, the fraction of instance-query pairs where $\mathcal{A}$ answers \YES~can be potentially reduced by at most a fraction of $O(d/n)$: this still leaves a fraction of $\frac{2}{3}(1-o(1))-O(\frac{d}{n}) > 1/2$ for sufficiently small $d = O(n)$. That is, even when restricted to simple graphs, we still obtain that $\mathcal{A}$ returns \YES~on at least half of all compatible instance-query pairs

Over the uniform distribution over all instances, in expectation any deterministic algorithm $\mathcal{A}$ must return \YES~on more than $m/2$ edges. Employing Yao's principle, we conclude that any (randomized) LCA cannot compute a spanning subgraph with $o(m)$ edges with success probability $2/3$ using $o(\min\{\sqrt{n},n/d\})$ probes. Substituting $d = 2m/n$ yields the desired bound.
\end{proof}

\subparagraph*{Acknowledgements.} 
MP is supported by Minerva Foundation (124042) and ISF-2084/18.
RR is supported by the NSF grants CCF-1650733, CCF-1733808, IIS-1741137 and CCF-1740751.
AV is supported by the NSF grant CCF-1535851.
AY is supported by the NSF grants CCF-1650733, CCF-1733808, IIS-1741137 and the DPST scholarship, Royal Thai Government.

\end{document}